\documentclass[11pt,a4paper]{article}

\usepackage[english]{babel}
\usepackage{latexsym}
\usepackage[latin1]{inputenc}
\usepackage{amsfonts} 
\usepackage{amssymb}
\usepackage{amsmath}
\usepackage{amsthm}
\usepackage{enumerate}
\usepackage{exscale}
\usepackage{epsfig}
\usepackage{fullpage}
\usepackage{textcomp}
\usepackage{microtype}

\newtheorem{definition}{Definition}[section]
\newtheorem{lemma}{Lemma}[section]
\newtheorem{theorem}{Theorem}
\newtheorem{corollary}{Corollary}

\title{\bf On Alternation and the Union Theorem}
\author{Mathias Hauptmann\thanks{Dept. of Computer Science, University of Bonn.
    e-mail:{ \tt hauptman@cs.uni-bonn.de}}}
\date{}

\begin{document}
\maketitle
\begin{abstract}
Under the assumption P$=\Sigma_2^p$, we prove a new variant of the Union Theorem of McCreight and Meyer
for the class $\Sigma_2^p$. This yields a union function $F$ which is computable in time $F(n)^c$ for some constant $c$
and satisfies $\mbox{P}=\mbox{DTIME}(F)=\Sigma_2(F)=\Sigma_2^p$ with respect to a subfamily $(\tilde{S}_i)$ of $\Sigma_2$-machines.
We show that this subfamily does not change the complexity classes P and $\Sigma_2^p$. Moreover, a padding construction shows that
this also implies $\mbox{DTIME}(F^c)=\Sigma_2(F^c)$. 
On the other hand, we prove a variant of Gupta's result who showed that $\mbox{DTIME}(t)\subsetneq\Sigma_2(t)$ for time-constructible functions $t(n)$.
Our variant of this result holds with respect to the subfamily $(\tilde{S}_i)$ of $\Sigma_2$-machines. We show that these two results contradict each other.
Hence the assumption P$=\Sigma_2^p$ cannot hold.

\end{abstract}

\section{Introduction}
Alternating Turing Machines (ATM) have been introduced by Chandra and Stockmeyer \cite{CS76}. This model of computation generalizes both 
nondeterministic computations and co-non\-determinis\-tic computations. ATMs have been intensively used in the Complexity Theory (see e.g.
\cite{CS76,CKS81,PPR80,PR81,PR812}). 
In their seminal paper, Paul, Pippenger, Szemeredi and Trotter \cite{PPST83} 
gave a separation between deterministic linear time and nondeterministic linear time. They showed that $\mbox{DLIN}\neq\mbox{NLIN}$.
The proof is based on the result that deterministic
machines can be simulated faster on alternating machines with at most four alternations: 
\begin{theorem}\label{PPST_theorem}\cite{PPST83}\\
For every time-constructible function $t(n)$ with $t(n)\geq n\log^*(n)$,
$\mbox{DTIME}(t\log^*(t))\subseteq\Sigma_4(t))$.
\end{theorem}
Subsequent attempts did not succeed in generalizing the result $\mbox{DLIN}\neq\mbox{NLIN}$ to arbitrary polynomial time bounds.
No separation of $\mbox{DTIME} (n^k)$ from $\mbox{NTIME}(n^k)$ for any $k>1$ is known so far.
Kannan \cite{K81} gave a separation of nondeterministic time $n^k$ from deterministic time $n^k$ with $o(n^k)$ space. He showed that
there exists some constant $c$ such that for all $k$,
\[\mbox{NTIME}\left (n^k\right )\not\subseteq\: \mbox{DTIME-SPACE}\left (n^{k},n^{k\slash c}\right ).\]
Gupta \cite{G96} was able to reduce the number of alternations in Theorem \ref{PPST_theorem} and obtained the following result.
\begin{theorem}\cite{G96}\\
For every time-constructible function $t(n)$ with $t(n)\geq n\log^*n$, $\mbox{DTIME}(t\log^*(t))\:\subseteq\:\Sigma_2(t)$.
\end{theorem}
Combining this result with the diagonalization power of $\Pi_2(t(n))$-machines over $\Sigma_2(t(n)\slash\sqrt{\log^*(n)})$-machines,
he obtained the following separation between deterministic and $\Sigma_2$-compu\-tations for time-constructible time-bounds.
\begin{theorem}\label{gupta-theorem}\label{gupta_thm}\cite{G96}\\
For every time-constructible function $t(n)$ with $t(n)\geq n\log^*n$, $\mbox{DTIME}(t)\neq\Sigma_2(t)$.
\end{theorem}
In Section \ref{gupta_section} we shall briefly discuss this result.
Santhanam \cite{S01} extended the techniques from \cite{PPST83} and showed that 
\begin{equation}\label{santhanam_ineq}
\mbox{DTIME}\left (n\sqrt{\log^*(n)}\right )\:\subsetneq\: \mbox{NTIME}\left (n\sqrt{\log^*(n)}\right ).
\end{equation}
Furthermore, he showed that at least one of the following two statements must hold: (1) $\mbox{DTIME}(t(n))\neq\mbox{NTIME}(t(n))$ for
all polynomially bounded constructible time bounds $t(n)$, (2) $P\neq \mbox{LOGSPACE}$. However, as Santhanam already stated in \cite{S01}, 
it is not known if the inclusion $\mbox{DTIME}(n)\subseteq\mbox{DTIME}(n\sqrt{\log^*(n)})$ is strict. So it is open if (\ref{santhanam_ineq}) is really a new 
inequality between complexity classes.

In 1969, McCreight and Meyer proved the \emph{Union Theorem} in \cite{McCM69}. This theorem states that for every Blum complexity measure (cf. \cite{B67})
and every 
sufficiently bounded enumerable family of recursive functions $(f_i)$, the union of the complexity classes given by the functions $f_i$ is equal to the 
complexity class of a single recursive function $f$. Applied to the deterministic time complexity measure and the polynomial functions $f_i(n)=n^i$, the 
Union Theorem states that there exists a recursive function $f(n)$ such that $P=\bigcup_i\mbox{DTIME}(n^i)=\mbox{DTIME}(f(n))$. In particular, this function $f(n)$ has the 
following property for every deterministic machine $M$: If $\mbox{time}_M(n)$ denotes the running time of machine $M$ on input length $n$, then there exists some 
polynomial function $f_i(n)$ such that $\mbox{time}_M(n)\leq f_i(n)$ for all $n$ if and only if $\mbox{time}_M(n)=O(f(n))$.\\[0.6ex]
\emph{{\bf Our Contribution.}} We assume that $P=\Sigma_{2}^p$. 
The general idea is now to construct a union function $F(n)$ such that $P=\mbox{DTIME}(F)=\Sigma_2(F)=\Sigma_2^p$ and to obtain a contradiction to 
Theorem \ref{gupta-theorem}. However, if we construct the function $F$ directly as in \cite{McCM69}, then the function $F(n)$ will not be time-constructible.
Taking a close look at the proof of the Union Theorem in \cite{McCM69} and taking our assumption into account, the second idea is to construct $F$ in such a way
that $F(n)$ is computable in time $F(n)^C$, for some constant $C$. Then we might want to use a padding construction to show that 
$\mbox{DTIME}(F)=\Sigma_2(F)$ also implies $\mbox{DTIME}(F^C)=\Sigma_2(F^C)$, and now the function $F^C$ is actually time-constructible.
However, now another problem occurs. For this implication to hold via padding, it is necessary that $F(n)^C$ can be bounded by a function value $F(q(n))$, for some
fixed polynomial $q(n)$. As we shall see below in Section \ref{outline_section}, in a straight forward union construction the function $F$ fails to have this
{\sl padding property}. Our solution will be to switch from a standard family of $\Sigma_2$-machines to a subfamily. This subfamily will contain for each 
$\Sigma_2$-machine $S_i$ and every integer $d$ a machine $\tilde{S}_{i,d}$. We construct this subfamily in such a way that the machines $\tilde{S}_{i,d}$
have the following property: Whenever the running time of the machine $\tilde{S}_{i,d}$ at input length $n$ exceeds a bound of the form $an^b$, then there
already exist sufficiently many smaller input lengths $m<n$ at which the running time of $\tilde{S}_{i,d}$ exceeds a similar but weaker bound of the form 
$(a-i\slash 2)m^{b-i\slash 2}$. The union function will then be constructed in such a way that precisely this property of the machines enforces $F$ to satisfy 
the padding inequality. Moreower, we have to show that switching from the family of all $\Sigma_2$-machines to this subfamily does not change the 
complexity classes. Namely, we show that for each $L\in \Sigma_2^p$ there also exists a polynomial time machine in this subfamily accepting $L$, and for 
every $L\in P$ there exists such a machine which is deterministic. Finally we show that the result from Theorem \ref{gupta-theorem} also holds with respect
to this restricted subfamily of machines.  

Let us now describe this in a slightly more detailed way.
We let $(S_i)$ be a standard enumeration of $\Sigma_2$-machines.
First we show that under the assumption $P=\Sigma_{2}^p$, the problem of deciding for a given tuple $i,x,a,b$ consisting of a machine index $i$, an input string $x$ of length $n$
and two integers
$a,b$ if $\mbox{time}_{S_i}(n)>an^b$ can be solved deterministically in time $c\cdot (i\cdot a\cdot n^b)^c$, where $c$ is a constant (not depending on $i,x,a,b$).
The proof consists of a standard padding construction.
Then we construct a new family $(\tilde{S}_{i,d})$ of $\Sigma_2$-machines, which contains for every machine $S_i$ and every integer $d$ a machine $\tilde{S}_{i,d}$ and
has the following properties:
\begin{itemize}
\item Whenever $S_i$ is a deterministic machine, then for all $d$, the machine $\tilde{S}_{i,d}$ is also deterministic.
\item The running time functions $\mbox{time}_{\tilde{S}_{i,d}}(n)$ of the machines $\tilde{S}_{i,d}$ satisfy the following kind of weak monotonicity condition: 
      For every pair $i,d$ and every input length $n$ with $n$ being sufficiently large, there exists an interval $I_{n,d}$ of integers of the form 
      $I_{n,d}=(n^{1\slash h},n^{1\slash h}\cdot (1+o(1)))$ such that whenever $\mbox{time}_{\tilde{S}_{i,d}}(n)>an^b$ with $a\leq b=O(\log (n))$, 
      then there exist $\Omega(\log\log n)$
      pairwise distinct integers 
      $m$ within the interval $I_{n,d}$ such that $\mbox{time}_{\tilde{S}_{i,d}}(m)>(a-\frac{i}{2})m^{b-\frac{i}{2}}$. 
      Here $h$ is an integer number which only depends on $c$
      (and not on $i,d,n$). We call this weak monotonicity condition \emph{Property $[\star]$}.
\end{itemize}
We let $\tilde{\Sigma}_2(t)$ and $\mbox{DTIM}\tilde{\mbox{E}}(t)$ denote the time complexity classes with respect to this new family $(\tilde{S}_{i,d})$. In particular,
$\tilde{\Sigma}_{2}^p=\bigcup_{p(n)}\tilde{\Sigma}_2(p(n))$ and $\tilde{P}=\bigcup_{p(n)}\mbox{DTIM}\tilde{\mbox{E}}(p(n))$, where the union goes over all polynomials $p(n)$.
We will show that $P=\tilde{P}=\tilde{\Sigma}_2^p=\Sigma_2^p$.\\[0.3ex]
Then we construct a union function $F$ for $\tilde{\Sigma}_2^p$ with respect to the family $(\tilde{S}_{i,d})$. This means that
\[\tilde{P}\: =\: \mbox{DTIM}\tilde{\mbox{E}}(F)\: =\: \tilde{\Sigma}_{2}(F)\: =\:\tilde{\Sigma}_2^p.\]
We show that the function $F(n)$ can be computed deterministically in time $F(n)^C$ for some constant $C<h$. Since 
we can compute the value $F(n)^{C^2}$ from $F(n)$ in time $\log(C^2)\cdot(\log (F(n)^{C^2}))^3$, this particularly yields that the function 
$t(n):=F(n)^{C^2}$ can be computed in time
$t(n)^{1-\epsilon}$ for some $\epsilon >0$. 
We show that the function $F(n)^{C^2}$ also satisfies the \emph{Property $[\star]$}. Moreover, we show that the function $F$ satisfies the following 
inequality:
\begin{equation}\label{zzz_i}
\begin{array}{l}
F(n^{1\slash h})^C\:\leq\: F(n)\:\:\mbox{for all $n$ such that $n^{1\slash h}$ is an integer,}\\[0.6ex]
\mbox{or equivalently:}\:\:F(n)^C\:\:\:\leq\: F(n^h)\:\:\mbox{for all $n$.}
\end{array}
\end{equation}
In the proof of this inequality we will extensively make use of the fact that the machines $\tilde{S}_{i,d}$ satisfy the weak monotonicity \emph{Property $[\star]$}, and also of how the 
union function is constructed.
This inequality will then enable us to apply \emph{Padding} in order to show that 
\begin{equation}\label{zzz_ii}
\mbox{DTIM}\tilde{\mbox{E}}(F)=\tilde{\Sigma}_{2}(F)\:\:\mbox{also implies}\:\: \mbox{DTIM}\tilde{\mbox{E}}(F^{C^2})=\tilde{\Sigma}_{2}(F^{C^2}).
\end{equation}
On the other hand, we will extend Gupta's result \cite{G96} to time classes $\tilde{\Sigma}_2(t)$. Namely, we show that for every function $t(n)\geq n\log^* n$ which is deterministically computable in time
$t(n)^{1-\epsilon}$ for some $\epsilon >0$ and satisfies \emph{Property $[\star ]$}, 
\begin{equation}\label{zzz_iii}
\mbox{DTIM}\tilde{\mbox{E}}(t)\:\subsetneq\:\tilde{\Sigma}_2(t).
\end{equation}
Now (\ref{zzz_ii}) and (\ref{zzz_iii}) contradict each other. Therefore, the assumption $P=\Sigma_2^p$ cannot hold.
%

Baker, Gill and Solovay \cite{BGS75} have shown that the $P$ versus $NP$ problem cannot be settled by relativizing proof techniques. Our results presented in this 
paper rely on the results from Paul et al. \cite{PPST83} and Gupta \cite{G96}. Gasarch \cite{G87} has shown that these results do not relativize
(cf. also Allender \cite{A90}, Hartmanis et al. \cite{H92}, Fortnow \cite{F94}). Natural proofs have been introduced by Razborov, Rudich \cite{RR94}. 
Since we do not prove any circuit lower bound, our results do not contradict the Natural Proof barrier from \cite{RR94}. Moreover, our results do not contradict 
the Algebrization barrier given by Aaronson and Wigderson \cite{AW09}. This is already stated in Section 10 of \cite{AW09}, where the DLIN versus NLIN result from \cite{PPST83}
is listed among examples of results to which their framework does not apply. 

Section \ref{outline_section} contains an extended outline of our constructions. It concludes with a roadmap of our 
constructions and results. Preliminaries are given in Section \ref{prelim_section}.
In Section \ref{implication_section}, we show that under the assumption $P=\Sigma_2^p$, we can test deterministically in time 
$c\cdot (i\cdot p(n))^c$ if the running time of machine $S_i$ on input length $n$ exceeds a given polynomial bound $p(n)$. The constant $c$ does not depend on $p(n)$ or the machine index $i$.
In Section \ref{prog_section}, we construct our new family $(\tilde{S}_{i,d})$ of $\Sigma_2$-machines.
In Section \ref{constr_section} we construct the union function $F$, and in Section \ref{gupta_section} we prove a variant of Gupta's Theorem \ref{gupta-theorem} 
for time classes $\mbox{DTIM}\tilde{\mbox{E}}(t)$ and $\tilde{\Sigma}_2(t)$. This will then give the desired contradiction,
and hence the assumption $P=\Sigma_2^p$ cannot hold.

\section{Motivation and Outline of our Construction}\label{outline_section}
We assume that $P=\Sigma_2^p$. As a direct application of the results and techniques from 
McCreight and Meyer \cite{McCM69}, we can construct a computable function $F$ such that
\begin{equation}\label{eq1}
P\: =\: \mbox{DTIME}(F)\: =\: \Sigma_2(F)\: =\: \Sigma_2^p.
\end{equation}
Here the first and the last equality hold since for every $\Sigma_2$-machine $M$, $\mbox{time}_M(n)$ is polynomially bounded iff $\mbox{time}_M(n)=O(F(n))$, so this 
holds especially for every deterministic machine. The second equality holds due to the assumption $P=\Sigma_2^p$.
Moreover, the function $F$ can be constructed in such a way that it can be computed deterministically in time $F(n)^C$ for some 
constant $C>1$. Thus the function $F^C$ is \emph{time-constructible}, i.e. $F(n)^C$ can be computed deterministically in time $F(n)^C$. 
The idea is now to show that (\ref{eq1}) also implies 
\begin{equation}\label{eq2}
\mbox{DTIME}(F^C)\: =\: \Sigma_2(F^C). 
\end{equation}
This would then directly contradict Theorem \ref{gupta-theorem}, hence the assumption $P=\Sigma_2^p$ cannot hold true.

A standard approach in order to show that (\ref{eq1}) implies (\ref{eq2}) is to make use of \emph{Padding}: 
For each decision problem $L\in \Sigma_2(F^C)$ we construct an associated problem $L'=\{x10^{k-|x|-1}\;|\;x\in L\}$, 
where $k$ is polynomially bounded in the input length 
$|x|$, say for simplicity $k=|x|^h$ for some constant $h$. $L'$ is called a polynomially padded version of $L$. 
Now suppose that we manage to choose $h$ in such a way 
that the following inequality holds:
\begin{equation}\label{eq3}
F(n)^C\:\leq\: F(n^h).
\end{equation} 
Then this directly implies that $L'\in\Sigma_2(F)=\Sigma_2^p$: 
On a given input $y$, we first check in linear time if $y$ is of the form $y=x10^{|x|^h-|x|-1}$
(otherwise we reject). Then we run the $\Sigma_2(F^C)$-algorithm for $L$ on input $x$ and return the result.
Due to inequality (\ref{eq3}), the running time of this $\Sigma_2$-algorithm for $L'$ on an input $y$ of length $m=n^h$ is bounded by 
$O(n^h+F(n)^C)=O(F(n^h))=O(F(m))$, and thus we have $L'\in\Sigma_2(F)$. 
By (\ref{eq1}) we obtain $L'\in\Sigma_2(F)=\mbox{DTIME}(F)=P$.
This implies $L\in P=\mbox{DTIME}(F)\subseteq\mbox{DTIME}(F^C)$. Thus we have shown that (\ref{eq1}) implies (\ref{eq2}). 

%
%

Unfortunately, the approach does not work in this way.
The following problem occurs. If we take a standard indexing $(S_i)$ of $\Sigma_2$-machines and construct a Union Function $F$ as in 
\cite{McCM69} such that for every $\Sigma_2$-machine $S_i$, $\mbox{time}_{S_i}(n)=n^{O(1)}$ iff $\mbox{time}_{S_i}(n)=O(F(n))$, 
then $F$ may not satisfy inequality (\ref{eq3}). In order to see why this is so, we have to take a closer look at the construction of
the union function in \cite{McCM69}. \\
$\mbox{ }$\\
\emph{McCreight and Meyer's Union Function.}
Let us briefly describe the construction of the union function. 
We only consider the case of $\Sigma_2$-machines and of polynomial time bounds, which means we describe the 
construction of the union function for the class $\Sigma_2^p$.
In \cite{McCM69}, the union function is constructed in stages. In stage $n$,
the function value $F(n)$ is determined. During the construction, a list ${\mathcal L}$ of guesses is maintained. A guess is here a pair $(S_i,b_i)$ consisting of a 
$\Sigma_2$-machine $S_i$ and a number $b_i$ which corresponds to the polynomial $b_i\cdot n^{b_i}$. 
Such a guess is \emph{satisfied at stage $n$} if $\mbox{time}_{S_i}(n)\leq b_i\cdot n^{b_i}$, otherwise
the guess is \emph{violated at stage $n$}. Let ${\mathcal L}_n$ denote the list of guesses at the beginning of stage $n$ of the construction. 
The construction starts with ${\mathcal L}_1=\{(S_1,1)\}$. For every $n$, the list ${\mathcal L}_n$ contains $n$ guesses, namely one for each of the first $n$ machines
$S_1,\ldots , S_n$. In stage $n$, if there are guesses in ${\mathcal L}_n$ which are violated at
stage $n$, the lexicographically first such guess, say $(S_i,b_i)$ is \emph{selected}, where lexicographically means that guesses are first ordered by 
increasing value $b_i$ and then by the machine index $i$. Then the function value is defined as $F(n):= n^{b_i}$. The guess $(S_i,b_i)$ is replaced by 
$(S_i,b_i+1)$. If none of the guesses in ${\mathcal L}_n$ is violated at stage $n$, then the function value is defined as $F(n)\colon =n^n$. Finally, at the end of stage 
$n$ a new guess $(S_{n+1},n+1)$ enters the list.\\[0.3ex]
Now one can show that this function $F(n)$ has the following two properties:
\begin{itemize}
\item[1.] For every machine $S_i$ whose running time is polynomially bounded, 
          we have $\mbox{time}_{S_i}(n)=O(F(n))$ which means that there exists a constant $a$ such that
          $\mbox{time}_{S_i}(n)\leq a\cdot F(n)$ for almost all $n$.
\item[2.] For every machine $S_i$ whose running time is not polynomially bounded, 
          we have that for every constant $a$, $\mbox{time}_{S_i}(n)>a\cdot F(n)$ for infinitely many $n$. 
\end{itemize}
Properties 1 and 2 yield that $\Sigma_2^p=\Sigma_2(F)$ and also $P=\mbox{DTIME}(F)$.\\[0.2ex] 
Now the reason why inequality (\ref{eq3}) may not hold becomes clear: It might happen that a guess $(S_i,b_i)$ is satisfied for a very long time, while the 
union function $F(n)$ is getting larger and larger. Then, eventually at some stage $n$, the running time of $S_i$ on input length $n$ exceeds $b_i\cdot n^{b_i}$,
and the guess $(S_i,b_i)$ is selected in stage $n$ of the construction of the union function. This will cause $F(n)$ to drop down to $n^{b_i}$, and thus 
it might be that  $F(n)\ll F(n^{1\slash 2}),F(n^{1\slash 3})$ and so on. Especially it might be that there does not exist any $h$ such that 
the inequality $F(n^{1\slash h})^C\leq F(n)$ holds for all $n$.\\[0.9ex]
%
Our first idea how to circumvent this obstacle is as follows. We modify the construction of the union function.
Instead of the family $(S_i)$ of $\Sigma_2$-machines, we want to work with a restricted subfamily $(\tilde{S}_i)$ 
such that each $\Sigma_2$-machine from this subfamily has
the following property: Whenever $\tilde{S}_i$ violates a guess $(\tilde{S}_i,a)$ at some stage (i.e. input length) $n$ such that $n$ is sufficiently large, then 
the guess $(\tilde{S}_i,a)$ is already violated at a sufficient number of input lengths $m$ within an interval of the form $(n^{1\slash h},n^{1\slash h}\cdot (1+o(1)))$.
Here $h$ is a \emph{global constant} of the construction, i.e. $h$ does not depend on the machine index $i$.

Let us describe why this property is useful for our purpose. We want to achieve that our union function $F(n)$ is computable in time $F(n)^C$ and satisfies
the inequality $F(n^{1\slash h})^C\leq F(n)$ for all $n$ for which $n^{1\slash h}$ is also integer (in that case we call $n$ an $h$-power). 
Suppose that for some $n$, this inequality does not hold. Say at stage $n^{1\slash h}$, the guess $(\tilde{S}_i,b_i)$ is violated and selected by $F$, and at
stage $n$ the guess $(\tilde{S}_j,b_j)$ is violated and selected by $F$. Thus we have
\begin{equation}\label{nb_ineq1}
F(n^{1\slash h})^C=n^{b_i\cdot C\slash h}\: >\: n^{b_j}=F(n).
\end{equation}
Now if we have $C<h$, then the inequality (\ref{nb_ineq1}) yields $b_i>b_j$. So suppose that we construct our union function in such a way that both new guesses 
entering the list ${\mathcal L}$ and guesses which result from a selection and replacement have the property that their $b$-value is always greater or equal to the largest $b$-value in the list ${\mathcal L}$ so far. Then 
this directly implies that the guess $(\tilde{S}_j,b_j)$ must have been already in the list of guesses ${\mathcal L}_{n^{1\slash h}}$ in stage $n^{1\slash h}$. 
Now if this guess is violated in stage $n$
and if we can assure that the number of stages $m$ within the interval $(n^{1\slash h},n^{1\slash h}(1+o(1)))$ in which the guess is also violated
is larger than the number of guesses in the list ${\mathcal L}$
within this interval, then this yields a contradiction: Since the machine $\tilde{S}_j$ satisfies the above property
and since the guess $(\tilde{S}_j,b_j)$ is violated at input length $n$, it is already violated at a sufficient number of input lengths $m$ with $n^{1\slash h}<m<n$.
If the number of these violations is larger than the number of guesses in the list ${\mathcal L}$, then eventually the guess $(\tilde{S}_j,b_j)$ will be violated and have the highest
priority of being selected in the construction of $F$.
Therefore it would have been already selected in a stage within that interval, and then replaced by a guess $(\tilde{S}_j,b_j')$ with $b_j'>b_j$. Thus the guess $(\tilde{S}_j,b_j)$
cannot be 
contained in the list ${\mathcal L}_n$ at stage $n$ anymore, a contradiction.

Thus we obtain the following approach. We want to construct a subfamily $(\tilde{S}_i)$ of $\Sigma_2$-machines with the following properties:
\begin{itemize}
\item Each machine $\tilde{S}_i$ has the property which we described above: Whenever a guess $(\tilde{S}_i,b_i)$ is violated at stage $n$ and $n$ is sufficiently large, then
      there exist sufficiently many input lengths $m$ within the interval $(n^{1\slash h},n^{1\slash h}\cdot (1+o(1)))$ such that the guess $(\tilde{S}_i,b_i)$
      is violated at stage $m$.  
\item This subfamily still defines the same classes $P$ and $\Sigma_2^p$. Namely, for each $L\in \Sigma_2^p$, there exists a polynomially time bounded machine
      $\tilde{S}_i$ in the subfamily such that $L=L(\tilde{S}_i)$, and for each $L\in P$ there exists a polynomially time bounded deterministic machine 
      $\tilde{S}_j$ in the subfamily such that $L=L(\tilde{S}_j)$.
\end{itemize}
Then we want to construct a union function $F$ for $P=\Sigma_2^p$ with respect to this subfamily. This means that
for each machine $\tilde{S}_i$, the running time of $\tilde{S}_i$ is polynomially bounded if and only if the running time is in $O(F(n))$.
As in the original construction of McCreight and Meyer, during the construction of the function $F$ we maintain a list ${\mathcal L}$ of guesses $(\tilde{S}_i,b_i)$. 
$F$ will be constructed in stages. In stage $n$, the function value $F(n)$ is determined.
As before, we let ${\mathcal L}_n$ denote this list of guesses at the beginning of stage $n$ of the construction.
In each stage $n$, we select from this list 
a lexicographically smallest violated guess $(\tilde{S}_i,b_i)$ (first ordered by $b_i$ and then by the index $i$) and define $F(n)=n^{b_i}$. Then this guess 
$(\tilde{S}_i,b_i)$ is replaced by $(\tilde{S}_i,b_n^*)$, where $b_n^*$ is the maximum of all values $b_j$ of guesses $(\tilde{S}_j,b_j)$ in the list ${\mathcal L}_n$.
So when a guess is selected and replaced, its new $b$-value is at least as large as all the other $b$-values of guesses in the list ${\mathcal L}_n$. 
Finally, we want to keep the list ${\mathcal L}$ sufficiently small such that the above argument works. Namely if a guess $(\tilde{S}_i,b_i)$ is contained in the list 
${\mathcal L}_n$ and in the list ${\mathcal L}_{n^{1\slash h}}$ and is violated at stage $n$, then the size of the list ${\mathcal L}_m$ in stages 
$m\in (n^{1\slash h},n^{1\slash h}\cdot (1+o(1)))$ must be smaller than the number of stages $m$ within this interval at which the guess
$(\tilde{S}_i,b_i)$ is also violated. As we described above, this will imply that the inequality $F(n^{1\slash h})^C\leq F(n)$ holds.
In our construction, the list ${\mathcal L}_n$ will be of size $\log^*n$.

We let $\mbox{DTIM}\tilde{\mbox{E}}(t)$ and $\tilde{\Sigma}_2(t)$ denote the time complexity classes with respect to the family $(\tilde{S}_i)$ of machines.
The union function $F$ has the property that $\mbox{DTIM}\tilde{\mbox{E}}(F)=\tilde{\Sigma}_2(F)$.
Then we use a padding construction to show that this also implies $\mbox{DTIM}\tilde{\mbox{E}}(F^C)=\tilde{\Sigma}_2(F^C)$.
Since $F(n)$ is computable in time $F(n)^C$, the function $F(n)^C$ is time-constructible.
The padding construction works since $F(n)$ satisfies the inequality $F(n^{1\slash h})^C\leq F(n)$. 

Finally we want to achieve that the following variant of Gupta's separation result holds for the subfamily $(\tilde{S}_i)$: 
For functions
$t(n)\geq n\log^*n$ which are computable in time $O(t(n))$ by some machine $\tilde{S}_i$, the deterministic class $\mbox{DTIM}\tilde{\mbox{E}}(t)$ is 
strictly contained in $\tilde{\Sigma}_2(t)$.\\[1.3ex]
\emph{First Attempt.} 
We describe now our first attempt how to construct the subfamily $(\tilde{S}_i)$ of $\Sigma_2$-machines and the union function $F$. 
We start by giving a preliminary definition of the property the machines $\tilde{S}_i$ are supposed to have. Since this property will play a central role in our construction, we 
give a name to it and call it Property $[\star]$.\\[0.5ex]
\emph{Property $[\star ]$ (first definition)}\\
        We say a function $g(n)$ satisfies Property $[\star]$ with parameter $c_g$ if for every $n\geq c_g$ and every pair of integers $a,b$ with 
        $c_g\leq a\leq b\leq \log (n)\slash c_g$, the following holds: If $g(n)>an^b$, then there exist at least $\frac{\log\log (n)}{c_g}$ 
        distinct integers $m_i$ within the interval $(n^{1\slash h},n^{1\slash h}\cdot (1+\log (n)\slash n^{1\slash h}))$ such that 
        $g(m_i)\geq am_i^b, i=1,\ldots , \frac{\log\log (n)}{c_g}$.\\[2ex]
Now we start from the standard indexing $(S_i)$ of $\Sigma_2$-machines. We construct for each machine index $i$ a new $\Sigma_2$-machine $\tilde{S}_i$ 
such that the function $\mbox{time}_{\tilde{S}_i}(n)$ satisfies the  
Property $[\star ]$ with some parameter $c_i$. Moreover, if the running time function $\mbox{time}_{S_i}(n)$ of machine $S_i$ already satisfies
Property $[\star ]$, then $L(\tilde{S}_i)=L(S_i)$ and the running time 
of $\tilde{S}_i$ equals the running time of $S_i$. In this way, we obtain the subfamily $(\tilde{S}_i)$. 
Now we want proceed as follows: 
\begin{itemize}
\item We show that $P=\tilde{P}$ and $\tilde{\Sigma}_2^p=\Sigma_2^p$. Since we assume $P=\Sigma_2^p$, this will especially yield $\tilde{P}=\tilde{\Sigma}_2^p$.
\item We construct a union function $F$ for $\tilde{\Sigma}_2^p$. This means that for every index $i$,  
      $\mbox{time}_{\tilde{S}_i}(n)$ 
      is polynomially bounded if and only if $\mbox{time}_{\tilde{S}_i}(n)=O(F(n))$. 
      Hence we have $\tilde{\Sigma}_2(F)=\tilde{\Sigma}_2^p$ and $\mbox{DTIM}\tilde{\mbox{E}}(F)=\tilde{P}$.
\item We show that $F(n)$ can be computed deterministically in time $F(n)^C$ for some constant $C<h$. Note that $h$ is the constant from the definition of 
      \emph{Property $[\star ]$}. 
\item We show that the union function $F$ satisfies the inequality $F(n^{1\slash h})^C\leq F(n)$ for all $h$-powers $n$. 
      This will allow us to make use of a padding construction in order to show that
      $\mbox{DTIM}\tilde{\mbox{E}}(F^{C^2})=\tilde{\Sigma}_2(F^{C^2})$. Furthermore we show that the function
      $F^{C^2}(n)$ also satisfies Property $[\star ]$.
\item We extend the result of Gupta \cite{G96} and show that for each function $t(n)\geq n\cdot\log^*(n)$ which can be computed in time $t(n)^{1-\epsilon}$ for some
      $\epsilon >0$ and satisfies \emph{Propery $[\star ]$},
      $\mbox{DTIM}\tilde{\mbox{E}}(t)\subsetneq\tilde{\Sigma}_2(t)$.
\end{itemize} 
The last two items contradict each other: Since $F(n)$ is deterministically computable in time $F(n)^C$ and satisfies $F(n)\geq n\log^*(n)$, we obtain that the function 
$t(n):=F(n)^{C^2}$ can be computed deterministically in time $F(n)^C=t(n)^{1-\epsilon}$ for $\epsilon=1-C^{-1} >0$. 
Thus we conclude that the assumption $P=\Sigma_2^p$ cannot hold. \\[1ex]
Now we will give a first outline of the construction of the union function $F$. Afterwards we describe why this first attempt does not yet work and
we have to modify the construction and also the 
definition of
Property $[\star ]$.\\[1ex]
\emph{Construction of the Union Function $F$ (preliminary version).}\\
We construct our new union function similarly to the one in \cite{McCM69}. Recall that we want to keep the size of the list of guesses at stage $n$ of order
$\log^*n$. 
This means that whenever the $\log^*$ function increases by one, we will add a new guess to the list. 

More precisely, we make use of the following version of the $\log^*$-function:
\[\log^*(n)\: :=\: \min\{t\:|\:2^{2^{\ldots 2}}|t\:\geq n\}\]
Here by $2^{2^{\ldots 2}}|t$ we denote a tower $T_2(t)$ of height $t$, i.e. $T_2(1)=2,T_2(2)=2^2,T_2(3)=2^{2^2}$ and so on.
We consider the intervals $I_t$ on which this function is equal to $t$. Namely,
\[I_t=[\beta_{t-1}+1,\beta_t]\:\:\mbox{with $\beta_0=0,\beta_1=2,\beta_{t+1}=2^{\beta_t},t\geq 1$}.\]
$F$ is constructed in stages. In stage $n$, the function value $F(n)$ is defined. We maintain a list of guesses ${\mathcal L}$, where a guess is now a pair 
$(\tilde{S}_i,b_i)$ consisting of a $\Sigma_2$-machine $\tilde{S}_i$ and an integer number $b_i$. As before, a guess $(\tilde{S}_i,b_i)$ is called 
\emph{safisfied at stage $n$} if $\mbox{time}_{\tilde{S}_i}(n)\leq b_i\cdot n^{b_i}$. Otherwise, the guess is called \emph{violated at stage $n$}.
We denote by ${\mathcal L}_n$ the list of guesses at the beginning of stage $n$ of the construction. Within our construction, we maintain the following invariants:
\begin{itemize}
\item For all $t\in {\mathbb N}$ and for every stage $n\in I_t$, the number of guesses in the list ${\mathcal L}_n$ is equal to $t=\log^*n$.
      Furthermore, the maximum $b_i$ value that occurs in the list ${\mathcal L}_n$ is also equal to $t$:
      \[\max\{b_i|\: (\tilde{S}_i,b_i)\in {\mathcal L}_n\}=\log^*n=t\]
\item When a guess $(\tilde{S}_i,b_i)$ is violated and selected in stage $n$ of the construction, this implies that $F(n)=n^{b_i}$. Furthermore, this guess is 
      then replaced by $(S_i,\log^*n)$ in the list of guesses ${\mathcal L}$.
\item The maximum value which is attained by $F$ within stages $n$ in the interval $I_t$ is equal to $n^t$:
      \[\mbox{For all $n\in I_t$, $F(n)\leq n^t=n^{\log^*(n)}$.}\]
      Furthermore, in almost all stages $n\in I_t$ we will actually have $F(n)=n^t$, namely 
      in at least $|I_t|-t$ stages $n$ in $I_t$. 
\end{itemize}
Now we can already see why we have to modify the definition of the Property $[\star ]$. Recall that we want to achieve that the function $F^{C^2}(n)$ also satisfies 
Property $[\star ]$. But if $F(n)\leq n^{\log^*(n)}$, then there might be stages $n\in I_t$ such that $F(n)=n^t$ and such that the interval
$(n^{1\slash h},n^{1\slash h}\cdot (1+o(1)))$ is contained in the previous interval $I_{t-1}$. 
Thus, there won't be any stages $m$ within the interval $(n^{1\slash h},n^{1\slash h}\cdot (1+o(1)))$
such that $F(m)=m^t$, and the preliminary version of Property $[\star ]$ which we have described above cannot hold.

A first idea is to proceed as follows: We modify the Property $[\star]$ such that in case when $\mbox{time}_{\tilde{S}_i}(n)>a\cdot n^b$,
we require that there exist sufficiently many stages $m$ within the interval $(n^{1\slash h},(1+o(1))\cdot n^{1\slash h})$ such that
$\mbox{time}_{\tilde{S}_i}(m)> (a-1)\cdot m^{b-1}$. However, this does still not work. The reason is that
we want to use padding in order to show that  $\mbox{DTIM}\tilde{\mbox{E}}(F)=\tilde{\Sigma}_2(F)$ implies 
$\mbox{DTIM}\tilde{\mbox{E}}(F^{C^2})=\tilde{\Sigma}_2(F^{C^2})$. In that situation, we start from some $L\in \tilde{\Sigma}_2(F^{C^2})$ 
and construct an associated padded version $L'$.
Then we have to show that $L'\in\tilde{\Sigma}_2(F)$. As we have already described above, we can construct a $\Sigma_2$-machine $S'$ for $L'$ such that
$\mbox{time}_{S'}(n)=O(F(n))$. But now we also have to assure that the function $\mbox{time}_{S'}(n)$ satisfies Property $[\star]$. For this purpose,
we want to make use of the fact that $L=L(S)$ for some $\Sigma_2$-machine $S$ for which $\mbox{time}_S(n)$ satisfies Property $[\star ]$.
Essentially this means that the Property $[\star]$ needs to be preserved under polynomial padding. The modification above does not preserve Property $[\star ]$
under polynomial padding. Namely, if $t(n)$ is a function which satisfies the above version of the Property $[\star]$, then we cannot conclude that powers $t(n)^{\gamma}$
of the function $t(n)$ also satisfy Property $[\star ]$. This can be seen as follows. If $t(n)^{\gamma}>an^b$, this means that $t(n)>a^{1\slash\gamma}n^{b\slash\gamma}$.
Since the function $t(n)$ satisfies Property $[\star ]$, this implies that there exist sufficiently many smaller integers $m$ such that $t(m)>(a^{1\slash\gamma}-1)m^{b\slash\gamma -1}$.
But this only implies that $t(m)^{\gamma}>(a^{1\slash\gamma}-1)^{\gamma}m^{b-\gamma}$, but not necessarily $t(m)^{\gamma}>(a-1)m^{b-1}$. 

It turns out that the following version of the Property $[\star]$ works. This definition depends now on three parameters $c_g,p_g,d_g$.
The first parameter $c_g$ basically determines from which $n$ on the condition holds. The second parameter $p_g$ says that if $g(n)>an^b$, then
there are sufficiently many smaller integers $m$ with $g(m)>(a-p_g)n^{b-p_g}$. The third parameter $d_g$ determines the size of the interval.\\[0.3ex]
\emph{Definition of the Property $[\star ]$}\\
We say that a function $g(n)$ satisfies \emph{Property $[\star ]$} with parameters $c_g,p_g,d_g$ 
if for all $n\geq 2^{(c_g)^2}$ and all $c_g\leq a\leq b\leq\frac{\log (n)}{c_g}$, if 
$g(n)>an^b$,
then there exist pairwise distinct integers
\[m_1,\ldots, m_{\frac{\log\log (n)}{c_g}}\in I_{n,d_g}
=\left (n^{1\slash h},\: n^{1\slash h}\cdot \left (1+\frac{\log (n)}{n^{1\slash (h\cdot d_g)}}\right )^{d_g}\right )\]  
such that for $i=1,\ldots , \frac{\log\log (n)}{c_{g}}$, $g(m_i)>(a-p_g)\cdot n^{b-p_g}$ and $m_i$ is not an $h$-power, i.e. not of the form $m_i=(m_i')^h$
for any integer $m'_i$.\\[1.5ex]
Now the machines in our subfamily will depend on two parameters: the machine index $i$ of the original machine $S_i$ and the parameter $d$.
For each such pair $i,d$, we construct a new machine $\tilde{S}_{i,d}$ such that the function $\mbox{time}_{\tilde{S}_{i,d}}(n)$
will satisfy the Property $[\star]$ with parameters $c_{i,d}=2^{dh}\cdot (dh)^3\cdot 2^c\cdot i^c$, $p_i=\lceil i\slash 2\rceil$ and $d$.
The particular choice of the parameter $c_{i,d}$ will become
clear below in Section \ref{implication_section} in the proof of Lemma \ref{sig1lemma}. These machines will then be arranged in a linear list, denoted
as $\tilde{S}_{(1)},\tilde{S}_{(2)},\ldots $, and in the construction of the union function, machines will be added to the list of guesses in this order.
 \\[0.3ex]
Furthermore, we also have to modify the construction of the union function $F$. Recall what we want to achieve.
\begin{itemize}
\item $F$ is supposed to satisfy Property $[\star ]$.\\
      As we have already pointed out, this is the reason why the implication in the definition of Property $[\star ]$ is of the form
      $g(n)>an^b\:\Longrightarrow\:\exists\ldots g(m_l)>(a-p_g)m_l^{b-p_g}$ with the additional parameter $p_g$.
\item $F$ is supposed to satisfy the Padding Inequality $F(n)^C\leq F(n^h)$ for all $n$.\\
      As we have already described, the intended way of assuring that $F$ satisfies the Padding Inequality is as follows: Suppose that
      $F(n)^C>F(n^h)$, and let the guess $(\tilde{S}_{(i)},b_i)$ be selected in stage $n$ and the guess $(\tilde{S}_{(j)},b_j)$ in stage $n^h$.
      Then this implies $b_j<b_i$, and therefore the guess $(\tilde{S}_{(j)},b_j)$ is already contained in the list of guesses at stage $n$ of the 
      construction. Now we make use of the fact that the running time function of machine $\tilde{S}_{(j)}$ satisfies Property $[\star ]$.
      This yields existence of sufficiently many integers $m_l\in I_{n^h,d_j}$ where the machine $\tilde{S}_{(j)}$ exceeds the time bound 
      $(b_j-p_j)m_l^{b_j-p_j}$, i.e. the guess $(\tilde{S}_{(j)},b_j-p_j)$ is violated. 
      We want to argue that in this case, in at least one of the stages $m_l$ the guess $(\tilde{S}_{(j)},b_j)$ is already selected and replaced, such that
      it cannot be contained in the list of guesses at stage $n^h$ of the construction, a contradiction. 
\item Since a violation of the guess $(\tilde{S}_{(j)},b_j)$ at stage $n^h$ only guarantees violations of the weaker guess $(\tilde{S}_{(j)},b_j-p_j)$
      at stages $m_l\in I_{n^h,d_j}$, guesses of this form have also to be taken into account in the construction of the union function.
      Therefore we distinguish now between stages $n$ such that $n$ is an $h$-power and other stages. When $n$ is an $h$-power, 
      we just diagonalize against violated guesses of the form $(\tilde{S}_{i,d},b)$, as described before.
      When $n$ is not an $h$-power, we also take into account guesses of the form
      $(\tilde{S}_{i,d},b-p_i)$, where $(\tilde{S}_{i,d},b)$ is a guess in the original list.
      We call this the {\sl extended list of guesses} and denote it as ${\mathcal L}^e$.
\item This has another consequence for the precise definition of Property $[\star ]$. Namely, since 
      in the construction of the union function $F$, 
      guesses of the form $(\tilde{S}_{i,d},b-p_i)$ are considered only in stages $m$ such that $m$ is not an $h$-power, we require
      in the definition of Property $[\star ]$
      that if $g(n)>an^b$, then there exist sufficiently many integers $m_l$ within the interval $I_{n,d_g}$ which are {\sl not $h$-powers} such that
      $g(m_l)>(a-p_g)m_l^{b-p_g}$. 
\item This in turn has another consequence for the construction of the union function. 
      Recall that we want to use the padding inequality to show that $P=\mbox{DTIM}\tilde{\mbox{E}}(F)=\tilde{\Sigma}_2(F)=\Sigma_2^p$ also implies
      $\mbox{DTIM}\tilde{\mbox{E}}(F^{C^2})=\tilde{\Sigma}_2(F^{C^2})$. Given some problem $L\in\tilde{\Sigma}_2(F^{C^2})$, we want to construct an associated padded 
      version $L'$ of $L$, where $L'=\{x10^k|x\in L,k=|x|^{h^2}-|x|-1\}$. Then the decision problem $L'$ can be solved by a $\Sigma_2$-machine in time
      $F(|x|)^{C^2}$, and the padding inequality yields that $F(|x|)^{C^2}\leq F(|x|^h)^C\leq F(|x|^{h^2})$. Thus we obtain that $L'\in\Sigma_2(F)$. Now we also 
      want to conclude that $L'\in\tilde{\Sigma}_2(F)$, which would then yield $L'\in P$ and therefore 
      $L\in P=\mbox{DTIM}\tilde{\mbox{E}}(F)\subseteq \mbox{DTIM}\tilde{\mbox{E}}(F^{C^2})$. 

      But now the following problem occurs. The padded version $L'$ of $L$ has the property that every element $y\in L'$ has an $h$-power length. On the other hand, Property $[\star ]$
      requires existence of integers $m_l$ which are not $h$-powers such that the running time on input length $m_l$ is sufficiently large. This means that we cannot guarantee that 
      $L'\in \tilde{\Sigma}_2(F)$.
\end{itemize}

We solve this last problem as follows. First we modify the padding construction. When we start from some problem $L\in\tilde{\Sigma}_2(F^{C^2})$, then 
the strings $y=x10^k$ in the padded version $L'$ of $L$ are constructed in such a way that their string length is not an $h$-power but an $h$-power minus one.
Namely, 
\[L'\:\: =\:\: \{x10^k\mid x\in L,\: k=|x|^{h^2}-1\}.\]
Now we obtain that $L'$ can be solved by a $\Sigma_2$-machine in time $F(|x|)^{C^2}\leq F(|x|^h)^C\leq F(|x|^{h^2})$, i.e. $L'\in\Sigma_2(F(n+1))$, which means that 
$L'\in\Sigma_2(G(n))$ for the function $G(n):= F(n+1)$. As we shall prove below, we also have $L'\in\tilde{\Sigma}_2(F(n+1))$, and here it is important that the string lengths
of elements from $L'$ are not $h$-powers.  
But we would still like to conclude that $L'\in P$. Therefore, we
construct the union function $F$ such that it does not only satisfy $\tilde{\Sigma}_2^p=\tilde{\Sigma}_2(F)$, but also 
$\tilde{\Sigma}_2^p=\tilde{\Sigma}_2(F(n+1))$. Here $\tilde{\Sigma}_2(F(n+1))$ is the class of all problems that can be solved by some 
machine $\tilde{S}_{i,d}$ whose running time at input length $n$ is bounded by a constant times $F(n+1)$.

So intuitively we do not only require that the running time of a machine at input length $n$ is polynomially bounded iff it is bounded by $O(F(n))$, but also 
that it is polynomially bounded iff it is bounded by $O(F(n+1))$. While this distinction is usually unnecessary for {\sl ''harmless'' functions} like 
polynomials, it may in general make a huge difference for functions which are constructed in a diagonalization process.

Now achieving $\tilde{\Sigma}_2^p=\tilde{\Sigma}_2(F(n))=\tilde{\Sigma}_2(F(n+1))$ is not difficult at all. We just 
maintain for each machine $\tilde{S}_{i,d}$ two guesses in the list: One which is tested for violations at input length $n$, and one which is tested at input length
$n-1$. Thus the list of guesses ${\mathcal L}$ consists of two sublists ${\mathcal L}_1$ and ${\mathcal L}_2$. At stage $n$, both sublists contain guesses for 
the first $\log^*(n)$ machines. At stage $n$ when the function value $F(n)$ is determined, the guesses in the first list
${\mathcal L}_1$ are tested for violations at inputs of length $n$, and guesses in the second list ${\mathcal L}_2$ are tested for violations at input length $n-1$.
%
Thus we obtain the following approach.\\[2ex]
\emph{Construction of the Union Function $F$.} 
The union function $F$ for $\Sigma_2^p=\tilde{\Sigma}_2^p$ with respect to the indexing $(\tilde{S}_{i,d})$ is now constructed as follows.
We arrange these machines in a linear order, denoted as $\tilde{S}_{(1)},\tilde{S}_{(2)},\ldots$ such that each $\tilde{S}_{i,d}$ occurs in this list and such that from
the number $j$ we can efficiently compute the parameters $i,d$ with $\tilde{S}_{i,d}=\tilde{S}_{(j)}$. The $\Sigma_2$-machines $\tilde{S}_{i,d}$ will be added to 
the list of guesses in this order.
As before, ${\mathcal L}_n={\mathcal L}_{n,1}\cup {\mathcal L}_{n,2}$ denotes the list of guesses at the beginning of stage $n$. We construct $F$ as follow:
\begin{itemize}
\item For each stage $n\in I_t$ and $j\in\{1,2\}$, the list ${\mathcal L}_{n,j}$ 
      will contain guesses for the first $t=\log^*n$ machines $\tilde{S}_{(1)},\ldots , \tilde{S}_{(t)}$.
      Additionally we maintain lists ${\mathcal L}'_{n,j},j=1,2$ which contain guesses of the form $(\tilde{S}_{(i)},b_i-p_{(i)})$. These lists
      are generated at the beginning of each interval $I_t$. Namely, if $n$ is the first integer in the interval $I_t$, then 
      ${\mathcal L}'_{n,j}$ consists of all the guesses $(\tilde{S}_{(i)},b_i-p_{(i)})$ such that  $(\tilde{S}_{(i)},b_i)$ is contained in ${\mathcal L}_{n,j}$.
\item In stages $n$ such that $n$ is an $h$-power (i.e. $n^{1\slash h}$ is an integer number), we proceed as before and select the smallest violated guess
      $(\tilde{S}_{(i)},b_i)$ from the list ${\mathcal L}_n={\mathcal L}_{n,1}\cup {\mathcal L}_{n,2}$.
      If this guess is selected from the first list ${\mathcal L}_{n,1}$, we 
      set $F(n)=n^{b_i}$ and the guess is replaced by $(\tilde{S}_{(i)},t)$ in the list ${\mathcal L}_{n,1}$.
      If the guess is selected from the second list ${\mathcal L}_{n,2}$, then 
      we set $F(n)=(n-1)^{b_i}$ and the guess is replaced by $(\tilde{S}_{(i)},t)$ in the list ${\mathcal L}_{n,2}$.
\item In stages $n$ such that $n$ is not an $h$-power, we consider the so called extended list of guesses 
      ${\mathcal L}_{n,1}\cup {\mathcal L}'_{n,1}\cup {\mathcal L}_{n,2}\cup {\mathcal L}'_{n,2}$.
      We select from this extended list the smallest violated guess with respect to an order which we describe in Section \ref{constr_section}. 
      However, if a guess is selected from ${\mathcal L}'_{n,1}\cup {\mathcal L}'_{n,2}$, it will be removed from that list.
      This means that in stages within the
      interval $I_t$, for each machine $\tilde{S}_{(i)},i\leq t$ and every $j\in\{1,2\}$, 
      at most once a guess of the form $(\tilde{S}_{(i)},b_i-p_i)$ is selected from the sublist ${\mathcal L}'_{j}$. 
\end{itemize}
It will turn out that the use of the extended list of guesses implies that the resulting union function $F$ will satisfy the inequality
$F(n^{1\slash h})^C\leq F(n)$ for every $h$-power $n$.

Finally we will show that the result from Gupta \cite{G96} also holds for the restricted class of $\Sigma_2$-machines $(\tilde{S}_{(i)})$.
This will then give a contradiction, hence the assumption $P=\Sigma_2^p$ cannot hold.

At the end of this section we give a roadmap of the constructions and results given in this paper.
\begin{itemize}
\item We assume $P=\Sigma_2^p$. 
\item We start from a standard indexing $(S_i)$ of $\Sigma_2$-machines. Without loss of generality, $S_1$ is a $\Sigma_2$-machine which never terminates on any input.
      Furthermore, we assume that each $\Sigma_2$-machine occurs infinitely often in this indexing.
\item We show in Lemma \ref{sig1lemma} that there exists a constant $c$ such that for a given $\Sigma_2$-machine $S_i$, 
      an input $x$ of length $n$ and integer numbers $a,b$
      it can be checked deterministically in time $c\cdot (i\cdot a\cdot |x|^b)^{c}$ 
      if the running time of $S_i$ on any input of length $n$ exceeds $a\cdot n^{b}$.
      Here $c$ is a \emph{global constant}, i.e. it does not depend on $i,n,a,b$.
\item In Section \ref{prog_section} we give the definition of Property $[\star ]$.
      The parameter $h$ in the definition of Property $[\star ]$ is defined as $h=20(c+2)$ and thus
      only depends on the constant $c$ from Lemma \ref{sig1lemma}. We construct our new family $(\tilde{S}_{i,d})$ of $\Sigma_2$-machines. These machines will be 
      arranged in a linear manner at the beginning of Section \ref{constr_section}. Thus $\tilde{S}_{(i)}$ will be the $i$th machine in this order.
\item In Lemma \ref{progsys_lemma} we show that for all $i,d$, the running time of $\tilde{S}_{i,d}$ is bounded by the running time of $S_i$.
      Furthermore we show that for all $i,d$, the running time function 
      $\mbox{time}_{\tilde{S}_{i,d}}(n)$ of machine $\tilde{S}_{i,d}$ satisfies
      Property $[\star ]$.
      Furthermore, if $\mbox{time}_{S_i}(n)$ already satisfies Property $[\star]$, then $L(S_i)=L(\tilde{S}_{i,d})$. Since functions of the form $b\cdot n^b$ satisfy the Property $[\star ]$,
      this implies that $P=\tilde{P}=\Sigma_2^p=\tilde{\Sigma}_2^p$ holds (Lemma \ref{classes_lemma}).
\item We construct our union function $F$ in Section \ref{constr_section}. In Lemma \ref{union_lemma} we prove that \\
      $\tilde{P}=\mbox{DTIM}\tilde{\mbox{E}}(F)=\mbox{DTIM}\tilde{\mbox{E}}(F(n+1))=\tilde{\Sigma}_2(F(n+1))=\tilde{\Sigma}_2(F)=\tilde{\Sigma}_2^p$.
\item In Lemma \ref{F_computation_lemma} we show that $F(n)$ can be computed deterministically in time $F(n)^C$, where $C=10c$ and $c$ is the constant from
      Lemma \ref{sig1lemma}. Especially we have $C<h\slash 2$.
\item In Lemma \ref{ineq_lemma} we show that $F$ satisfies the {\sl Padding Inequality}, namely the inequality $F(n^{1\slash h})^C\leq F(n)$ for every $h$-power $n$.
\item In Lemma \ref{padding_lemma} we make use of this inequality and a padding construction to show that 
      $\mbox{DTIM}\tilde{\mbox{E}}(F)=\tilde{\Sigma}_2(F)$ also implies $\mbox{DTIM}\tilde{\mbox{E}}(F^{C^2})=\tilde{\Sigma}_2(F^{C^2})$.
\item In Lemma \ref{F_propertystar_lemma} we show that the function $F(n)^{C^2}$ has Property $[\star ]$.
\item Section \ref{gupta_section} provides a variant of Gupta's result for our new indexing $(\tilde{S}_{(i)})$. We show in Theorem \ref{separation_theorem}
      that for each function $t(n)\geq n\log^*n$ which is deterministically computable in time $t(n)^{1-\epsilon}$ for some $\epsilon >0$ and satisfies Property $[\star ]$,
      we have $\mbox{DTIM}\tilde{\mbox{E}}(t)\subsetneq\tilde{\Sigma}_2(t)$. Especially this holds for $t(n)=F(n)^{C^2}$, which yields the desired contradiction. 
      Hence the assumption $P=\Sigma_2^p$ cannot hold.
\end{itemize}

\section{Preliminaries}\label{prelim_section}
An Alternating Turing Machine 
has states which are labelled as accepting, rejecting, universal or existential. The semantics of ATMs generalizes both nondeterministic 
and co-nondeterministic computations: A subtree $T$ of the computation tree of an ATM on a given input $x$ is called accepting subtree
if either the root of this subtree is a configuration with an accepting state, or the state is existential and there exists 
a child of this node whose subtree is accepting, or the state is universal and for every child of this node, the subtree rooted at this 
child is an accepting subtree.
The input $x$ is accepted by the ATM if the computation tree itself is accepting.

%

For a precise description of \textit{Alternating Turing Machines (ATM)} we refer to \cite{CKS81}, \cite{PPR80} and \cite{PR81}.
Most of the notations which we use here are taken from \cite{PPST83}.

\emph{Running Time of Alternating Turing Machines.}
Two different notions of running time of ATMs have been used in the literature (cf. \cite{BGW70},\cite{SFM78}). One version is to say
that for a given ATM $M$ and input $x\in L(M)$, the running time of $M$ on input $x$ is the minimum depth of an accepting subtree 
of the computation tree of $M$ on input $x$. 
A second version which we denote as $\mbox{time}_M$ is defined as follows: $\mbox{time}_M(x)\leq t(|x|)$ if all computation paths of $M$ on input $x$
have length at most $t(|x|)$.
For time-constructible time bounds $t(n)$, both notions are equivalent. 
In this paper, we will use the second notion $\mbox{time}_M$. 


A function $t(n)$ is called \emph{time-constructible} if $t(n)\geq n$ and $t(n)$ can be computed in time $O(t(n))$, 
i.e. there exists some deterministic machine $M$ such that
for all $x\in\{0,1\}^*$, $M(x)=t(|x|)$ in time $O(t(|x|))$.

For a function $f(n)$ mapping integers to integers, let
$\Sigma_k(f(n))$ denote the set of decision problems which are recognized by some alternating machine in time $O(f(n))$ which starts in an
existential state and changes the quantification (between existential and universal states) at most $k-1$ times. 
$\Pi_k(f(n))$ is the set of decision problems which are recognized by some alternating machine in time $O(f(n))$ which starts in an
universal state and changes quantification at most $k-1$ times. 

For a class of functions $C$, let $\Sigma_k(C)$ denote the union of all sets $\Sigma_k(t),\:t\in C$,
and let the classes $\Pi_k(C)$ be defined accordingly. 
Let $\mbox{poly}$ denote the class of all functions $p(n)=k\cdot n^k$. Thus, $\mbox{DTIME}(\mbox{poly})=P$ and
$\Sigma_1(\mbox{poly})=NP$.

$\mbox{PH}=\bigcup_k\Sigma_k^p$ is the \emph{Polynomial Hierarchy}
\cite{S76} with levels $\Sigma_k^p=\Sigma_k(\mbox{poly})$. 
In particular, $P=\Sigma_0^p$ and $NP=\Sigma_1^p$.
The polynomial hierarchy is well known to provide \emph{Downward Separation}. This means that whenever at some level $k$ we have $\Sigma_k^p=\Sigma_{k+1}^p$, 
then the whole hierarchy collapses to the level $k$, namely $\Sigma_k^p=\mbox{PH}$ (cf. \cite{S76}, Thm. 3.2).

%
%
\emph{Indexings of Alternating Machines.}
An \emph{indexing} (also called \emph{G{\"o}delization}) of a class of machines is an encoding of those machines by strings or integer numbers.
Let us briefly describe indexings for the classes of $\Sigma_k$-machines and give a statement about the 
time complexity of an associated universal function. All the machines which we consider here work over the binary alphabet $\{0,1\}$.
An \emph{Alternating Turing Machine (ATM)} over the alphabet $\{0,1\}$ is a tuple $M=(Q,q_0,k,\delta,\lambda)$ with 
\begin{itemize}
\item[] $Q$ being the set of states, $q_0\in Q$ the initial state,
\item[] $k$ being the number of tapes of $S$,
\item[] $\delta\subseteq Q\times\{0,1\}^k\:\times\: Q\times\{0,1\}^k\times\{-1,0,1\}^k$
	the transition function (a multi-valued function, in \cite{CKS81} called the \emph{next move relation}), 
\item[] $\lambda\colon Q\to\{\forall,\exists,a,r\}$ the labeling of states
	(universal states $q$ with $\lambda (q)=\forall$, existential states $q$ with $\lambda (q)=\exists$, accepting states
	$q$ with $\lambda (q)=a$ and rejecting states $q$ with $\lambda (q)=r$).
\end{itemize}
Existential or universal states for which the transition function is single-valued for every combination of input symbols on the $k$ tapes 
can be considered as deterministic states. Alternatively one could also explicitly encode deterministic states by an additional label.
ATMs can be encoded as binary strings in a standard way. Suppose we fix such an encoding such that the set $L_{pc}\subseteq\{0,1\}^*$
of all binary strings which encode an ATM (i.e. the program codes) is in $\mbox{DTIME}(n)$, no program code is prefix of another and there exists an
ATM $U$ (''universal simulator'') such that for each program code $e\in L_{pc}$ and every $x\in\{0,1\}^*$, $\langle e,x\rangle\in L(U)$ iff 
$x\in L(M_e)$, and furthermore $\mbox{time}_{U}(\langle e,x\rangle )=|e|\cdot \mbox{time}_{M_e}(x)^{O(1)}$.  
Here, $M_e$ denotes the ATM whose program code is the string $e$. The time-bound for the universal machine $U$ can be achieved by a 
standard step-by-step simulation of the machine $S_e$ on input $x$. For further details we refer to \cite{SFM78} and \cite{K80}. 

Now we can fix some ATM $M_0$ and extend the indexing in such a way that $M_e:=M_0$ for all $e\not\in L_{pc}$. If we now 
identify integer numbers with their binary representations, we obtain an indexing of the Alternating Turing Machines
$\langle M_i\rangle_{i\in {\mathbb N}}$ together with the universal machine $U_0$ such that 
$L(U_0)=\{\langle i,x\rangle\: |\: x\in L(M_i)\}$ and $\mbox{time}_{U_0}(\langle i,x\rangle )=|i|\cdot\mbox{time}_{M_i}(x)^{O(1)}$.

In the same way we obtain an indexing $\langle S_i\rangle_{i\in {\mathbb N}}$ with a universal machine $U$ for the $\Sigma_2$-machines.

\emph{Properties of Integer Numbers.}
Let $A$ be a property of integer numbers, i.e. a subset $A\subseteq {\mathbb N}$. We write $A(n)$ for $n\in A$. The property $A$ 
\emph{holds almost everywhere} (Notation: $A(n)$ a.e.) if the set ${\mathbb N}\setminus A$ is finite.
The property $A$ holds \emph{infinitely often} (Notation: $A(n)$ i.o.) if $A$ is an infinite subset of ${\mathbb N}$.

\emph{Logarithm.} In this paper, $\log(n)$ denotes the ceiling of the binary logarithm: $\log (n) :=\lceil \log_2(n)\rceil$.\\[0.3ex]
It is well known that polynomial functions grow asymptotically faster than polylogarithmic functions. In our constructions and proofs, we need some explicit estimates
of the point from which on the linear function $f(n)=n$ majorizes a given polylogarithmic function $a\cdot\log (n)^b$. Such an estimate
is provided in the following auxiliary lemma, first for the binary logarithm $\log_2(n)$ and then, based on that, for the ceiled logarithm $\log (n)=\lceil\log_2(n)\rceil$.
The estimate is not tight, but sufficient for our purpose. 
\begin{lemma}\label{auxiliary-lemma}
Let $\alpha$ and $\beta$ be positive integers such that $\alpha\geq 4$ and $\beta\geq 16$. Then (i) and (ii) hold.
\begin{itemize}
\item[(i)] For every $n\geq 2^{\alpha\beta^2}$, we have $\alpha\cdot \log_2(n)^{\beta}\:\leq\: n$.
\item[(ii)] For every $n\geq 2^{\alpha (\beta +1)\cdot 2^{\beta}\cdot \beta^2}$, we have $\alpha\cdot \log (n)^{\beta}\:\leq\: n$.
\end{itemize}
\end{lemma}
\begin{proof}
(i) First we show that the inequality holds for $n=2^{\alpha\cdot \beta^2}$. Then we use first and second derivatives in order to show that it also 
holds for all $n>2^{\alpha\cdot \beta^2}$.\\[0.2ex]
For $n=2^{\alpha\cdot \beta^2}$, we have to show that 
\[\alpha\cdot \log_2(n)^{\beta}=\alpha\cdot (\alpha\beta^2)^{\beta}\:\leq\: 2^{\alpha\beta^2},\]
which is (by taking logarithms) equivalent to 
\begin{equation}\label{aux_xx}
\log_2(\alpha)+\beta (\log_2(\alpha)+2\log_2(\beta))\:\leq\: \alpha\beta^2.
\end{equation}
The left hand side of inequality (\ref{aux_xx}) is less or equal $\beta (2\log_2(\alpha)+2\log_2(\beta))$. 
We divide both sides of (\ref{aux_xx}) by $\beta$, and thus it is sufficient to show
\[2\log_2(\alpha)+2\log_2(\beta)\:\leq \alpha\cdot \beta.\]
Now $\alpha\geq 4$ implies $2\log_2(\alpha)\leq \alpha$ and $\beta\geq 16$ implies $2\log_2(\beta)\leq \beta$. Thus, the left hand side of the last inequality is 
$\leq \alpha+\beta\leq \alpha\cdot \beta$, and the inequality $\alpha\cdot \log_2(n)^{\beta}\leq n$ holds for $n=2^{\alpha\beta^2}$.\\[0.2ex]
Now we build the first derivatives of both sides $L(n)=\alpha\log_2(n)^{\beta}$ and $R(n)=n$ of the inequality in (i):
\[L'(n)=\frac{d}{dn}\alpha\log_2(n)^{\beta}=\alpha\beta\cdot\frac{\log_2(n)^{\beta -1}}{\ln(2)\cdot n},\:R'(n)=\frac{d}{dn}n=1\]
Again, for $n=2^{\alpha\beta^2}$ we have $L'(n)= \alpha\beta\cdot\frac{(\alpha\beta^2)^{\beta -1}}{\ln(2)\cdot 2^{\alpha\beta^2}}$.
Thus we have $L'(2^{\alpha\beta^2})\leq R'(2^{\alpha\beta^2})=1$ iff
\begin{equation}\label{aux_x}
\frac{1}{\ln (2)}\cdot \alpha\beta\cdot (\alpha\beta^2)^{\beta -1}\:\leq\: 2^{\alpha\beta^2}.
\end{equation}
The logarithm of the left hand side of (\ref{aux_x}) is
\[\beta\log_2(\alpha)+(2\beta -1)\log_2(\beta )-\log_2\ln (2)\leq \beta (\log_2(\alpha)+2\log_2(\beta ))-\log_2\ln (2).\]
We set $l:=-\log_2\ln (2)$. Note that $l>0$.
Thus (\ref{aux_x}) holds if $\beta (\log_2(\alpha )+2\log_2(\beta ))+l\leq \alpha\beta^2$, 
i.e. $\log_2(\alpha)+2\log_2(\beta)+l\slash \beta\leq \alpha\cdot \beta$. Now again since $\alpha\geq 4$ and $\beta\geq 16$, we have 
$\alpha\geq\log_2(\alpha)$ and $\beta\geq 2\log_2(\beta )$. Since $\frac{1}{2}<l<\frac{3}{5}$, the inequality $\alpha+\beta+\frac{l}{\beta}<\alpha\cdot\beta$
holds, which implies  
$L'(n)\leq R'(n)$ for $n=2^{\alpha\beta^2}$. Now if we show that the second
derivatives satisfy $L''(n)\leq R''(n)$ for all $n\geq 2^{\alpha\beta^2}$, then (i) holds. Since $R''(n)=0$, it suffices to show 
that $L''(n)\leq 0$ for $n\geq 2^{\alpha\beta^2}$.
We have 
\[L''(n)=\frac{\alpha\beta}{\ln (2)}\cdot \frac{(\beta -1)\log_2(n)^{\beta -2}\cdot\frac{1}{\ln (2)}\cdot \frac{1}{n}\cdot n-\log_2(n)^{\beta -1}\cdot 1}{n^2},\]
and therefore we have $L''(n)\leq 0$ iff $\frac{\beta -1}{\ln (2)}\log_2(n)^{\beta -2}\leq\log_2(n)^{\beta -1}$ iff $\frac{\beta -1}{\ln (2)}\leq\log_2(n)$. 
This last inequality holds, since we have $\ln (2)>\frac{1}{2}$ and therefore also 
$\log_2(n)\geq\log_2(2^{\alpha\beta^2})=\alpha\beta^2\geq 2\cdot (\beta -1)>\frac{\beta -1}{\ln (2)}$. 
Altogether we have shown that $L(2^{\alpha\beta^2})\leq R(2^{\alpha\beta^2})$ and $L'(n)\leq R'(n)$ for all $n\geq 2^{\alpha\beta^2}$, and therefore
$L(n)\leq R(n)$ holds for all $n\geq 2^{\alpha\beta^2}$. This concludes the proof of (i). \\[0.2ex]
\emph{Proof of (ii):} 
Since $\log (n)$ is defined as $\lceil\log_2(n)\rceil$, we have $\log (n)\leq\log_2(n)+1$ and therefore also 
$\log (n)^{\beta}\leq(\log_2(n)+1)^{\beta}\leq (\beta +1)\cdot 2^{\beta}\cdot \log_2(n)^{\beta}$. 
Hence in order to show that $\alpha\log(n)^{\beta}\leq n$, it suffices to show that
$\alpha(\beta +1)\cdot 2^{\beta}\cdot\log_2 (n)^{\beta}\leq n$. 
Using (i), this last inequality holds for $n\geq 2^{\alpha'\cdot \beta^2}$ with $\alpha'=\alpha (\beta +1)\cdot 2^{\beta}$, which concludes the proof of (ii).
\end{proof}

\section{An Implication from $P=\Sigma_2^p$}\label{implication_section}
%
We assume $P=\Sigma_2^p$. In this section we show that this implies that we can test deterministically in time $p(n)^c$ 
if a given $\Sigma_2$-machine exceeds running time $p(n)$ on input length $n$. Both the construction of the subfamily $(\tilde{S}_i)$ of $\Sigma_2$-machines and 
the construction of the union function will rely on this result.
Recall that we assume $(S_i)$ to be a standard indexing of the $\Sigma_2$-machines. 
Let us consider the following decision problem: We are given a 
tuple $\langle i,x,a,b\rangle$ consisting of a machine index $i$, some input string $x$ and two integers $a,b$, and we want to decide
if the running time of
machine $S_i$ on input strings of length $n=|x|$ is always bounded by $a\cdot n^{b}$. As a direct consequence of the assumption $P=\Sigma_2^p$,
the following lemma shows that this problem can be solved deterministically in time $c\cdot (i\cdot a\cdot n^{b})^c$ for some fixed constant $c$.
\begin{lemma}\label{sig1lemma}
Suppose $\mbox{P}=\Sigma_2^p$. Then there exists some constant $c$ such that the decision problem
\[\begin{array}{l}
L_{check}=\left\{\langle x,i,a,b\rangle\: |\:\mbox{for all inputs $x'$ of length}\:\: |x'|=|x|,\: \mbox{all computation}\right.\\ 
\quad\quad\quad\quad\quad\quad\left.\mbox{paths of $S_i$ on input $x'$ terminate within $\leq a\cdot |x|^{b}$ steps}\right\}
\end{array}\]
can be solved deterministically in time $c\cdot  (i\cdot a\cdot n^{b})^c$,
more precisely: There exists a deterministic algorithm 
which recognizes $L_{check}$ and has a running time on input $\langle x,i,a,b\rangle$ being bounded by $c\cdot (i\cdot a\cdot |x|^{b})^c$. 
\end{lemma}
\begin{proof}
We make use of padding. Consider the following auxiliary decision problem $L_{check}'$ which is a padded version of the complement of $L_{check}$: 
\[\begin{array}{l}
L_{check}'=\left \{\langle x,i\rangle 10^k\: |\: \exists x'\: \left [|x'|=|x|\:\mbox{and at least one computation path of $S_i$}\right.\right.\\
\left.\left.\hspace{3.1cm}\mbox{on input $x'$ which does not terminate within $\leq k$ steps} \right ]\:\right \}
\end{array}\]
We observe that $L_{check}'\in \mbox{NP}$: For a given input $\langle x,i\rangle 10^k$, we can just guess nondeterministically some string $x'$ of length $|x|$
and some computation path of length $k$ for the machine $S_i$ on input $x'$ and then check deterministically (by a step-by-step simulation of
the computation path of $S_i$ on input $x'$) that this path does not terminate within $k$ steps.

Since $\mbox{NP}=\Sigma_1^p\subseteq\Sigma_2^p$ and we assume $\mbox{P}=\Sigma_2^p$, this implies that 
the decision problem $L_{check}'$ is in $\mbox{P}$. 
Thus, the complement of $L_{check}'$ is also in $P$. Therefore, let ${\mathcal B}'$ be a deterministic algorithm which accepts precisely the 
complement of $L_{check}'$ and whose running time on input $\langle x,i\rangle 10^k$ is bounded by $c_0\cdot (|x|+|i|+k)^{c_0}$ for  
all instances $\langle x,i\rangle 10^k$, for some constant $c_0$.
%
This directly gives us the following deterministic algorithm ${\mathcal B}$ for the decision problem $L_{check}$:
\begin{itemize}
\item[] {\bf \underline{Algorithm $\mathbf{{\mathcal B}}$ for $\mathbf{L_{check}}$}}\\
	{\bf Input:} $\langle x,i,a,b\rangle$\\
	Construct the string $\langle x,i\rangle 10^{a\cdot |x|^{b}}$\\
	Compute and return ${\mathcal B}'\left (\langle x,i\rangle 10^{a\cdot |x|^{b}}\right )$
\end{itemize}
Algorithm ${\mathcal B}$ needs $O(a\cdot |x|^b)$ steps to construct the string $\langle x,i\rangle 10^{a\cdot |x|^{b}}$.
Furthermore, algorithm ${\mathcal B}$ uses $O\left (c_o\cdot \left (|x|+|i|+a\cdot |x|^{b}\right )^{c_0}\right )$ steps to simulate ${\mathcal B}'$ 
on input $\langle x,i\rangle 10^{a\cdot |x|^{b}}$.
Thus there exists some constant $c$ such that the running time of algorithm ${\mathcal B}$ on input $\langle x,i,a,b\rangle$ is 
always less or equal $c\cdot (i\cdot a\cdot |x|^{b})^c$. This concludes the proof of the lemma.
\end{proof}
{\bf Remark.} From now on we fix a constant $c$ such that the statement in the Lemma \ref{sig1lemma} holds. Moreover, we can choose $c$ sufficiently large for later purpose.
In particular, we assume that $c$ is sufficiently large that 
for all $n\geq 2^c$, $\log\log (n-1)>3(\log^*(n))^2$. This will be used below in the proof of Lemma \ref{ineq_lemma} . 
\section{A New Family of $\Sigma_2$-Machines}\label{prog_section}
In this section we will first define the \emph{Property $[\star ]$} for functions $f\colon {\mathbb N}\to {\mathbb N}$. Then 
we start from a standard indexing $(S_i)$ of $\Sigma_2$-machines and construct our new family of $\Sigma_2$-machines $(\tilde{S}_{i,d})$.
This family will contain for each machine $S_i$ and every integer $d\in {\mathbb N}$ a machine $\tilde{S}_{i,d}$.
We will construct the machines $\tilde{S}_{i,d}$ in such a way that for every $i$ and $d$, the running time function $\mbox{time}_{\tilde{S}_{i,d}}(n)$ of machine $\tilde{S}_{i,d}$
satisfies Property $[\star ]$.
Furthermore we show that if the running time function $\mbox{time}_{S_i}(n)$ of machine $S_i$ already satisfies Property $[\star ]$, 
then there exist some machine index $j$ and some integer $d$ such that
$\mbox{time}_{\tilde{S}_{j,d}}(n)\leq \mbox{time}_{S_i}(n)$ and such that $L(\tilde{S}_{j,d})=L(S_i)$. Especially this will hold for all $i$ with 
$\mbox{time}_{S_i}(n)$ being polynomial in $n$. We shall show that this implies $\tilde{\Sigma}_2^p=\Sigma_2^p$.

Let us start by giving the precise definition of Property $[\star ]$.
\begin{definition}\label{property_def}\emph{(Property $[\star ]$ for Functions)}\\
We let $h=20\cdot (c+2)$, where $c$ is the constant from Lemma \ref{sig1lemma}. \\[1ex]
We say a function $g\colon {\mathbb N}\to {\mathbb N}$ has \emph{Property $[\star ]$} with parameters 
$c_g,d_g$ and $p_g$ if the following condition is satisfied.
\begin{itemize}
\item[{$[\star ]$}] For all $n\geq 2^{c_g^2}$, for all pairs of integers $a,b$ with $c_g\leq a\leq b\leq\frac{\log (n)}{c_g}$, the following holds:
                  If $g(n)>an^b$, then there exist 
                  pairwise distinct integers
                  \[m_1,\ldots , m_{\left\lceil\frac{\log\log n}{c_g}\right\rceil}\in I_{n,d}=
                  \left (n^{1\slash h},n^{1\slash h} \cdot\left (1+\frac{\log (n)}{n^{1\slash (h\cdot d_g)}}\right )^{d_g}\right )\]
                  such that $g(m_l)>(a-p_g)m_l^{b-p_g}$ and $m_l$ is not an $h$-power, $l=1,\ldots , \left\lceil\frac{\log\log n}{c_g}\right\rceil$.
\end{itemize}
\end{definition}
Now we construct our new family $(\tilde{S}_{i,d})$ of $\Sigma_2$-machines, which contains for each machine $S_i$ and each integer $d$ a machine $\tilde{S}_{i,d}$.
The machine $\tilde{S}_{i,d}$ will have the following property:
\begin{itemize}
\item The running time function $\mbox{time}_{\tilde{S}_{i,d}}(n)$ satisfies Property $[\star ]$.
\item If the running time function of the machine $S_i$ already satisfies the Property $[\star]$, then the machine 
      $\tilde{S}_{i,d}$ still computes the same as the machine $S_i$, formally: $L(S_i)=L(\tilde{S}_{i,d})$.
\end{itemize}
We shall show that functions of the form $b\cdot n^b$ satisfy Property $[\star ]$. So if a problem $L$ is contained in $\Sigma_2^p$, then there
always exists a $\Sigma_2$-machine $S_i$ which solves $L$ and has a running time precisely equal to $b\cdot n^b$ for some $b$. Therefore, the machine $\tilde{S}_{i,d}$
will also solve $L$, and this shows that $\Sigma_2^p=\tilde{\Sigma}_2^p$. Furthermore, if the machine $S_i$ is deterministic, then $\tilde{S}_{i,d}$ will also be
deterministic, which implies that we also have $P=\tilde{P}$.

The idea of how to construct $\tilde{S}_{i,d}$ is now as follows: On a given input $x$ of length $n$, the machine $\tilde{S}_{i,d}$ simulates the computation of the machine 
$S_i$ on input $x$. But at the same time, $\tilde{S}_{i,d}$ checks for increasing values $a,b$ if there are sufficiently many input lengths $m$ within the interval 
$I_{n,d}$ on which the running time of $S_i$ is at least $(a-p_i)m^{b-p_i}$. If this is not the case, then $\tilde{S}_{i,d}$ will stop the simulation within at most 
$an^b$ steps. Since the Property $[\star ]$ is recursive (the condition must hold for all $n$), it does not suffice to consider input lengths within the interval $I_{n,d}$, but
also within the intervals $I_{m,d}$ for $m\in I_{n,d}$ and so forth. The crucial part in the construction will be to show that machine $\tilde{S}_{i,d}$ always has enough time
to perform all these checks. Let us now give the details.

Recall that $S_1$ is a $\Sigma_2$-machine which runs to infinity on every input. We set $c_{1,d}=1$ for all $d\in {\mathbb N}$.
For $i>1$, we set 
\begin{eqnarray}
c_{i,d} & := & 2^{dh}\cdot (dh)^3\cdot 2^c\cdot i^c 
\end{eqnarray}
and $p_i:=\lceil i\slash 2\rceil$, where $c$ is the constant from Lemma \ref{sig1lemma}. 
The reason for this specific choice of the parameters $c_{i,d}$ and $p_i$ will become clear in the construction of the machines $\tilde{S}_{i,d}$.

The machine $\tilde{S}_{i,d}$ simulates the machine $S_i$ on the given input, but in parallel it will check on some extra working tapes
if the Property $[\star ]$ holds for $S_i$.
In order to give the precise definition of the machine $\tilde{S}_{i,d}$, we make use of the following predicate $P(\:)$. Intuitively, at a given 
point in the computation of the machine $\tilde{S}_{i,d}$, the predicate $P(\: )$ tells us if $\tilde{S}_{i,d}$ is allowed to continue its computation such as 
to satisfy Property $[\star ]$.
\begin{definition}{(Predicate $P$ for machines $\tilde{S}_{i,d}$)}\\
We say that predicate $P(i,d,n,a,b)$ holds if $(n\geq 2^{c_{i,d}^2}\:\mbox{and}\: c_{i,d}\leq a\leq b\leq\frac{\log (n)}{c_{i,d}})$ implies that
there exist pairwise distinct integers $m_l,l=1,\ldots \left\lceil\frac{\log\log (n)}{c_{i,d}}\right\rceil$
in the interval 
\[I_{n,d}\: =\: \left (n^{1\slash h},n^{1\slash h}\cdot \left (1+\frac{\log (n)}{n^{1\slash (d\cdot h)}}\right )^d\right )\]
such that for each such $l$, the integer $m_l$ is not an $h$-power, $\mbox{time}_{S_i}(m_l)>(a-p_i)m_l^{b-p_i}$ and $P(i,d,m_l,a-p_i,b-p_i)$ holds.
\end{definition}
We observe that the running time function $\mbox{time}_{S_i}(n)$ satisfies Property $[\star ]$ with parameters $c_{i,d},d,p_i$ iff for all $n\geq 2^{c_{i,d}^2}$ 
and for all $c_{i,d}\leq a\leq b\leq \frac{\log (n)}{c_{i,d}}$, $\mbox{time}_{S_i}(n)>an^b$ implies that the predicate $P(i,d,n,a,b)$ holds.

The $\Sigma_2$-machines $\tilde{S}_{i,d}$ will have the following properties.
For $i=1$ and $d\in {\mathbb N}$, $\tilde{S}_{1,d}$ is a $\Sigma_2$-machine which runs to infinity on every input.
For $i>1$,
the machine $\tilde{S}_{i,d}$ is defined as follows.
On input $x$ of length $n$, if $n\leq 2^{c_{i,d}^2}$, then $\tilde{S}_{i,d}$ just simulates $S_i$ on input $x$ and makes in total the same number of computation
steps as $S_i$. On the other hand, if $n>2^{c_{i,d}^2}$, then the machine $\tilde{S}_{i,d}$ proceeds as follows: Let $L=L_{i,d}(n)$ denote the number of pairs
of integers $(a,b)$ with $c_{i,d}\leq a\leq b\leq\frac{\log (n)}{c_{i,d}}$, and let
$(a_l,b_l),l=1,\ldots , L$ be these pairs in lexicographically increasing order, first ordered by the second entry $b_l$ and then by the first entry $a_l$.
We consider the associated time intervals $T_0=\left (0,a_1n^{b_1}\right ]$, $T_{l}=\left (a_ln^{b_l},a_{l+1}n^{b_{l+1}}\right ],l<L$
and $T_{L}=\left (a_{L}n^{b_{L}},\infty\right )$.
Now $\tilde{S}_{i,d}$ performs also a step-by-step simulation of the machine $S_i$ on input $x$, but 
additionally it uses a number of additional tapes to do the following. Within each interval $T_l,l<L$,
it checks if the predicate $P(i,d,n,a_{l+1},b_{l+1})$ holds.
If within some interval $T_l$ the computation of $S_i$ on input $x$ terminates, then $\tilde{S}_{i,d}$
terminates as well, with the same output (accept/reject).
If within some interval $T_l$, the computation of $S_i$ does not yet terminate but the
predicate $P(i,d,n,a_{l+1},b_{l+1})$ does not hold, $\tilde{S}_{i,d}$ also completes this interval and then terminates and rejects.
Otherwise, it continues within the next inerval $T_{l+1}$. If the computation reaches the interval $T_L$, then it just continues to simulate
the computation of $S_i$ and does not check the predicate $P(\: )$ anymore.
The computation of $\tilde{S}_{i,d}$ is organized in such a way that while being in an interval $T_l,l<L$,
it always makes precisely as many computation steps as $S_i$.
%
We give a pseudo-code description of the machine $\tilde{S}_{i,d}$.\\[1.1ex]
\hspace*{0.3cm}{\sl Machine $\tilde{S}_{i,d}$}
\vspace*{-0.1cm}
\begin{itemize}
\item[] {\sl Input:} $x$ of length $n$
\item[] {\sl If} $n< 2^{c_{i,d}^2}$, simulate the computation of machine $S_i$ on input $x$\\ and make in total the same number of computation steps
      as $S_i$.
\item[] {\sl If} $n\geq 2^{c_{i,d}^2}$\\
        \hspace*{0.5cm} {\sl For} $l=0,\ldots L-1$ (where $L=L_{i,d}(n)$)\\
        \hspace*{1.0cm} Continue the simulation of computation of $S_i$ on input $x$\\
        \hspace*{1.0cm} and at the same time check if $P(i,d,n,a_{l+1},b_{l+1})$ holds.\\
        \hspace*{1.0cm} If the computation of $S_i$ terminates within the interval $T_l$, \\
        \hspace*{1.0cm} then $\tilde{S}_{i,d}$ also terminates with the same output.\\
        \hspace*{1.0cm} If $P(i,d,n,a_{l+1},b_{l+1})$ does not hold, stop and reject.\\
        \hspace*{0.5cm} $\slash\star$ Now we are in the interval $T_{L}=\left (a_{L}n^{b_{L}},\infty \right )$ $\star\slash$\\
        \hspace*{0.5cm} Continue the simulation of computation of $S_i$ on input $x$.
\end{itemize}
We will now show that for each $l<L$, the size of the interval $T_l$ suffices to check if the predicate $P(i,d,n,a_{l+1},b_{l+1})$ holds.\\[1.0ex]
The predicate $P(i,d,n,a,b)$ is defined recursively. For $P(i,d,n,a,b)$ to hold it is required that there are sufficiently many
smaller integers $m$ in the interval $I_{n,d}$ such that $\mbox{time}_{S_i}(m)>(a-p_i)m^{b-p_i}$ and such that the predicate $P(i,d,m,a-p_i,b-p_i)$ holds. 
This in turn means 
that there exist sufficiently many integers $m'$ in the interval $I_{m,d}$ such that $\mbox{time}_{S_i}(m')>(a-2p_i)(m')^{b-2p_i}$, and so on. 
Now we have $\mbox{time}_{S_i}(m)>(a-p_i)m^{b-p_i}$ iff the tuple $(0^m,i,a-p_i,b-p_i)$ is a no-instance of the decision problem 
$L_{check}$ from Lemma \ref{sig1lemma}. 
We want to give an estimate of the number of instances of the problem $L_{check}$ we have to solve in order to decide if the predicate $P(i,d,n,a,b)$ holds. 
Therefore we consider the following set of integers $R_{i,d}(n)$, which is defined recursively along the definition of the predicate $P$:
$R_{i,d}(n)=\bigcup_lR_{i,d}^l(n)$ with
\[\begin{array}{r@{\:\:}c@{\:\:}l}
R_{i,d}^1(n) & = & I_{n,d}\: =\: \left (n^{1\slash h},\: n^{1\slash h}\left (1+\frac{\log (n)}{n^{1\slash (dh)}}\right )^d \right ),\\
R_{i,d}^{l+1}(n) & = & \bigcup_{{m\in R_{i,d}^l(n), m\geq 2^{c_{i,d}^2} }}I_{m,d}.
\end{array}\]
It follows directly from the definition of the predicate $P$ that in order to decide if $P(i,d,n,a,b)$ holds, it suffices to 
solve appropriate instances of the problem $L_{check}$ from Lemma \ref{sig1lemma} for the integers $m\in R_{i,d}(n)$. 
Now we want to give a bound on the cardinality of the set $R_{i,d}(n)$. 
For this purpose we will make use of the following auxiliary result.
\begin{lemma}\label{auxil_lemma}
If $n\geq 2^{c_{i,d}^2}$, then $\frac{\log (n)}{n^{1\slash (d\cdot h)}}\leq 1$.
\end{lemma}
\begin{proof}
We want
to make use of Lemma \ref{auxiliary-lemma}.
Since $n\geq 2^{c_{i,d}^2}$ and $c_{i,d}$ is defined as $2^{dh}(dh)^3\cdot 2^c\cdot i^c$, we have
\[n\:\geq\: 2^{(2^{dh}(dh)^32^ci^c)^2}=2^{2^{2dh}\cdot (dh)^6\cdot 2^{2c}\cdot i^{2c}}\]
We want to show that for $\alpha =4$ and $\beta =dh$, this implies that
\begin{equation}\label{eq_xxx_z}
n\geq 2^{\alpha (\beta +1)\cdot 2^{\beta}\cdot \beta^2}=2^{4(dh+1)\cdot 2^{dh}\cdot (dh)^2}.
\end{equation}
It is sufficient to show that
\begin{equation}\label{eq_xxxx_z}
2^{2dh}\cdot (dh)^6\cdot 2^{2c}\cdot i^{2c}\:\geq\: 4(dh+1)\cdot 2^{dh}\cdot (dh)^2,
\end{equation}
which is equivalent to $2^{dh}\cdot (dh)^4\cdot 2^{2c}\cdot i^{2c}\geq 4(dh+1)$. We have $2^{2c}\geq 4$ and $(dh)^2\geq dh+1$.
Thus (\ref{eq_xxxx_z}) holds, which implies that (\ref{eq_xxx_z}) holds as well.
Thus Lemma \ref{auxiliary-lemma} (ii) yields that $4\cdot\log (n)^{dh}\leq n$.
This implies that $n\geq (\log n)^{dh}$, which is equivalent to $\frac{\log (n)}{n^{1\slash (dh)}}\leq 1$.
\end{proof}

Now we can give a bound on the size of the set $R_{i,d}(n)$ as follows.
Since $n\geq 2^{c_{i,d}^2}$, we can apply Lemma \ref{auxil_lemma} and obtain that
\[I_{n,d}\: =\: \left (n^{1\slash h},\: n^{1\slash h}\left (1+\frac{\log (n)}{n^{1\slash (dh)}}\right )^d \right )\:\subseteq\: 
\left (n^{1\slash h},\: n^{1\slash h}\cdot 2^d\right ).\]
We have $R_{i,d}^l(n)\subseteq [L_l,R_l]$, where
$L_l=n^{1\slash h^l}$, $R_1=n^{1\slash h}\cdot 2^d$, $R_2=(n^{1\slash h}\cdot 2^d)^{1\slash h}\cdot 2^d$ and in general, 
\begin{eqnarray*}
R_l & = & 2^{d\cdot (1+\frac{1}{h}+\ldots +\frac{1}{h^{l-1}})}\cdot n^{1\slash h^l}\: =\:
        2^{d\cdot \frac{1-1\slash h^l}{1-1\slash h}}\cdot n^{1\slash h^l}\\
    & = & 2^{d\cdot (1-\frac{1}{h^l})\cdot\frac{h}{h-1}}\cdot  n^{1\slash h^l}
    \: \leq \: 2^{\frac{d\cdot h}{h-1}}\cdot n^{1\slash h^l}
\end{eqnarray*}
Moreover, if $R_l< 2^{c_{i,d}^2}$, then it follows that $R_{i,d}^{l+1}(n)=\emptyset$. 
Since we have shown that $R_l\leq  2^{\frac{d\cdot h}{h-1}}\cdot n^{1\slash h^l}$,
%
we obtain that $2^{\frac{d\cdot h}{h-1}}\cdot n^{1\slash h^l}<2^{c_{i,d}^2}$ implies $R_l<2^{c_{i,d}^2}$. Taking logarithms, we obtain
\begin{eqnarray*}
 & & \frac{1}{h^l}\cdot\log (n)+\frac{d\cdot h}{h-1}\: <\: c_{i,d}^2\\
 &\Leftrightarrow & h^l\: >\: \frac{1}{c_{i,d}^2-\frac{d\cdot h}{h-1}}\cdot\log (n)\\
 &\Leftrightarrow & l\: >\: \frac{1}{\log (h)}\left ( \log\log (n)-\log\left (c_{i,d}^2-\frac{d\cdot h}{h-1} \right )\right )
\end{eqnarray*}
So in particular, for $l>\log\log (n)$ we have $R_{i,d}^l(n)=\emptyset$. Since $R_{i,d}^l\subseteq [L_l,R_l]$, we have 
\[|R_{i,d}^l|\:\leq R_l-L_l+1\leq R_l\leq R_1.\]
Thus we obtain the following very rough bound on the cardinality of the set $R_{i,d}(n)$:
\[|R_{i,d}(n)|\:\:\leq\:\:2^{\frac{d\cdot h}{h-1}}\cdot n^{1\slash h}\cdot \log\log (n)\]
Now the running time for deciding the predicate $P(i,d,n,a,b)$ is dominated by the time needed to test for all integers $m\in R_{i,d}^l(n),1\leq l\leq \log\log (n)$ if
$\mbox{time}_{S_i}(m)>(a-l\cdot p_i)m^{b-l\cdot p_i}$. According to Lemma \ref{sig1lemma}, the time for each single test can be bounded by 
$c\cdot ((a-p_i)m^{b-p_i})^c$. Moreover, we have shown that $m\in R_{i,d}(n)=\bigcup_lR_{i,d}^l(n)$ implies that $m\leq 2^d\cdot n^{1\slash h}$. 
The number of tests is bounded by 
\[\sum_l|R_{i,d}^l(n)|\:\leq\: 2^{\frac{d\cdot h}{h-1}}\cdot n^{1\slash h}\cdot \log\log (n)\:\leq\: n^{2\slash h}\cdot\log\log (n).\]
Therefore, the time needed to solve all these instances of $L_{check}$ is bounded by
\begin{equation}\label{bound_term}
n^{2\slash h}\cdot\log\log (n)\cdot c\cdot\left (i\cdot (a-p_i)\cdot \left (n^{1\slash h}\cdot 2^d \right )^{b-p_i} \right )^c.
\end{equation}
In order to obtain an upper bound for the term in (\ref{bound_term}), we use Lemma \ref{auxil_lemma} and obtain
$\log (n)\leq n^{1\slash (d\cdot h)}$. We take the logarithm, apply again Lemma \ref{auxil_lemma} and obtain
\begin{equation}\label{hilf_1}
\log\log (n)\:\:\leq\:\:\frac{1}{dh}\cdot n^{1\slash (d\cdot h)}.
\end{equation}
Moreover, $c_{i,d}=2^{dh}\cdot (dh)^3\cdot 2^c\cdot i^c$ yields $i\leq c_{i,d}^{1\slash c}$. Since we have $n\geq 2^{c_{i,d}^2}\geq 2^{i^{2c}}$, and again using
Lemma \ref{auxil_lemma}, we obtain that
\begin{equation}\label{hilf_2}
i^{2c}\:\leq\: \log (n)\leq n^{1\slash (d\cdot h)},\:\:\mbox{i.e.}\:\: i^c\:\leq\: n^{1\slash (2dh)}.
\end{equation}
Since $a\leq\frac{\log (n)}{c_{i,d}}$, we obtain 
\begin{equation}\label{hilf_3}
a-p_i\:\leq\: a\:\leq\: \frac{\log (n)}{c_{i,d}}\:\leq\: n^{1\slash (dh)}.
\end{equation}
Finally, $n\geq 2^{c_{i,d}^2}$ directly implies that
\begin{equation}\label{hilf_4}
2^{d}\:\:\leq\:\: n^{1\slash (dh)}.
\end{equation}
Now we use (\ref{hilf_1})-(\ref{hilf_4}) in order get an upper bound for (\ref{bound_term}) and obtain that the instances of the problem $L_{check}$ in the 
computation of the predicate $P(i,d,n,a,b)$ can be solved in time 
\[\begin{array}{l@{\:\:}l@{\:\:}l}
         & n^{2\slash h}\cdot\log\log (n)\cdot c\cdot \left (i\cdot (a-p_i)\cdot \left (n^{1\slash h}\cdot 2^d\right )^{b-p_i} \right )^c & \\[0.8ex]
   \leq  & n^{2\slash h}\cdot  n^{1\slash (dh)}\cdot\frac{c}{dh} \cdot \left (i\cdot (a-p_i)\cdot \left (n^{1\slash h}\cdot 2^d\right )^{b-p_i} \right )^c 
         & \mbox{(using (\ref{hilf_1}))}\\[0.8ex]
   \leq  & n^{2\slash h}\cdot  n^{1\slash (dh)}\cdot\frac{c}{dh} \cdot n^{1\slash (2dh)}\cdot n^{c\slash (dh)}
           \cdot \left (n^{1\slash h}\cdot 2^d\right )^{c\cdot (b-p_i)} & \mbox{(using (\ref{hilf_2}) and (\ref{hilf_3}))}\\[0.8ex] 
   \leq  & n^{2\slash h}\cdot n^{1\slash (dh)}\cdot\frac{c}{dh}\cdot n^{1\slash (2dh)}\cdot 
           n^{\frac{c}{dh}}\cdot\left (n^{1\slash (dh)} \right )^{(b-p_i)\cdot c}\cdot n^{\frac{b-p_i}{h}\cdot c} & \mbox{(using (\ref{hilf_4}))}\\[0.8ex]
   \leq  & \left (n^{\frac{c}{dh}+\frac{c}{h}}\right )^{b-p_i}\cdot (n^{1\slash h})^{2+\frac{1}{d}+\frac{1}{2d}+\frac{c}{d}} & \mbox{(since $\frac{c}{dh}\leq 1$)}\\[0.8ex]
   \leq  & \left (n^{\frac{c}{dh}+\frac{c}{h}}\right )^{b-p_i}\cdot \left (n^{\frac{2}{h}+\frac{1}{dh}+\frac{1}{2dh}+\frac{c}{dh}} \right )^{b-p_i} & \\[0.8ex]
   \leq  & \left (n^{\frac{6c}{h}}\right )^{b-p_i} & \\[0.8ex]
   \leq  & n^{b-p_i}. & \mbox{(since $\frac{6c}{h}\leq 1$)}  
\end{array}\]
Altogether we obtain that the predicate $P(i,d,n,a,b)$ can be computed in time $n^{b-p_i}$. 

Now we turn back to the construction of the machines $\tilde{S}_{i,d}$. We consider the case when the input length $n$ satisfies $n\geq 2^{c_{i,d}}$.
In that case, the computation of the machine $\tilde{S}_{i,d}$ is split into time intervals
$T_0=\left (0,a_1n^{b_1}\right ],T_l=\left (a_ln^{b_l},a_{l+1}n^{b_{l+1}}\right ],1\leq l\leq L-1$ and $T_L=\left (a_{L}n^{b_L},\infty \right )$.
In each interval $T_l$ with $l<L$, the machine $\tilde{S}_{i,d}$ might have to solve the predicate $P(i,d,n,a_{l+1},b_{l+1})$, and we have to show that the
computation time within the interval $T_l$ is sufficient to do so. This follows now from the following lower bound on the size of the intervals $T_l$:
\begin{equation}\label{interval_bound}
\mbox{For $l=0,\ldots , L-1$, }\:\quad\quad |T_l|\geq n^{b_{l+1}}.
\end{equation} 
In order to prove (\ref{interval_bound}), we consider two cases. If $b_l=b_{l+1}$, then we have $a_{l+1}=a_l+1$, and therefore
$|T_l|=a_{l+1}n^{b_{l+1}}-a_ln^{b_l}=n^{b_{l+1}}$. Otherwise, if $b_{l+1}=b_l+1$, then we have $a_{l+1}=c_{i,d}>1$ and $a_l=\lfloor\frac{\log (n)}{c_{i,d}}\rfloor <n$, thus
$|T_l|=a_{l+1}n^{b_{l+1}}-a_ln^{b_l}>(c_{i,d}-1)n^{b_{l+1}}>n^{b_{l+1}}$. Thus we obtain that (\ref{interval_bound}) holds.

Thus the machine $\tilde{S}_{i,d}$ can simulate $S_i$ on the given input $x$ in such a way that 
$\mbox{time}_{\tilde{S}_{i,d}}(n)\leq\mbox{time}_{S_i}(n)$ and such that for each pair $a,b$, the machine $\tilde{S}_{i,d}$ continues the simulation for 
more than $a\cdot n^b$ steps only if $P(i,d,n,a,b)$ holds. Thus from the definition of the predicate $P(i,d,n,a,b)$ it follows that
the function $\mbox{time}_{\tilde{S}_{i,d}}$ satisfies Property $[\star ]$ with 
parameters $c_{i,d},d$ and $p_i=\lceil\frac{i}{2}\rceil$.
Moreover, if the function $\mbox{time}_{S_i}(n)$ already satisfies Property $[\star ]$ with parameters $c_{i,d},d,p_i$, then 
the machine $\tilde{S}_{i,d}$ completely simulates $S_i$ on every input, i.e. $L(\tilde{S}_{i,d})=L(S_i)$.
Altogether we have shown:
\begin{lemma}\label{progsys_lemma}
For all $i$ and $d$, the machine $\tilde{S}_{i,d}$ has the following properties.
\begin{itemize}
\item[(a)] For all $n$, $\mbox{time}_{\tilde{S}_{i,d}}(n)\leq\mbox{time}_{S_i}(n)$.
\item[(b)] The function $\mbox{time}_{\tilde{S}_{i,d}}(n)$ satisfies Property $[\star]$ with parameters $c_{i,d},d,p_i$.
\item[(c)] If the function $\mbox{time}_{S_i}(n)$ already satisfies Property $[\star ]$ with parameters $c_{i,d},d,p_i$, then
           $L(\tilde{S}_{i,d})=L(S_i)$.
\end{itemize}
\end{lemma}
Now we define the complexity classes associated to the new family $(\tilde{S}_{i,d})$ of $\Sigma_2$-machines in a standard way.
\begin{definition} $\mbox{ }$\\
For a function $t\colon {\mathbb N}\to {\mathbb N}$ we define the classes $\mbox{DTIM}\tilde{\mbox{E}}(t)$ and $\tilde{\Sigma}_2(t)$:
\[\begin{array}{l@{\:}c@{\:}l}
\mbox{DTIM}\tilde{\mbox{E}}(t) & = & \left\{L\: |\: \exists\:\mbox{$t(n)$-time bounded}\right.\\
                                  &   &  \quad\quad\quad\left.\mbox{deterministic machine $\tilde{S}_{i,d}$ with $L(\tilde{S}_{i,d})=L$}\right\}\\
\tilde{\Sigma}_2(t)            & = & \{L\: |\: \exists\:\mbox{$t(n)$-time bounded $\Sigma_2$-machine $\tilde{S}_{i,d}$ with $L(\tilde{S}_{i,d})=L$}\}
\end{array}\] 
Especially we define $\tilde{P}=\bigcup_{p(n)}\mbox{DTIM}\tilde{\mbox{E}}(p(n))$ and $\tilde{\Sigma}_{2}^p=\bigcup_{p(n)}\tilde{\Sigma}_2(p(n))$,
where the union goes over all polynomials $p(n)$.
\end{definition}
The next lemma shows that the resulting polynomial time classes are equal to the standard polynomial time classes $P$ and $\Sigma_2^p$ respectively. 
\begin{lemma}\label{classes_lemma}
We have $P=\tilde{P}$ and $\Sigma_2^p=\tilde{\Sigma}_2^p$.
\end{lemma}
\begin{proof}
Let $L\in \Sigma_2^p$. There exist 
a $\Sigma_2$-machine $S_i$ and some constant $q\in {\mathbb N}$ such that $L=L(S_i)$ and 
such that for each input $x$ of length $n$, $S_i$ makes precisely $qn^q$ computation steps on every computation path for inputs of length $n$. 
According to Lemma \ref{progsys_lemma}(c), it suffices to show that there exists some $d\in {\mathbb N}$ such that the function
$\mbox{time}_{S_i}(n)$ satisfies Property $[\star]$ with parameters $c_{i,d},p_i,d$. We will now actually show that this holds for all $d\in {\mathbb N}$.
So suppose that $n\geq 2^{c_{i,d}^2}$ and $c_{i,d}\leq a\leq b\leq\frac{\log (n)}{c_{i,d}}$ are such that 
$\mbox{time}_{S_i}(n)=qn^q>an^b$. It suffices to show that the following two conditions hold:
\begin{itemize}
\item[(i)] For all $m\in I_{n,d}$, $\mbox{time}_{S_i}(m)>(a-p_i)\cdot m^{b-p_i}$.
\item[(ii)] The number of integers in the interval $I_{n,d}$ which are not $h$-powers is at least $\left\lceil\frac{\log\log (n)}{c_{i,d}}\right\rceil$.
\end{itemize}
{\sl Concerning (i):} Since $\mbox{time}_{S_i}(m)=qm^q$ for all $m$, it suffices to show that $q\geq b$. So suppose for the contrary that $q<b$. 
Since $b\leq\frac{\log (n)}{c_{i,d}}\leq\log (n)$ and $n\geq 2^{c_{i,d}^2}$, using Lemma \ref{auxiliary-lemma}, this implies that $b\leq\log (n)\leq n^{1\slash (dh)}<n$.
But then we have $qn^q<n^{q+1}\leq an^b$, a contradiction. Thus we have $q\geq b$. \\[0.3ex]
{\sl Concerning (ii):} The number of integers $m\in I_{n,d}$ which are not $h$-powers is at least $\frac{1}{2}\cdot |I_{n,d}|$, since two consecutive integers cannot 
both be $h$-powers simultaneously. The size of the interval $I_{n,d}$ can be estimated as follows:
\begin{eqnarray*}
|I_{n,d}| & = & n^{1\slash h}\left (1+\frac{\log (n)}{n^{1\slash (d\cdot h)}}\right )^d\: -n^{1\slash h} -1\\
          & = & \left (n^{1\slash h}+\log (n)\cdot n^{\frac{1}{h}-\frac{1}{d\cdot h}}\right )\cdot \left (1+\frac{\log (n)}{n^{1\slash (d\cdot h)}}\right )^{d-1}
                -n^{1\slash h}-1\\
          & \geq & \log (n)\cdot n^{(1-\frac{1}{d})\frac{1}{h}}\\
          & >    &  2\cdot \left\lceil\frac{\log\log (n)}{c_{i,d}}\right\rceil\:\:\:\:\:\:
                    \mbox{(since $\log (n)>\log\log (n)$, $n^{(1-\frac{1}{d})\frac{1}{h}}>1$ and $c_{i,d}>4$).}
\end{eqnarray*}
Thus (ii) holds as well. 
Altogether we have shown that the function 
$\mbox{time}_{S_i}(n)$ satisfies Property $[\star ]$ with parameters $c_{i,d},p_i,d$. Hence, due to Lemma \ref{progsys_lemma},
$L(\tilde{S}_{i,d})=L(S_i)$. Moreover, if $S_i$ is a deterministic machine,
then $\tilde{S}_{i,d}$ is also a deterministic machine. This yields $P=\tilde{P}$ and $\Sigma_2^p=\tilde{\Sigma}_2^p$.
\end{proof}
\section{Construction of the Union Function}\label{constr_section}
Now we describe how the Union Function $F$ is constructed. The general approach is the same as in \cite{McCM69}. We have already given an outline of the construction 
in Section \ref{outline_section}. Here we will first briefly recall the notions and notations which we are making use of. Then we will give a detailed pseudo-code
description of the construction of $F$. Afterwards, we will prove in Lemma \ref{union_lemma} 
that $F$ is indeed a union function for $\Sigma_2^p=\tilde{\Sigma}_2^p$ with respect to the
family $(\tilde{S}_{i,d})$ of $\Sigma_2$-machines which we constructed in the preceeding section. Finally we will show in Lemma \ref{ineq_lemma} that $F$ satisfies the inequality 
$F(n^{1\slash h})^C\leq F(n)$ for all $h$-powers $n$ and that $F(n)$ can be computed deterministically in time $F(n)^C$ for some constant $C$, namely for $C=10c<\frac{1}{2}h$.

The function $F\colon {\mathbb N}\to {\mathbb N}$ is constructed in stages. In stage $n$ of the construction, the value $F(n)$ is defined.
Within the construction, we maintain a list ${\mathcal L}$ of \emph{guesses} $(\tilde{S}_{i,d},b_{i,d})$. We arrange the machines $\tilde{S}_{i,d}$
in a list $\tilde{S}_{(1)},\tilde{S}_{(2)},\ldots , \tilde{S}_{(n)},\ldots$ such that each machine $\tilde{S}_{i,d}$ occurs in this list. 
In order to guarantee that the union function $F(n)$ can be computed in time $F(n)^C$, we will construct the list $\tilde{S}_{(1)},\tilde{S}_{(2)},\ldots$
in such a way that when some machine $\tilde{S}_{i,d}$ is the $j$th machine in this list, then both the
%
associated constant $c_{i,d}$ and the machine index with respect to the original enumeration $(S_l)$ of $\Sigma_2$-machines we were starting from 
are sufficiently small, and $i,d$ and the machine index of $\tilde{S}_{(j)}$ with respect to the numbering $(S_l)$ can be computed efficiently from $j$ 
(conditions (i)-(iii) below). 

Let us describe this now in detail. First we note that for a given machine $\tilde{S}_{i,d}$, we can compute a machine index, say $k(i,d)$, of this machine 
with respect to the original enumeration $(S_l)$ of $\Sigma_2$-machines. This means that $\tilde{S}_{i,d}=S_{k(i,d)}$, and this function $k(\cdot,\cdot)$ is 
computable. Without loss of generality we assume that $k(1,d)=1$ for all $d$ - recall that $S_1$ is a machine which runs to infinity on every input.
Moreover, the parameter $c_{i,d}$ was defined as $c_{i,d}=2^{dh}\cdot (dh)^3\cdot 2^c\cdot i^c$, and the machine $\tilde{S}_{i,d}$ satisfies
Property $[\star ]$ with parameters $c_{i,d},d$ and $p_i=\lceil i\slash 2\rceil$. We construct now the list 
$\tilde{S}_{(1)},\tilde{S}_{(2)},\ldots , \tilde{S}_{(n)},\ldots$ in such a way that the following conditions hold.
\begin{itemize}
\item[(i)] If $\tilde{S}_{i,d}$ is the $j$th machine in this list, i.e. $\tilde{S}_{i,d}=\tilde{S}_{(j)}$, then 
           the associated constant $c_{i,d}$ satisfies $c_{i,d}\leq j$. 
\item[(ii)] If $\tilde{S}_{i,d}$ is the $j$th machine in the list, then the index $k(i,d)$ of this machine with respect to the original 
            enumeration $(S_l)$ of $\Sigma_2$-machines satisfies $k(i,d)\leq j$.
\item[(iii)] The parameters $i,d$ and $k(i,d)$ such that $\tilde{S}_{(j)}=\tilde{S}_{i,d}=S_{k(i,d)}$ can be computed from $j$ in time $O(j^2)$.
\end{itemize}
We construct the list $\tilde{S}_{(1)},\tilde{S}_{(2)},\ldots$ as follows. We let $\alpha\colon {\mathbb N}\times {\mathbb N}\to {\mathbb N}$ be a bijection such that
both $\alpha$ and its inverse are efficiently computable (precisely: in time polynomial in the bit-length of the input)
and such that $\alpha(1,1)=1$. Now we add the machines into the list
in the order given by the bijection $\alpha$, but in a delayed way such as to satisfy (i)-(iii):
We define $\tilde{S}_{(1)}=\tilde{S}_{1,1}$. Now if 
$\tilde{S}_{(1)},\ldots ,\tilde{S}_{(j-1)}$ are already constructed, then we spend at most $j^2-(j-1)^2$ computation steps to do the following:

Let $\tilde{S}_{(j-1)}=\tilde{S}_{i',d'}$. Compute $\alpha(i',d')=j'$ and compute $\alpha^{-1}(j'+1)=(i,d)$. Now compute the index $k(i,d)$ and 
check if $c_{i,d}\leq j$ and if $k(i,d)\leq j$.
If $j^2-(j-1)^2=2j-1$ computation steps are not sufficient to do this or if $c_{i,d}> j$ or $k(i,d)>j$, then we define $\tilde{S}_{(j)}:=\tilde{S}_{(j-1)}$.
Otherwise we define $\tilde{S}_{(j)}:=\tilde{S}_{i,d}=S_{k(i,d)}$. 
In this way, the initial part $\tilde{S}_{(1)},\ldots , \tilde{S}_{(t)}$ of the list can be constructed in $t^2$ steps, and (i)-(iii) are satisfied.

{\sl Notation.} If $\tilde{S}_{i,d}=\tilde{S}_{(j)}$, then we let $c_{(j)}$ and $p_{(j)}$ denote the values $c_{i,d}$ and $p_i$ respectively.

We have already introduced the intervals $I_t=[\delta_{t-1}+1,\delta_t]$ on which the $\log^*$ function is equal to $t$, i.e. 
with $\delta_0=0,\delta_1=2$ and $\delta_{t+1}=2^{\delta_t},t\geq 1$.
These intervals are a partition of ${\mathbb N}$, and we have $\log^*(n)=t$ for $n\in I_t$, where $\log^*(n)=\min\{t|2^{2^{\ldots 2}}|t\:\geq n\}$.

Recall that we have to construct the union function $F$ such as to achieve three things. First, $F(n)$ is supposed to be a union function for 
$\tilde{P}=\tilde{\Sigma}_2^p$, namely such that
\begin{equation}\label{union_eq_1}
\tilde{P}=DTIM\tilde{E}(F)=\tilde{\Sigma}_2(F)=\tilde{\Sigma}_2^p.
\end{equation}
Moreover, $F(n)$ is supposed to satisfy the padding inequality 
\begin{equation}\label{union_eq_pad}
F(n)^C\leq F(n^h)\quad\quad\mbox{for all $n$.}
\end{equation} 
Finally, in order to let the padding construction be consistent with 
the definition of Property $[\star]$, we also need to assure that 
\begin{equation}\label{union_eq_2}
\tilde{P}=DTIM\tilde{E}(F(n+1))=\tilde{\Sigma}_2(F(n+1))=\tilde{\Sigma}_2^p.
\end{equation}
Recall that in general a union function is constructed in terms of guesses $(\tilde{S}_{(i)},b_i)$ and selecting lexicographically smallest violated guesses
and diagonalizing against them. Now it is not difficult to satisfy the conditions (\ref{union_eq_1}) and (\ref{union_eq_2}) simultaneously: 
We just maintain in every stage two different kinds of guesses for every machine $\tilde{S}_{(i)}$. In stage $n$ of the construction, one of these guesses is 
checked for violation at input length $n$, and the other one at input length $n-1$.

Therefore, in stages $n\in I_t=[\delta_{t-1}+1,\delta_{t}]$, the list ${\mathcal L}$ of guesses consists of two sublists
${\mathcal L}_{1}$ and ${\mathcal L}_{2}$. Both sublists contain guesses for the machines
$\tilde{S}_{(1)},\ldots , \tilde{S}_{(t)}$. In stage $n$, guesses in the list ${\mathcal L}_{1}$ are tested for violations at input length $n$, and 
guesses in the sublist ${\mathcal L}_{2}$ are checked for violations at input length $n-1$ (i.e. if $\mbox{time}_{\tilde{S}_{(i)}}(n-1)>b_i(n-1)^{b_i}$) . 
The two guesses for a machine $\tilde{S}_{(i)}$ in list ${\mathcal L}_{1}$ and ${\mathcal L}_{2}$ are
treated independently. Each guess $(\tilde{S}_{(i)},b_i)$ has the property that $b_i\leq \log^*(n)=t$. 

The next detail in the construction guarantees that the padding inequality (\ref{union_eq_pad}) will be satisfied.
In the construction, we distinguish between stages $n$ such that $n$ is an $h$-power and stages $n$ such that $n$ is not an $h$-power. In the case when 
$n$ is an $h$-power, we select a smallest violated guess $(\tilde{S}_{(i)},b_i)$ from 
the list ${\mathcal L}_{1}\cup {\mathcal L}_{2}$ with respect to the following order: 
first ordered increasingly by the value $b_i$, then increasingly by the index of the sublist ($1$ or $2$) which they belong to, and then by the index $i$.
If a guess $(\tilde{S}_{(i)},b_i)\in {\mathcal L}_{1}$ is selected , we set $F(n)=n^{b_i}$ and replace 
$(\tilde{S}_{(i)},b_i)$ by the guess  $(\tilde{S}_{(i)},\log^*(n))$ in ${\mathcal L}_{1}$. If 
a guess $(\tilde{S}_{(i)},b_i)\in {\mathcal L}_{2}$ is selected, we define $F(n)$ as $(n-1)^{b_i}$ and 
replace the guess $(\tilde{S}_{(i)},\log^*(n))$ in the list ${\mathcal L}_{2}$. 
Note that since $S_1$ is a $\Sigma_2$-machine which runs to infinity on every input, 
the list ${\mathcal L}_n$ will always contain at least one violated guess.

In the case when $n$ is not an $h$-power, we proceed differently. We maintain two additional lists ${\mathcal L}'_1,{\mathcal L}'_2$ which are constructed as follows.
At the beginning of each interval $I_t$, we consider the lists ${\mathcal L}_1,{\mathcal L}_2$ and let
for $j=1,2$ the list ${\mathcal L}'_j$ consist of all guesses $(\tilde{S}_{(i)},b_i-p_{(i)})$ such that the guess $(\tilde{S}_{(i)},b_i)$ is contained in 
${\mathcal L}_j$. Now in stages $n$ such that $n$ is not an $h$-power, we select violated guesses of the form $(\tilde{S}_{(i)},b_i)$ or $(\tilde{S}_{(i)},b_i-p_{(i)})$ from the 
{\sl extended list} ${\mathcal L}_1\cup {\mathcal L}'_1\cup {\mathcal L}_2\cup {\mathcal L}'_2$,
namely with respect to the following order: First ordered increasingly by the value $b_i$, then by the index $j$ of the sublist (i.e. $1$ or $2$), then by the second entry
($b_i$ or $b_i-p_{(i)}$) of the guess and then by the machine index $i$.
The reason for this particular order will become clear in the proof 
of Lemma \ref{ineq_lemma}, where we show that the function $F$ satisfies the padding inequality $F(n)^C\leq F(n^h)$. 

Now the function value $F(n)$ and the list updates are defined as follows.
\begin{itemize}
\item If some guess $(\tilde{S}_{(i)},b_i)$ from ${\mathcal L}_1$ is selected, $F(n)$ is defined as $n^{b_i}$, and the guess is replaced by $(\tilde{S}_{(i)},\log^*(n))$ in 
${\mathcal L}_1$. 
\item If some guess $(\tilde{S}_{(i)},b_i)$ from ${\mathcal L}_2$ is selected, $F(n)$ is defined as $(n-1)^{b_i}$, and the guess is replaced by $(\tilde{S}_{(i)},\log^*(n))$ in
${\mathcal L}_2$. 
\item If a guess $(\tilde{S}_{(i)},b_i-p_{(i)})$
is selected from ${\mathcal L}'_1$, $F(n)$ is defined as $n^{b_i-p_{(i)}}$, the guess is removed from ${\mathcal L}'_1$ and the guess
$(\tilde{S}_{(i)},b_i)$ in the list  ${\mathcal L}_1$ is replaced by $(\tilde{S}_{(i)},\log^*(n))$. 
\item If a guess $(\tilde{S}_{(i)},b_i-p_{(i)})$ is selected from ${\mathcal L}'_{2}$, $F(n)$ is defined as $(n-1)^{b_i-p_{(i)}}$, the guess is 
removed from ${\mathcal L}'_2$ and the guess $(\tilde{S}_{(i)},b_i)$ in the list  ${\mathcal L}_2$ is replaced by $(\tilde{S}_{(i)},\log^*(n))$.
\end{itemize}
At the end of each stage $\delta_{t}$ (the last stage within an interval $I_t$, i.e. immediately before the value $\log^*(n)$ increases by $1$), 
a new machine $\tilde{S}_{(t+1)}$ enters the lists. Namely, the guess $(\tilde{S}_{(t+1)},t+1)$ is added to the list ${\mathcal L}_1$ and to the list ${\mathcal L}_2$. 
\begin{itemize}
\item[] {\sl Notation:} For $j=1,2$, let ${\mathcal L}_{n,j}$ denote the list ${\mathcal L}_j$ at the beginning of stage $n$ of the construction, and 
        let  ${\mathcal L}'_{n,j}$ denote the list ${\mathcal L}'_j$ at the beginning of stage $n$.
\end{itemize}
Directly from this construction it follows that the following invariants hold during this construction.\\
\begin{itemize} 
\item The maximum $b$-value $b_n^*$ in the list ${\mathcal L}_n={\mathcal L}_{n,1}\cup {\mathcal L}_{n,2}$ 
      at the beginning of stage $n$ satisfies $b_n^*=t=\log^*(n)$ for all $n\in I_t$.
\item Both lists ${\mathcal L}_{n,1}$ and ${\mathcal L}_{n,2}$ are of size $t=\log^*n$.
\item The extended list ${\mathcal L}_{n,1}\cup {\mathcal L}'_{n,1}\cup {\mathcal L}_{n,2}\cup {\mathcal L}'_{n,2}$ is of size at most $4\log^*(n)$.
\end{itemize}
Recall that in stage $n$, list ${\mathcal L}_{n,1}$ is used for diagonalization against violations at input length $n$ and 
the list ${\mathcal L}_{n,2}$ for diagonalization against violations at input length $n-1$.
This will guarantee that the conditions (\ref{union_eq_1}) and (\ref{union_eq_2}) holds. This condition, combined with the padding inequality (\ref{union_eq_pad}), will 
be crucial in the padding construction in the proof of Lemma \ref{padding_lemma}.
%
%
%
We are now ready to give a pseudocode description of the construction of the union function $F$.\\[1.2ex]
        {\bf\underline{Construction of $\mathbf{F(n)}$}}\\[0.6ex]
        {\bf Stage $\mathbf{1}$ (Initialization):}\\
        \hspace*{0.5cm} ${\mathcal L}_1:={\mathcal L}_2:={\mathcal L}'_1:={\mathcal L}'_2:=\{(\tilde{S}_{(1)},1)\}$\\
        \hspace*{0.5cm} $F(1):=1$\\[0.6ex]
        {\bf Stage $\mathbf{n>1}$:}\\
        \hspace*{0.5cm} Let $t\in {\mathbb N}$ such that $n\in I_t$, i.e. $t=\log^*n$.\\
        \hspace*{0.5cm} {\bf If} $n=\delta_{t-1}+1$ is the first stage in the interval $I_t$\\
        \hspace*{0.9cm} For $j=1,2$, let ${\mathcal L}'_j:=\:\{(\tilde{S}_{(i)},b_i-p_{(i)})\mid (\tilde{S}_{(i)},b_i)\in {\mathcal L}_j\}$\\
        \hspace*{0.5cm} {\bf If} $n$ is an $h$-power\\
        \hspace*{0.9cm} Select the smallest violated guess $(\tilde{S}_{(i)},b_{i})\in {\mathcal L}_1\cup {\mathcal L}_2$ with respect to the\\
        \hspace*{0.9cm} lexicographic order (first ordered by $b_i$, then by the list index $j=1$ or $2$ and then by $i$).\\
        \hspace*{0.9cm} {\bf If} $(\tilde{S}_{(i)},b_{i})$ is selected from ${\mathcal L}_1$\\
        \hspace*{1.1cm} Set $F(n):=n^{b_i}$ and replace $(\tilde{S}_{(i)},b_{i})$ by $(\tilde{S}_{(i)},\log^*(n))$ in ${\mathcal L}_1$\\
         \hspace*{0.9cm} {\bf else}\\
        \hspace*{1.1cm} Set $F(n):=(n-1)^{b_i}$ and replace $(\tilde{S}_{(i)},b_{i})$ by $(\tilde{S}_{(i)},\log^*(n))$ in ${\mathcal L}_2$\\
        \hspace*{0.5cm} {\bf else} $\slash\star$ $n$ is not an $h$-power $\star\slash$\\
        \hspace*{0.9cm} Select the smallest violated guess from ${\mathcal L}_1\cup {\mathcal L}'_1\cup {\mathcal L}_2\cup {\mathcal L}'_2$\\
        \hspace*{0.9cm} with respect to the following order:\\
        \hspace*{1.1cm} First ordered by $b_i$, then by the list index $j\in\{1,2\}$, then by the second entry\\
        \hspace*{1.1cm} (i.e. the entry $b_i$ or $b_i-p_{(i)}$ respectively), and then by machine index $i$\\
        \hspace*{0.9cm} {\bf If} a guess $(\tilde{S}_{(i)},b_{i}-p_{(i)})$ is selected from ${\mathcal L}'_{1}$\\
        \hspace*{1.3cm} Set $F(n):=n^{b_{i}-p_{(i)}}$\\
        \hspace*{1.3cm} Replace $(\tilde{S}_{(i)},b_{i})$ by $(\tilde{S}_{(i)},\log^*(n))$ in ${\mathcal L}_1$\\
        \hspace*{1.3cm} Remove $(\tilde{S}_{(i)},b_{i}-p_{(i)})$ from ${\mathcal L}'_1$\\
        \hspace*{0.9cm} {\bf else if } a guess $(\tilde{S}_{(i)},b_{i}-p_{(i)})$ is selected from ${\mathcal L}'_{2}$\\
        \hspace*{1.3cm} Set $F(n):=(n-1)^{b_{i}-p_{(i)}}$, replace $(\tilde{S}_{(i)},b_{i})$ by $(\tilde{S}_{(i)},\log^*(n))$ in ${\mathcal L}_2$\\
        \hspace*{1.3cm} and remove $(\tilde{S}_{(i)},b_{i}-p_{(i)})$ from ${\mathcal L}'_{2}$\\
        \hspace*{0.9cm} {\bf else} let $(\tilde{S}_{(i)},b_{i})$ be the guess which is selected, say from list ${\mathcal L}_j$.\\
        \hspace*{1.3cm} {\bf If} $j=1$, set $F(n):=n^{b_{i}}$ and replace $(\tilde{S}_{(i)},b_{i})$ by $(\tilde{S}_{(i)},\log^*(n))$ in ${\mathcal L}_1$.\\
        \hspace*{1.3cm} {\bf else} set $F(n):=(n-1)^{b_{i}}$ and replace $(\tilde{S}_{(i)},b_{i})$ by $(\tilde{S}_{(i)},\log^*(n))$ in ${\mathcal L}_2$.\\
        \hspace*{0.5cm} {\bf If} $n=\delta_{t}$ is the last stage in the interval $I_t$\\
        \hspace*{0.9cm} {\bf Add} the guess $(\tilde{S}_{(t+1)},t+1)$ both to the list ${\mathcal L}_1$ and to the list ${\mathcal L}_2$.\\
        {\bf End of Stage $\mathbf{n}$}\\

\begin{lemma}\label{union_lemma}
The function $F$ satisfies 
\[P=\tilde{P}=\mbox{DTIM}\tilde{\mbox{E}}(F)=\mbox{DTIM}\tilde{\mbox{E}}(F(n+1))=\tilde{\Sigma}_2(F(n+1))=\tilde{\Sigma}_2(F)=\tilde{\Sigma}_2^p=\Sigma_2^p.\]
\end{lemma}
\begin{proof} 
It is sufficient to prove that the equations $\tilde{\Sigma}_2(F)=\tilde{\Sigma}_2(F(n+1))=\tilde{\Sigma}_2^p$ 
and $\mbox{DTIM}\tilde{\mbox{E}}(F)=\mbox{DTIM}\tilde{\mbox{E}}(F(n+1))=\tilde{P}$ hold, since we have already shown 
in Lemma \ref{classes_lemma} that
$P=\tilde{P}=\tilde{\Sigma}_2^p=\Sigma_2^p$ holds. Now, in order to show $\tilde{\Sigma}_2(F)=\tilde{\Sigma}_2(F(n+1))=\tilde{\Sigma}_2^p$ 
and $\mbox{DTIM}\tilde{\mbox{E}}(F)=\mbox{DTIM}\tilde{\mbox{E}}(F(n+1))=\tilde{P}$, 
it is sufficient to 
show that for each machine $\tilde{S}_{(i)}$, the following three properties (I), (II) and (III) are equivalent:
\begin{itemize}
\item[(I)] $\tilde{S}_{(i)}$ is polynomially time bounded, i.e. there exists a constant $a$ such that $\mbox{time}_{\tilde{S}_{(i)}}(n)\leq a\cdot n^a$ for all $n$. 
\item[(II)] There exists a constant $b$ such that $\mbox{time}_{\tilde{S}_{(i)}}(n)\leq b\cdot F(n)$ a.e.
\item[(III)] There exists a constant $b$ such that $\mbox{time}_{\tilde{S}_{(i)}}(n)\leq b\cdot F(n+1)$ a.e.
\end{itemize}
Let us first show the implication (I)$\Rightarrow$(II). In order to show that for every machine  $\tilde{S}_{(i)}$, (I) implies (II), it suffices to show
that 
\begin{equation}\label{majo_eq}
\mbox{for all $a\in {\mathbb N}$, $F(n)\:\geq\: a\cdot n^a$ a.e.}
\end{equation}
which means that the union function $F(n)$ majorizes every polynomial. So let us show that (\ref{majo_eq}) holds. 
Let $a\in {\mathbb N}$. We have to show that $F(n)\geq a\cdot n^a$ almost everywhere.
Recall that ${\mathcal L}_n={\mathcal L}_{n,1}\cup {\mathcal L}_{n,2}$ denotes the list of guesses ${\mathcal L}={\mathcal L}_1\cup {\mathcal L}_2$ 
at the beginning of stage $n$. 
We let ${\mathcal L}'_n={\mathcal L}'_{n,1}\cup {\mathcal L}'_{n,2}$. 
Moreover, $b_n^*=\log^*(n)$ is the maximum $b$-value in the list ${\mathcal L}_n$.
In the construction of the function $F$, the value $b_n^*$ is increased by $1$ at every stage in which a new guess enters the list. 
For a given integer $a$, let $n_a$ be the smallest integer such that $b_n^*\geq 2(a+1)$. The list ${\mathcal L}_{n_a}\cup {\mathcal L}'_{n_a}$ is finite. Thus the set of guesses 
in ${\mathcal L}_{n_a}\cup {\mathcal L}'_{n_a}$ which will ever be violated and selected in some stage $m\geq n_a$ is finite. 
Say at some stage $n_1\geq n_a$, the last such guess is selected in the construction
of the function $F$. This means that from that stage on, i.e. for all stages $m\geq n_1+1$, 
the guess - say $(\tilde{S}_{(i)},b_{i})$ or $(\tilde{S}_{(i)},b_{i}-p_{(i)})$ -  which is selected in stage $m$
satisfies $b_{i}>b_i-p_{(i)}\geq a+1$. Then we have $F(m)\geq (m-1)^{b_{i}-p_{(i)}}\geq (m-1)^{a+1}$. Moreover, for $m$ being sufficiently large, 
we have $(m-1)^{a+1}>a\cdot m^a$, which yields that for $m$ being sufficiently large, $F(m)>a\cdot m^a$. Thus we have shown that (I) implies (II).\\[0.3ex] 

Now we show the implication (II)$\Rightarrow$(I), namely by showing that $\neg$(I) implies $\neg$(II). 
In stage $n$ in the construction of the function $F$, the function value $F(n)$ is determined based on the list ${\mathcal L}_n$ 
in case when $n$ is an $h$-power. Otherwise, if 
$n$ is not an $h$-power, $F(n)$ is determined based on the extended list 
\[{\mathcal L}_n\cup {\mathcal L}'_n\: =\: {\mathcal L}_{n,1}\cup  {\mathcal L}'_{n,1}\cup {\mathcal L}_{n,2}\cup  {\mathcal L}'_{n,2}.\] 
Let us denote by ${\mathcal L}(n)$ the list of guesses which are taken into account in stage $n$ of the construction, i.e. 
${\mathcal L}(n)$ is defined as ${\mathcal L}_n$ in case when $n$ is an $h$-power, and ${\mathcal L}(n)={\mathcal L}_n\cup {\mathcal L}'_n$ in case when $n$ is not an $h$-power.
So we can say for every stage $n$ that $F(n)$ is determined based on the list ${\mathcal L}(n)$.
In the construction of the union function $F(n)$, the guesses in ${\mathcal L}(n)$ are linearly ordered in the following way: 
First guesses are ordered increasingly by $b_i$, then by the list index $j\in\{1,2\}$, then by the second entry ($b_i$ or $b_i-p_{(i)}$) and then by the index $i$.
In each stage $n$, the smallest violated guess from ${\mathcal L}(n)$ with respect to this order is selected. 
Note that since $S_1$ and therefore also $\tilde{S}_{(1)}$ is a machine which runs to infinity on every 
input, in each stage $n$ of the construction there is at least one guess in the list which is violated at stage $n$.

Now suppose $\tilde{S}_{(j)}$ is a $\Sigma_2$-machine whose running time is not polynomially bounded, i.e.
such that for every $a\in {\mathbb N}$, $\mbox{time}_{\tilde{S}_{(j)}}(n)>a\cdot n^a$ infinitely often.
This directly implies that for every $b_{j}\in {\mathbb N}$, the guesses $(\tilde{S}_{(j)},b_{j})$ and $(\tilde{S}_{(j)},b_{j}-p_{(j)})$ are violated infinitely often.
For each such $b_{j}$, the number of guesses
$(\tilde{S}_{(l)},b_{l})$ which eventually occur in the extended list ${\mathcal L}\cup {\mathcal L'}$ and are smaller than $(\tilde{S}_{(j)},b_{j})$ 
or $(\tilde{S}_{(j)},b_{j}-p_{(j)})$ is finite.
Whenever such a guess is violated and selected in the construction of $F$ at some stage $n$ such that ${\mathcal L}(n)$ already contains the 
guess $(\tilde{S}_{(j)},b_{j})$, it is replaced by a guess $(\tilde{S}_{(l)},b_n^*)=(\tilde{S}_{(l)},\log^*(n))$. 
Since $b_n^*=\log^*(n)$ is a monotone increasing unbounded function of $n$,
after a finite number of stages, 
the guess $(\tilde{S}_{(j)},b_{j})$ is the smallest guess in the list ${\mathcal L}\cup {\mathcal L}'$ for which one of the guesses 
$(\tilde{S}_{(j)},b_{j}),\: (\tilde{S}_{(j)},b_{j}-p_{(j)})$ will ever be violated and selected again. Furthermore we will have
$b_{j}<b_n^*=\log^*(n)$ if $n$ is sufficiently large.
From that point on, whenever one of the guesses 
$(\tilde{S}_{(j)},b_{j}),\: (\tilde{S}_{(j)},b_{j}-p_{(j)})$ is violated again 
and is contained in the list ${\mathcal L}\cup {\mathcal L}'$, 
say in some stage $p$, 
it will be selected in the construction of $F$, which means that $F(p)=p^{b}$ and
$\mbox{time}_{\tilde{S}_{(j)}}(p)>b\cdot p^{b}$ for some $b\in\{b_j,b_j-p_{(j)}\}$. 
Then the guess $(\tilde{S}_{(j)},b_{j})$ will be replaced by some guess $(\tilde{S}_{(j)},b_p^*)$ with 
$b_p^*\geq b_{j}+1$. Since by assumption, every guess for $\tilde{S}_{(j)}$ is violated infinitely often, this yields a monotone increasing unbounded sequence of
integers 
$b_{j,1}<b_{j,2}<b_{j,3}\leq\ldots$ such that each guess $(\tilde{S}_{(j)},b_{j,l})$ will eventually be in the list of guesses ${\mathcal L}\cup {\mathcal L}'$ 
and such that this guess will eventually be selected in the construction of $F$, say in some stage $p_l$. Since for all $l\in {\mathbb N}$,
$F(p_l)=p_l^{b_{j,l}}$ and $\mbox{time}_{\tilde{S}_j}(p_l)>b_{j,l}\cdot p_l^{b_{j,l}}$, there cannot exist any constant $b$ such that 
$\mbox{time}_{\tilde{S}_{(j)}}(n)\leq b\cdot F(n)$ for almost all $n$. Thus we have shown that $\neg$(I) implies $\neg$(II).\\[0.6ex]
Since we have already shown that the function $F(n)$ (and therefore also the function $n\mapsto F(n+1)$) majorizes every polynomial, this immediately yields that 
(I) implies (III). Now the proof that $\neg$(I) implies $\neg$(III) is basically the same as the proof of $\neg$(I)$\Rightarrow\neg$(II):
Suppose that $\tilde{S}_{(j)}$ is a machine for which (I) does not hold, i.e. the running time of $\tilde{S}_{(j)}$ is not polynomially bounded. Then for every
integer $b_j$, both guesses $(\tilde{S}_{(j)},b_j)$ and $(\tilde{S}_{(j)},b_j-p_{(j)})$ will be violated infinitely often, i.e. for infinitely many integers $n$ 
they will be violated at length $n-1$. For every such $b_j$ such that $(\tilde{S}_{(j)},b_j)$ is eventually contained in the list ${\mathcal L}_2$, 
only finitely many guesses in the list ${\mathcal L}\cup {\mathcal L}'$ will have a higher priority of being selected than the guess $(\tilde{S}_{(j)},b_j)$ or 
$(\tilde{S}_{(j)},b_j-p_{(j)})$ respectively. From some input length on, whenever such a guess is selected, its $b$-value will be updated to a value greater than 
$b_{(j)}$. Thus after finitely many stages, the following holds: Whenever $(\tilde{S}_{(j)},b_j)$ or $(\tilde{S}_{(j)},b_j-p_{(j)})$ is violated again at some input length $n-1$,
then in stage $n$ of the construction of the union function,
it will have the highest priority and therefore be selected from the list ${\mathcal L}_2$ or ${\mathcal L}'_2$. 
If $(\tilde{S}_{(j)},b_j)$ is selected from list ${\mathcal L}_2$ in stage $n$, then 
this means that $\mbox{time}_{\tilde{S}_{(j)}}(n-1)>b_j\cdot (n-1)^{b_j}$, and the value of the union function is defined as 
$F(n)=(n-1)^{b_j}$. 
If $(\tilde{S}_{(j)},b_j-p_{(j)})$ is selected from the list ${\mathcal L}'_2$ in stage $n$, then the value of the union function 
is $F(n)=(n-1)^{b_j-p_{(j)}}$, while the running time of the machine $\tilde{S}_{(j)}$ at input length $n-1$ satisfies $\mbox{time}_{\tilde{S}_{(j)}}(n-1)>(b_j-p_{(j)})(n-1)^{b_j-p_{(j)}}$.
In both cases we have $\mbox{time}_{\tilde{S}_{(j)}}(n-1)>(b_j-p_{(j)})F(n)$, and the sequence of values $b_j-p_{(j)}$ is monotone increasing and unbounded. Thus
the condition (III) does not hold. 

This concludes the proof of the lemma. 
\end{proof}
Now we are going to show that the function value $F(n)$ can be computed in time polynomial in $F(n)$. This means that there exists 
a constant $C$ such that the function $F(n)^C$ is time constructible. In particular we show that this holds for $C=10c$, where $c$ is the constant from 
Lemma \ref{sig1lemma}. Below we will then show that $F$ also satisfies the inequality
$F(n^{1\slash h})^C\leq F(n)$ for every $h$-power $n$. This allows us to apply a padding technique in order to show that 
we also have $\mbox{DTIM}\tilde{\mbox{E}}(F^{C^2})=\tilde{\Sigma}_2(F^{C^2})$, which will yield a contradiction. 
\begin{lemma}\label{F_computation_lemma}
There is a constant $C$ such that $F(n)$ can be computed deterministically in time $F(n)^C$. More precisely, there exists a deterministic algorithm which gets 
as an input the integer $n$ and computes the function value $F(n)$ in at most $F(n)^C$ steps.
\end{lemma}
\begin{proof} We describe an algorithm which computes the function value $F(n)$ for a given $n$. Let us first give some intuition.
In order to compute the function value $F(n)$ we first have to compute the function values $F(1),\ldots, F(n-1)$, 
or at least the lists of guesses ${\mathcal L}(1),\ldots, {\mathcal L}(n-1)$. Recall that for every $m$, the list ${\mathcal L}(m)$ was defined as ${\mathcal L}_m$ in case when $m$ is an 
$h$-power, and as ${\mathcal L}_m\cup {\mathcal L}'_m$ otherwise. 
If we would compute $F(n)$ just directly along
the definition of $F$, the following problem would occur. 
It might happen that $F(n)=n^b$ for some integer $b$, but some of the previous function values have a much larger exponent, for instance
$F(m)=m^B$ for some $B\gg b$. In order to compute $F(n)$, we would first
compute $F(1),\ldots , F(n-1)$, which might then take time $\approx n^B$ in order to compute the function value $n^b$. 
In this way, we would not be able to compute $F(n)$ within time 
polynomial in $F(n)$.

The idea how to circumvent this obstacle and to compute $F(n)$ in time $F(n)^{O(1)}$ is now as follows: 
We choose some integer number $b$ and simulate the computation of the function $F(n)$, 
but within this simulation we cut off all the guesses in the lists at $b+1$. i.e. 
replace all the values $b_i,b_i-p_{(i)}$ which occur in the lists ${\mathcal L}(n)$
by the values $\min\{b_i,b+1\}$ and $\min\{b_i-p_{(i)},b+1\}$. 
We denote the function which is computed by this simulation as $F_b(n)$.
It will turn out that this function $F_b(n)$ can be computed in time $n^{O(b)}$.
Moreover, we will show that the functions $F(n)$ and $F_b(n)$ are related as follows:
\begin{itemize}
\item[] If $F(n)=n^a$, then $F_b(n)=F(n)$ for all $b\geq 3a$. \\If $F_b(n)=n^a$ with $a\leq \frac{b}{3}$, then $F(n)=F_b(n)$.
\end{itemize}
This gives then the following method for computing the function $F(n)$.
We compute $F_b(n)$ for increasing values of $b$ until we find the smallest $b$ for which $F_b(n)\leq n^{b\slash 3}$. For this $b$, we 
will then know that the function value is correct, i.e. we have $F(n)=F_b(n)$.

Let us now describe this method in detail.
Recall that the list ${\mathcal L}_n={\mathcal L}_{n,1}\cup {\mathcal L}_{n,2}$ contains guesses of the form $(\tilde{S}_{(i)},b_i)$,
and the list ${\mathcal L}'_n={\mathcal L}'_{n,1}\cup {\mathcal L}'_{n,2}$ contains guesses of the form $(\tilde{S}_{(i)},b_i-p_{(i)})$.
%
If $n\in I_t=[\delta_{t-1}+1,\delta_t]$, then the largest $b_i$ which occurs in the list ${\mathcal L}_n$ is denoted as $b_n^*$ and satisfies $b_n^*=t=\log^*(n)$.
Now we let $F_b(n)$ be the function which is computed by the modification of the algorithm for $F$ where all guesses $(\tilde{S}_{(i)},b_i)$ are
replaced by guesses $(\tilde{S}_{(i)},\min\{b_i,b+1\})$ and the guesses $(\tilde{S}_{(i)},b_i-p_{(i)})$ 
are replaced by $(\tilde{S}_{(i)},\min\{b_i-p_{(i)},b+1\})$. In the pseudocode description of the function $F_b$ below we will use the following notation.
\begin{itemize}
\item The lists of guesses are denoted as ${\mathcal L}^b={\mathcal L}_1^b\cup {\mathcal L}_2^b$ and ${\mathcal L}^{'b}={\mathcal L}_1^{'b}\cup{\mathcal L}_2^{'b}$. 
\item Guesses in the list ${\mathcal L}^b$ are denoted as $(\tilde{S}_{(i)},\beta_i)$, where $\beta_i$ denotes the value $\min\{b_i,b+1\}$.
\item Guesses in the list ${\mathcal L}^{'b}$ are denoted as $(\tilde{S}_{(i)},\pi_i)$, where 
      $\pi_i$ denotes\\ the value $\min\{b_i-p_{(i)},b+1\}$.
\end{itemize}
Now $F_b$ is the function which is computed by the following algorithm.\\
$\mbox{ }$\\
        {\bf\underline{Construction of $\mathbf{F_b(n)}$}}\\[0.6ex]
        {\bf Stage $\mathbf{1}$ (Initialization):}\\
        \hspace*{0.5cm} $F_b(1):=1,\:\:{\mathcal L}_1^b={\mathcal L}_2^b:={\mathcal L}_1^{'b}:={\mathcal L}_2^{'b}:=\{(\tilde{S}_{(1)},\beta_{1})\}$\\[0.6ex]
        {\bf Stage $\mathbf{n>1}$:}\\
        \hspace*{0.5cm} Let $t\in {\mathbb N}$ such that $n\in I_t$.\\
        \hspace*{0.5cm} {\bf If} $n=\delta_{t-1}+1$ is the first stage in the interval $I_t$\\
        \hspace*{0.9cm} For $j=1,2$ let ${\mathcal L}_j^{'b}:=\{(\tilde{S}_{(i)},\pi_i)\mid (\tilde{S}_{(i)},\beta_i)\in {\mathcal L}^b_j,\: \pi_i=\beta_i-p_{(i)}\}$\\
        \hspace*{0.5cm} {\bf If} $n$ is an $h$-power\\ 
        \hspace*{0.9cm} Let $(\tilde{S}_{(i)},\beta_{i})$ be the smallest violated guess in ${\mathcal L}_1^b\cup {\mathcal L}_2^b$\\
        \hspace*{0.9cm} (in lexicographic order, first ordered by $\beta_i$, then by the list index $j=1$ or $2$, then by $i$)\\
        \hspace*{0.9cm} {\bf If} $j=1$\\
        \hspace*{1.1cm} Set $F_b(n):=n^{\beta_i}$\\
        \hspace*{1.1cm} Replace $(\tilde{S}_{(i)},\beta_{i})$ by $(\tilde{S}_{(i)},\min\{\log^*(n),b+1\})$ in ${\mathcal L}_1^b$\\
        \hspace*{0.9cm} {\bf else}\\ 
        \hspace*{1.1cm} Set $F_b(n):=(n-1)^{\beta_i}$\\
        \hspace*{1.1cm} Replace $(\tilde{S}_{(i)},\beta_{i})$ by $(\tilde{S}_{(i)},\min\{\log^*(n),b+1\})$ in ${\mathcal L}_2^b$\\
        \hspace*{0.5cm} {\bf else} $\slash\star$ $n$ is not an $h$-power $\star\slash$\\
        \hspace*{0.9cm} Select the smallest violated guess from ${\mathcal L}_1^b\cup{\mathcal L}_1^{'b}\cup {\mathcal L}_2^b\cup {\mathcal L}_2^{'b}$\\
        \hspace*{0.9cm} with respect to the following order: first ordered by $\beta_i$, then by\\
        \hspace*{0.9cm} the list index $j$, then by the second entry $\beta_i$ or $\pi_i=\beta_i-p_{(i)}$ and then by $i$)\\
        \hspace*{0.9cm} {\bf If} a guess $(\tilde{S}_{(i)},\pi_i)$ is selected from ${\mathcal L}_j^{'b}$\\
        \hspace*{1.3cm} Set $F_b(n):=n^{\beta_{i}-p_{(i)}}$ in case $j=1$, $F_b(n):=(n-1)^{\beta_{i}-p_{(i)}}$ in case $j=2$\\
        \hspace*{1.3cm} Replace $(\tilde{S}_{(i)},\beta_{i})$ by $(\tilde{S}_{(i)},\min\{\log^*(n),b+1\})$ in ${\mathcal L}_j^b$.\\
        \hspace*{1.3cm} Remove  $(\tilde{S}_{(i)},\pi_i)$ from ${\mathcal L}_j^{'b}$\\
        \hspace*{0.9cm} {\bf else} \\
        \hspace*{1.3cm} Let $(\tilde{S}_{(i)},\beta_{i})$ be the guess which is selected, say from list ${\mathcal L}_j^b$.\\
        \hspace*{1.3cm} Set $F_b(n):=n^{\beta_i}$ for $j=1$, $(n-1)^{\beta_i}$ for $j=2$.\\
        \hspace*{1.3cm} Replace $(\tilde{S}_{(i)},\beta_i)$ by $(\tilde{S}_{(i)},\min\{\log^*(n),b+1\})$ in ${\mathcal L}_j^b$.\\
        \hspace*{0.5cm} {\bf If} $n=\delta_{t}$ is the last stage in the interval $I_t$\\
        \hspace*{0.9cm} {\bf Add} the guess $(\tilde{S}_{(t+1)},\beta_{t+1})$ to the list ${\mathcal L}_1^b$ and to ${\mathcal L}_2^b$\\
        \hspace*{1.2cm} both with $\beta_{t+1}=\min\{t+1,b+1\}$\\
        {\bf End of Stage $\mathbf{n}$}\\

$\mbox{ }$\\
So as to recall, in the computation of the union function $F(n)$ we maintain a list ${\mathcal L}={\mathcal L}_1\cup {\mathcal L}_2$ of guesses
of the form $(\tilde{S}_{(i)},b_i)$, together with the list ${\mathcal L}'={\mathcal L}'_1\cup {\mathcal L}'_2$, that contains
guesses of the form $(\tilde{S}_{(i)},b_i-p_{(i)})$. Now in the computation of the function $F_b(n)$ we maintain a list 
${\mathcal L}^b={\mathcal L}_1^b\cup {\mathcal L}_2^b$ which contains guesses $(\tilde{S}_{(i)},\beta_i)$. The values $\beta_i$ are generated as follows.
When a new guess with a machine $\tilde{S}_{(i)}$ enters the list at the end of a stage $n-1$, instead of the original value $b_i=\log^*n$ it gets 
the value $\beta_i=\min\{b_i,b+1\}=\min\{\log^*n,b+1\}$. 
Moreover, the lists ${\mathcal L}_1^{'b}$ and ${\mathcal L}_2^{'b}$ contain the guess $(\tilde{S}_{(i)},\pi_i)$, where 
$\pi_i= \min\{\log^*(n)-p_{(i)},b+1\}$. 
When a guess $(\tilde{S}_{(i)},b_i)$ or $(\tilde{S}_{(i)},b_i-p_{(i)})$ is selected
in the construction of $F(n)$. it is replaced by $(\tilde{S}_{(i)},\log^*n)$. In the computation of $F_b(n)$, 
a guess $(\tilde{S}_{(i)},\beta_i)$ is now replaced by $(\tilde{S}_{(i)},\min\{\log^*n,b+1\})$.

Let us call a guess $(\tilde{S}_{(i)},\beta_i)$ correct in ${\mathcal L}_n^b$ if $\beta_i<b+1$. Let us now take a look at a stage $n$ in the computation 
of the function $F_b$. Violated guesses in the lists ${\mathcal L}_n^b$ and ${\mathcal L}_n^{'b}$ which are correct always have a higher priority of being selected
than violated guesses whose value is equal to $b+1$. This is just due to the fact that we always give higher priority to guesses with smaller value.
Therefore, the following invariants are maintained in the computation of $F_b(n)$: 
\begin{itemize}
\item For every machine $\tilde{S}_{(i)}$ and every integer $\beta<b+1$, we have that 
      $(\tilde{S}_{(i)},\beta)$ enters the list ${\mathcal L}$ in stage $n$ of the computation of the function $F$ iff
      $(\tilde{S}_{(i)},\beta)$ enters the list ${\mathcal L}^b$ in stage $n$ of the computation of the function $F_b$. 
\item For every machine $\tilde{S}_{(i)}$, every $j\in\{1,2\}$ and every integer $\beta<b+1$,
      $(\tilde{S}_{(i)},\beta)$ is selected from list ${\mathcal L}_j$ in stage $n$ in the computation of the function $F$ iff
      $(\tilde{S}_{(i)},\beta)$ is selected from list ${\mathcal L}_j^b$ in stage $n$ in the computation of the function $F_b$.   
      Moreover, $\pi_i=\beta_i-p_{(i)}$, and $(\tilde{S}_{(i)},\beta_i-p_{(i)})$ is selected from list ${\mathcal L}'_j$ in stage $n$ 
      in the computation of the function $F$ iff
      $(\tilde{S}_{(i)},\pi_i)$ is selected in stage $n$ in the computation of the function $F_b$.
\end{itemize}
Thus selections of guesses $(\tilde{S}_{(i)},\beta_i)$ or $(\tilde{S}_{(i)},\pi_i)$ with the guess $(\tilde{S}_{(i)},\beta_i)$ being correct, i.e. $\beta_i<b+1$
are always the same in the computation of $F$ and $F_b$.
The value $p_{(i)}$ satisfies $p_{(i)}\leq\lceil \frac{i}{2}\rceil$. Furthermore, since the machine $\tilde{S}_{(i)}$ enters the list ${\mathcal L}$ within the guess
$(\tilde{S}_{(i)},i)$, we have $b_i\geq i$, which yields $p_{(i)}\leq \frac{b_i}{2}+1$
and therefore $b_{i}-p_{(i)}\geq \frac{b_{i}}{2}-1\geq\frac{b_{i}}{3}$.
Thus we obtain:
\begin{itemize}
\item If $F_b(n)=n^a$ with $b\geq 3a$, then $F(n)=F_b(n)$.
\item If $F(n)=n^a$, then for $b\geq 3a$ we have $F_b(n)=F(n)$.
\end{itemize}
Thus we 
can compute $F(n)=n^a$ by computing $F_1(n),\ldots , F_{3a}(n)$.
A single value $F_b(n)$ can be computed by running the algorithm for the function $F_b$ up to stage $n$. In each stage, 
at most $4\cdot\log^*(n)$ guesses of the form $(\tilde{S}_{(i)},\beta_i)$ with $i\leq\log^*(n)$ and $\beta_i\leq b+1$ have to be checked for 
violation at input length at most $n$. Testing a guess $(\tilde{S}_{(i)},\beta_i)$ for violation at some input length $m$ means to solve the associated instance
of the decision problem $L_{check}$ in Lemma \ref{sig1lemma}. 
Since we have constructed the list $\tilde{S}_{(1)},\tilde{S}_{(2)},\tilde{S}_{(3)},\ldots$ in such a way that 
for each $j\in {\mathbb N}$, $\tilde{S}_{(j)}=S_{k}$ for some $k\leq j$ and such that this index $k$ with respect to the original numbering $(S_i)$ of
$\Sigma_2$-machines can be computed from $j$ in time $O(j^2)$, we obtain directly from Lemma \ref{sig1lemma} that every single test of a guess $(\tilde{S}_{(i)},\beta_i)$
for violation at input length $m\leq n$
can be solved deterministically in time at most $c\cdot (i\cdot \beta_im^{\beta_i})^c\leq c\cdot (\log^*(n)\cdot  b\cdot n^{b})^c$.
In each stage $m\leq n$, the number of guesses which have to be taken into account in stage $m$ is bounded by $4\cdot \log^*(m)\leq 4\cdot\log^*(n)$.
The time needed to compute the function value $F_b(n)$ is dominated by the time for testing the guesses for violations, which is bounded by
\[n\cdot 4\log^*(n)\cdot c\cdot (\log^*(n)\cdot b\cdot n^b)^c.\]
Now in order to compute the function value $F(n)=n^a$, it suffices to compute the function values $F_1(n),\ldots , F_{3a}(n)$, which can be done deterministically
in time 
\[\begin{array}{c@{\:\:}l}
     & 3a\cdot n\cdot 4\log^*(n)\cdot c\cdot (\log^*(n)\cdot 3a\cdot n^{3a})^c\\
\leq & 3ac\cdot n^2\:\:\cdot \quad\quad\quad\quad (\log^*(n)\cdot 3a\cdot n^{3a})^c\\
\leq & 3ac\cdot n^2\:\:\cdot \quad\quad\quad\quad (n\cdot\:\: n\cdot\:\: n^{3a})^c\\
=    & 3ac\cdot n^2\cdot n^{(3a+2)c}\leq n^{10ac}=F(n)^{10c}.
\end{array}\]
Therefore we can compute $F(n)$ in time $F(n)^{10c}$. Hence for $C=10c$, we can
compute $F(n)$ deterministically in time $F(n)^C$.
This concludes the proof of the Lemma. 
%
\end{proof}
{\bf Remark.} This constant $C$ only depends on the constant $c$ from Lemma \ref{sig1lemma}. In the definition of Property $[\star]$, we have
defined $h$ as $h=20\cdot (c+2)$. Especially we have $h>2\cdot C$.\\[1.2ex]
\begin{lemma}\label{ineq_lemma}
For every $h$-power $n$, $F(n^{1\slash h})^C\leq F(n)$.
\end{lemma}
\begin{proof} Since $F(1)=1$, the inequality holds for $n=1$. Therefore we consider now the case $n\geq 2^h$. In order to prove the padding inequality from the lemma,
suppose for the contrary that $n\geq 2^h$ is an $h$-power (i.e. such that $n^{1\slash h}$ is also an integer number), and such that
$F(n^{1\slash h})^C>F(n)$. The idea is now to make use of Property $[\star ]$ in order to get a contradiction. Before we go into details, let us first 
describe the general idea of the proof. 

In both stages $n^{1\slash h}$ and $n$, guesses are selected and the function value is defined accordingly. We will show that 
when, say, a guess with a machine $\tilde{S}_{(i)}$ is selected in stage $n^{1\slash h}$ and a guess with machine $\tilde{S}_{(j)}$ is selected in stage $n$,
then from the inequality $F(n^{1\slash h})^C>F(n)$ and the fact that $h$ is sufficiently larger than $C$, it follows that $b_i>b_j$, where these are the two values
at stage $n^{1\slash h}$ and $n$ respectively. From this we can conclude that the guess $(\tilde{S}_{(j)},b_j)$ must be already contained in the list
${\mathcal L}_{n^{1\slash h}}$. Then we use Property $[\star ]$ to conclude that there must be $\Theta (\log\log (n))$ violations of the guess
$(\tilde{S}_{(j)},b_j-p_j)$ within the stages $n^{1\slash h}$ and $n$. Since the size of the list ${\mathcal L}={\mathcal L}_1\cup {\mathcal L}_2$
is only of order $\log^*(n)$, at least one of these violations must have highest priority and therefore be selected. At that point, the value 
$b_j$ will be replaced by some value $>\log^*(n^{1\slash h})>b_j$, and thus the guess cannot be contained in the list at stage $n$ anymore, a contradiction.

We are now ready to give the details of the proof.
Suppose that within the construction of the union function $F$, in stage $n^{1\slash h}$ one of the guesses $(\tilde{S}_{(i)},b_{i})$, $(\tilde{S}_{(i)},b_{i}-p_i)$
is violated
and selected, and in stage $n$ the guess $(\tilde{S}_{(j)},b_{j})$ is violated and selected. Note that since $n$ is an $h$-power, in stage $n$ no guess of the
form  $(\tilde{S}_{(j)},b_{j}-p_{(j)})$ is selected. Thus, we have the following two cases: \\[0.2ex]
Case 1: $(\tilde{S}_{(i)},b_{i})$ is selected in stage $n^{1\slash h}$. Then we have
$F(n^{1\slash h})\in\{(n^{1\slash h}-1)^{b_i},(n^{1\slash h} )^{b_i}\}$ and $F(n)\in\{(n-1)^{b_j},n^{b_j}\}$. Thus $F(n^{1\slash h})^C>F(n)$ implies that 
\[(n^{1\slash h})^{b_i\cdot C}\geq F(n^{1\slash h})^C >F(n)\geq (n-1)^{b_j}\geq n^{b_j\slash 2},\]
and using the fact that $h>2C$, we conclude that 
\[\frac{b_{i}}{2}\: >\: \frac{C}{h}\cdot b_{i}\:\: >\:\: \frac{b_{j}}{2},\:\:\mbox{i.e.}\:\: b_i>b_j.\]
Case 2: $(\tilde{S}_{(i)},b_{i}-p_{(i)})$ is selected in stage $n^{1\slash h}$. Then we proceed as in Case 1, now obtaining
\[(n^{1\slash h})^{b_i\cdot C}\:>(n^{1\slash h})^{(b_i-p_i)\cdot C}\: \geq \: F(n^{1\slash h})^C\:\: >\:\: F(n)>(n-1)^{b_j}>n^{b_{j}\slash 2},\]
from which we also conclude $b_i>b_j$.\\[0.3ex]
In the construction of the union function, new guesses always enter the list with a $b$-value larger than the currently largest $b$-value in the list.
When a guess is violated and selected in the construction, its second component 
is replaced by the currently largest $b$-value. This implies that in our situation, the guess $(\tilde{S}_{(j)},b_{j})$ that is selected in stage $n$ must be contained in the
list ${\mathcal L}_{n^{1\slash h}}$. Since the guess $(\tilde{S}_{(j)},b_{j})$ is selected in stage $n$, it is violated at input length $n$ or $n-1$, i.e. we have
$\mbox{time}_{\tilde{S}_{(j)}}(n)>b_{j}\cdot n^{b_{j}}$ or $\mbox{time}_{\tilde{S}_{(j)}}(n-1)>b_{j}\cdot (n-1)^{b_{j}}$.
Furthermore, the machine $\tilde{S}_{(j)}$ has entered the list at the end of the stage $\delta_{j-1}$, namely with the guess $(\tilde{S}_{(j)},j)$. By construction we 
have $c_{(j)}\leq j$. Therefore, 
$c_{(j)}\leq b_j \leq\log^*(n^{1\slash h})\leq \frac{\log (n-1)}{c_{(j)}}\leq\frac{\log (n)}{c_{(j)}}$. 
Now we consider the two cases $F(n)=n^{b_j}$ and $F(n)=(n-1)^{b_j}$.

{\sl Case: $F(n)=n^{b_j}$.} 
Since $\mbox{time}_{\tilde{S}_{(j)}}$ has Property $[\star ]$, there exist 
pairwise distinct integers
\[m_1,\ldots , m_{\left\lceil\frac{\log\log (n)}{c_{(j)}}\right\rceil}
\in\left (n^{1\slash h},(n^{1\slash h}\cdot \left (1+\frac{\log (n)}{n^{1\slash (d\cdot h)}} \right )^d \right )\]
which are not $h$-powers and such that 
$\mbox{time}_{\tilde{S}_{(j)}}(m_l)> (b_j-p_{(j)})m_l^{b_j-p_{(j)}},l=1,\ldots , \left\lceil\frac{\log\log n}{c_{(j)}}\right\rceil$, 
where $\tilde{S}_{(j)}=\tilde{S}_{j',d}$ and $c_{(j)}=c_{j',d}$.
For every stage $m\in \left.\left (n^{1\slash h},n\right. \right ]$ we have $|{\mathcal L}_m|=\log^*(m)\leq\log^* (n)\leq\frac{\log\log (n)}{c_{(j)}}$.
Whenever some guess different from $(\tilde{S}_{(j)},b_j)$ is selected within a stage $m>n^{1\slash h}$,
its $b$-value is replaced by $\log^*(m)>b_j$. Therefore, after at most $2\log^*n$ such stages, $(\tilde{S}_{(j)},b_j)\in {\mathcal L}_1$ 
is the guess with the smallest value $b_j$
in the list ${\mathcal L}$ such that one of the two guesses $(\tilde{S}_{(j)},b_j),(\tilde{S}_{(j)},b_j-p_{(j)})$ will ever be violated again. We have 
\[n^{1\slash h}\cdot \left (1+\frac{\log (n)}{n^{1\slash (dh)}}\right )^d\: =\: 
\left (n^{1\slash h}+\log (n)\cdot n^{\frac{1}{h}-\frac{1}{d\cdot h}}\right )\cdot \left (1+\frac{\log (n)}{n^{1\slash (dh)}}\right )^{d-1},\] 
and therefore 
\[\left | \left (n^{1\slash h},n^{1\slash h}\cdot \left (1+\frac{\log (n)}{n^{1\slash (d\cdot h)}} \right )^d\right )\right |\:\geq\: \log (n)\cdot n^{(1-\frac{1}{d})\frac{1}{h}}.\]
Since $c_{(j)}\leq j\leq\log^* (n)$, we have 
\[\frac{\log\log (n)}{c_{(j)}}\geq \frac{\log\log (n)}{\log^*n}> 2\log^* (n),\]
where the last inequality follows from the remark after the proof of Lemma \ref{sig1lemma} and the fact that $n\geq 2^h>2^c$.
Since in every stage of the construction of $F$, violated guesses are first ordered by their $b$-value and since $\frac{\log\log (n)}{c_{(j)}}>2\log^*(n)$, 
there exists some $m\in \left (n^{1\slash h},n^{1\slash h}\cdot \left (1+\frac{\log (n)}{n^{1\slash (dh)}}\right )^d\right )$
such that $(\tilde{S}_{(j)},b_j-p_{(j)})$ is violated 
and one of the two guesses $(\tilde{S}_{(j)},b_j),(\tilde{S}_{(j)},b_j-p_{(j)})$ from ${\mathcal L}_{m,1}\cup {\mathcal L}'_{m,1}$
is 
selected at stage $m$ in the construction of $F$. According to the construction of the union function, the guess $(\tilde{S}_{(j)},b_j)$ is then replaced by $(\tilde{S}_{(j)},\log^*(m))$
in the list ${\mathcal L}_1$. Since
\[b_j\: <\: b_i\: \leq\: \log^*(n^{1\slash h})\:\leq\:\log^*(m),\]
we have $\log^*(m)\neq b_j$, hence the guess $(\tilde{S}_{(j)},b_j)$ cannot be in the list ${\mathcal L}_{n,1}$, a contradiction.

{\sl Case: $F(n)=(n-1)^{b_j}$.} This means that in stage $n$, the guess $(\tilde{S}_{(j)},b_j)$ is selected from the list ${\mathcal L}_2$, since
$\mbox{time}_{\tilde{S}_{(j)}}(n-1)>b_j(n-1)^{b_j}$. Again since $\mbox{time}_{\tilde{S}_{(j)}}$ has Property $[\star ]$, there exist
pairwise distinct non-$h$-power integers 
\[m_1,\ldots , m_{\left\lceil\frac{\log\log (n-1)}{c_{(j)}}\right\rceil}\:\in\: I_{n-1,d}
=\left ((n-1)^{1\slash h},(n-1)^{1\slash h}\cdot\left (1+\frac{\log (n-1)}{(n-1)^{1\slash (d\cdot h)}} \right )^d\right )\]
such that 
\[\mbox{time}_{\tilde{S}_{(j)}}(m_l)>(b_j-p_{(j)})m_l^{b_j-p_{(j)}},\: l=1,\ldots , \left\lceil\frac{\log\log (n-1)}{c_{(j)}}\right\rceil.\]
Since $(n-1)^{1\slash h}\geq n^{1\slash h}-1$, we obtain that at least $\lceil\frac{\log\log (n-1)}{c_{(j)}}\rceil -1$ of these integers $m_l$ are 
greater than $n^{1\slash h}$. Again, from the fact that $n\geq 2^h>2^c$ and the remark after the proof of Lemma \ref{sig1lemma}, we conclude that
$\frac{\log\log (n-1)}{c_{(j)}}-1>2\log^*(n)$. Within the stages $n^{1\slash h}$ up to $n$, the list ${\mathcal L}={\mathcal L}_1\cup {\mathcal L}_2$ 
contains
at most $2\log^*(n)$ guesses, thus for at least $\frac{\log\log (n-1)}{c_{(j)}}-1-2\log^*(n)$ of the integers $m_l$, the guess 
$(\tilde{S}_{(j)},b_j-p_{(j)})\in {\mathcal L}_2$ has highest priority of being selected. Thus, as in the previous case, it will eventually be selected
within one of these stages, and thus the guess $(\tilde{S}_{(j)},b_j)$ cannot be contained in the list ${\mathcal L}_{n,2}$ in stage $n$ anymore, a contradiction.

Thus we obtain that $F$ satisfies the inequality $F(n^{1\slash h})^C\leq F(n)$ for each $h$-power $n$. This concludes the proof of the Lemma.
\end{proof} 
{\bf Remark.} In the proof of Lemma \ref{ineq_lemma} we have made use of the way in which guesses are ordered in the construction of the union function $F$.
In stages $n$ where $n$ is an $h$-power, the lexicographically smallest violated guess $(\tilde{S}_{(i)},b_i)$ from ${\mathcal L}_n$ is selected, first ordered by  
$b_i$ and then by $i$. In the case when $n$ is not an $h$-power, the smallest violated guess from the extended list ${\mathcal L}_n\cup {\mathcal L}'_n$ is selected, namely one of the 
guesses $(\tilde{S}_{(i)},b_i)$ and $(\tilde{S}_{(i)},b_i-p_{(i)})$ for some $i$, first ordered by $b_i$, then by the second entry $b_i$ or $b_i-p_{(i)}$ and then by $i$.
In the proof of Lemma \ref{ineq_lemma}, this property that guesses are always first ordered by $b_i$ was used to show that if the inequality $F(n^{1\slash h})^C\leq F(n)$
is violated, this gives a contradiction, since the guess selected in stage $n$ would have been selected earlier in the construction of $F$.\\[2ex]
Now we will make use of \emph{Padding} in order to show that $\mbox{DTIM}\tilde{\mbox{E}}(F)=\tilde{\Sigma}_2(F)$ also implies 
$\mbox{DTIM}\tilde{\mbox{E}}(F^{C^2})=\tilde{\Sigma}_2(F^{C^2})$. As we will see in the proof of the next lemma, the crucial properties 
which allow us
to use Padding are the padding inequality $F(n^{1\slash h})^C\leq F(n)$ and the fact that we also have 
$\tilde{P}=\mbox{DTIM}\tilde{\mbox{E}}(F(n+1))=\tilde{\Sigma}_2(F(n+1))=\tilde{\Sigma}_2^p$.
\begin{lemma}\label{padding_lemma}
$\mbox{DTIM}\tilde{\mbox{E}}(F)=\tilde{\Sigma}_2(F)$ implies $\mbox{DTIM}\tilde{\mbox{E}}(F^{C^2})=\tilde{\Sigma}_2(F^{C^2})$.
\end{lemma}
\begin{proof}
Let $L\in\tilde{\Sigma}_2(F^{C^2})$. We want to show that this also implies $L\in\mbox{DTIM}\tilde{\mbox{E}}(F^{C^2})$.
For this purpose, we want to construct some associated $L'\in\tilde{\Sigma}_2(F)$, a polynomially padded version of the given problem $L$.
Then we obtain $L'\in\mbox{DTIM}\tilde{\mbox{E}}(F)=\tilde{P}=P$. Since $L'$ is a polynomially padded version of $L$, this immediately gives $L\in P$, therefore
$L\in P=\mbox{DTIM}\tilde{\mbox{E}}(F)\subseteq \mbox{DTIM}\tilde{\mbox{E}}(F^{C^2})$.

Now we describe this in detail. We let the padded version $L'$ of $L$ be defined as
\[L'\: =\: \left\{x10^k|\: x\in L,\: |x|=n, |x10^k|=n^{h^2}-1 \right\}.\] 
Le us give the intuition for this choice, which is twofold. On the one hand, we want to use the padding inequality, namely twice in 
$F(n)^{C^2}\leq F(n^h)^C\leq F(n^{h^2})$. On the other hand, the integers $n^{h^2}$ are $h$-powers. If the string lengths of elements of $L'$ would all be 
$h$-powers, the Property $[\star ]$ would not be satisfied for $L$ - note that in case of violation of a guess, Property $[\star ]$ requires existence of 
smaller violations at input lengths which are {\sl not} $h$-powers. Therefore we choose the string lengths of elements of $L'$ to be 
of the form $n^{h^2}-1$. But still the padding inequality only gives the bound $F(n^{C^2})\leq F(n^{h^2})$. This is the reason why we
constructed the union function $F$ such that it also satisfies $\tilde{\Sigma}_2^p=\tilde{\Sigma}_2(F(n+1))$.

Let $\tilde{S}_L$ be a $\tilde{\Sigma}_2$-machine for $L$ which is $O(F^{C^2})$-time bounded. In order to keep notations simple, say that 
$\tilde{S}_L$ satisfies Property $[\star ]$ with parameters $c_L,d_L,p_L$. Let $S'$ be the $\Sigma_2$-machine which accepts $L'$ in the standard way:
Given an input $y$, it checks if $y$ is of the form $y=x10^k$ with $|y|=|x|^{h^2}-1$, and in parallel (i.e. on a separate track) it 
simulates the machine $\tilde{S}_L$ on input $x$. If the input string $y$ is not of the form $y=x10^k$, $S'$ rejects. Otherwise it continues the 
simulation f the machine $S_L$ on input $x$.

Let us use the following notation: $N$ denotes the string length of an
instance $y$ of the padded version $L'$, and $n$ denotes the length of the associated string $x$ with $y=x10^k$. The values $n$ and $N$ are related as follows:
\[n^{h^2}-1=N,\:\: n=(N+1)^{1\slash h^2}\]
Thus the running time of the machine $S'$ is in $O(F(|x|)^{C^2})=O(F((N+1)^{1\slash h^2})^{C^2})$.
Applying the padding inequality, we obtain that $F((N+1)^{1\slash h^2})^{C^2}\leq F(N+1)$. Altogether we obtain that
the running time of the machine $S'$ on input $y$ of length $N$ is in $O(F(N+1))$.
Thus we have $L'\in\Sigma_2(F(N+1))$. 
Now we want to show that we also have $L'\in \tilde{\Sigma}_2(F(n+1))$. According to Lemma \ref{progsys_lemma} it is sufficient to show that 
there exist parameters $c',d',p'$ such that the function $\mbox{time}_{S'}$ satisfies Property $[\star ]$ with these parameters. More precisely, we have to show
that there exist a machine index $j$ for the machine $S'$ and an integer $d'$ such that the running time $\mbox{time}_{S_j}$ of the machine $S'=S_j$ satisfies Property $[\star ]$ with 
parameters $c'=c_{j,d},d',p'=p_j$. We assume that we have already fixed $j$ and $d'$.
Now we will follow the lines of the definition of Property $[\star ]$ in order to derive conditions for
$c',d',p'$. The approach is somewhat technical yet straight forward.

Suppose that $N\geq 2^{c'^2}$ and $c'\leq a\leq b\leq\frac{\log (N)}{c'}$ are such that $\mbox{time}_{S'}(N)>aN^b$. We have to show that there exist 
integers $M_1,\ldots , M_{\left\lceil\frac{\log\log N}{c'}\right\rceil}\in I_{N,d'}$ which are not $h$-powers such that for 
$l=1,\ldots , \left\lceil\frac{\log\log N}{c'}\right\rceil$, $\mbox{time}_{S'}(M_l)>(a-p')M_l^{b-p'}$.
From the construction of the machine $S'$ it follows that $\mbox{time}_{S'}(N)>aN^b$ implies that 
\begin{equation}\label{padding_eq_aux}
\mbox{time}_{\tilde{S}_L}(n)>aN^b=a(n^{h^2}-1)^b>a\left (\frac{n^{h^2}}{2}\right )^b\: =\: a\cdot\frac{n^{h^2\cdot b}}{2^b}.
\end{equation}
Since $b\leq\frac{\log (N)}{c'}=\frac{\log (n^{h^2}-1)}{c'}\leq\frac{\log (n^{h^2})}{c'}$, we obtain
$2^b\leq n^{h^2\slash c'}$. 
Combining this with (\ref{padding_eq_aux}), we obtain
$\mbox{time}_{\tilde{S}_L}(n)>an^{h^2(b-1\slash c')}>an^{h^2(b-1)}$.
We want to achieve that we can now make use of the fact that the machine $S_L$ satisfies Property $[\star ]$. Thus the first condition is that 
\begin{itemize}
\item[(i)] $N\geq 2^{c'^2}$ implies that $n\geq 2^{c_L^2}$, and\\
           $c'\leq a\leq b\leq \frac{\log (N)}{c'}$ implies that $c_L\leq a\leq h^2(b-1)\leq\frac{\log (n)}{c_L}$
\end{itemize} 
Suppose that Condition (i) is satisfied. Then we can apply Property $[\star ]$ to the machine $\tilde{S}_L$. Thus there exist integers 
$m_1,\ldots , m_{\left\lceil\frac{\log\log (n)}{c_L}\right\rceil}\in I_{n,d_L}$ which are not $h$-powers such that 
\begin{equation*}
\mbox{time}_{\tilde{S}_L}(m_l)>(a-p_L)m_l^{h^2(b-1)-p_L},\: l=1,\ldots , \left\lceil\frac{\log\log (n)}{c_L}\right\rceil
\end{equation*}
This implies that for the integers $M_l:=m_l^{h^2}-1$ (which are not $h$-powers), we have
\begin{equation*}
\mbox{time}_{S'}(M_l)> (a-p_L)\left ((M_l+1)^{1\slash h^2}\right )^{h^2(b-1)-p_L}
> (a-p_L)M_l^{b-1-\frac{p_L}{h^2}},\: l=1,\ldots , \left\lceil\frac{\log\log (n)}{c_L}\right\rceil. 
\end{equation*}
We want to show that this yields sufficiently many integers $M\in I_{N,d'}$ such that $\mbox{time}_{S'}(M)>(a-p')M^{b-p'}$.
For this purpose we will first require that $p_L\leq p'$ and $1+\frac{p_L}{h^2}\leq p'$. Then this immediately yields that
\begin{equation}
\mbox{time}_{S'}(M_l)> (a-p')M_l^{b-p'},\: l=1,\ldots , \left\lceil\frac{\log\log (n)}{c_L}\right\rceil.
\end{equation}
Now we have to note that the inclusion 
\[\{m^{h^2}-1\mid m\in I_{n,d_L}\}\:\subseteq\: I_{N,d'}\]
does not hold. Therefore it is impossible to prove that for each integer $m_l\in I_{n,d_L}$ the corresponding integer $M_l=m_l^{h^2}-1$
is contained in the interval $I_{N,d'}$. But it suffices to show that sufficiently many of these integers $M_l$ are in $I_{N,d'}$, namely at least 
$\lceil\frac{\log\log (N)}{c'}\rceil$ of them. For this purpose, we will now proceed as follows. 
%
First we require that $\lceil\frac{\log\log (n)}{c_L}\rceil>3\lceil\frac{\log\log (N)}{c'}\rceil$. Then it will suffice to show for 
half of these integers $m_l$ that the associated $M_l$ is in $I_{N,d'}$. In particular
it will be sufficient to show that 
$n^{1\slash h}+1<m_l<n^{1\slash h}(1+\frac{\log (n)}{n^{1\slash (hd_L)}})^{d_L}-\frac{1}{2}\frac{\log\log n}{c_L}$ implies $M_l\in I_{N,d'}$
%
%
Hence we obtain the following conditions:
\begin{itemize}
\item[(ii)] $m_l\in I_{n,d_L}$ with $n^{1\slash h}+1<m_l<n^{1\slash h}(1+\frac{\log (n)}{n^{1\slash (hd_L)}})^{d_L}-\frac{1}{2}\frac{\log\log n}{c_L}$ 
            implies $M_l\in I_{N,d'}$
\item[(iii)] $\left\lceil\frac{\log\log (n)}{c_L}\right\rceil\: >\: 3\cdot \left\lceil\frac{\log\log (N)}{c'}\right\rceil$            
\item[(iv)] $p'\geq p_L$, which also implies $p'\geq 1+\frac{p_L}{h^2}$.
\end{itemize}
As we have shown above, if Conditions (i)-(iv) are satisfied, then this immediately yields that $S'$ satisfies Property $[\star ]$ with parameters $c',d',p'$.
Thus we will now show that the parameters $c',d',p'$ can be chosen such as to satisfy (i)-(iv).\\[0.3ex]
{\sl Concerning Condition (i):} We have $n=(N+1)^{1\slash h^2}>N^{1\slash h^2}$.
Thus $N\geq 2^{c'^2}$ implies $n>2^{c'^2\slash h^2}$. Thus in order to satisfy the first part of Condition (i), it suffices to choose $c'\geq c_L\cdot h$.
Then $c'\leq a$ also implies $c_L\leq a$. Moreover, $b\leq\frac{\log (N)}{c'}$ implies
\[h^2(b-1)\leq h^2b\leq h^2\frac{\log (N)}{c'}=h^2\frac{\log (n^{h^2}-1)}{c'}\leq h^2\frac{\log (n^{h^2})}{c'}=\frac{h^4\log (n)}{c'}\]
Thus altogether we obtain that Condition (i) holds provided we choose $c'\geq h^4\cdot c_L$.\\[0.3ex]
{\sl Concerning Condition (iii):} We have
\begin{eqnarray*}
\left\lceil\frac{\log\log (N)}{c'}\right\rceil & = & \left\lceil\frac{\log\log (n^{h^2}-1)}{c'}\right\rceil\\
    & < & \frac{h^2\log\log (n)}{c'}\\
    & \leq & \frac{h^2\log\log (n)}{h^4\cdot c_L}\quad\:\mbox{(since $c'\geq h^4\cdot c_L$)}\\
    \Longrightarrow  \left\lceil\frac{\log\log (n)}{c_L}\right\rceil & > & h^2\cdot\left\lceil \frac{\log\log (N)}{c'}\right\rceil
                     \: >\: 3\cdot\left\lceil\frac{\log\log (N)}{c'}\right\rceil.
\end{eqnarray*}
Thus $c'\geq h^4\cdot c_L$ also implies that Condition (iii) holds.\\[0.3ex]
{\sl Concerning Condition (ii):} 
Recall that the intervals $I_{n,d_L},I_{N,d'}$ are defined as follows:
\[I_{n,d_L}=\left (n^{1\slash h},n^{1\slash h}\left (1+\frac{\log (n)}{n^{1\slash (d_L\cdot h)}} \right )^{d_L} \right ),\:
I_{N,d'}=\left (N^{1\slash h},N^{1\slash h}\left (1+\frac{\log (N)}{N^{1\slash (d'\cdot h)}} \right )^{d'} \right )\]
We have to assure that 
$n^{1\slash h}+1<m<n^{1\slash h}(1+\frac{\log (n)}{n^{1\slash (d'\cdot h)}} )^{d'}-\frac{1}{2}\cdot\frac{\log\log (n)}{c_L}$
implies $M\in I_{N,d'}$, where $M=m^{h^2}-1$. The first part of this requirement is that 
$m>n^{1\slash h}+1$ implies $m^{h^2}-1>N^{1\slash h}$.
For such $m$ with $m>n^{1\slash h}+1$, we have 
\[m^{h^2}-1\: >\: (n^{1\slash h}+1)^{h^2}-1\: >\: (n^{1\slash h})^{h^2}\: =\: (n^{h^2})^{1\slash h}\: >\: (n^{h^2}-1)^{1\slash h}\: =\: N^{1\slash h}.\]

Now for the second part of the condition, suppose that 
$m<n^{1\slash h}\left (1+\frac{\log (n)}{n^{1\slash (d_Lh)}}\right )^{d_L}-\frac{1}{2}\cdot\frac{\log\log (n)}{c_L}$.
We have to show that for such $m$, the associated integer $M=m^{h^2}-1$ is smaller than the upper bound of the interval $I_{N,d'}$. 
For this purpose, it suffices to show that
\begin{equation}\label{padding_eq_xyz}\begin{array}{c@{\:}l@{\:}c@{\:}l}
      & \left (n^{1\slash h}\left (1+\frac{\log (n)}{n^{1\slash (d_Lh)}}\right )^{d_L}-1\right )^{h^2}-\frac{1}{2}\cdot\frac{\log\log (n)}{c_L} & \leq & 
        (n^{h^2}-1)^{1\slash h}\left (1+\frac{\log N}{N^{1\slash (d'h)}}\right )^{d'}\\
 \Longleftrightarrow & 
        \left (\left (n^{1\slash h}\left (1+\frac{\log n}{n^{1\slash (d_Lh)}}\right )^{d_L}-1\right )^{h^2}-\frac{1}{2}\cdot\frac{\log\log (n)}{c_L}\right )^h & \leq &
                       (n^{h^2}-1)\cdot\left (1+\frac{\log N}{N^{1\slash (d'h)}}\right )^{h\cdot d'}
\end{array}\end{equation}
We can bound the left hand side in (\ref{padding_eq_xyz}) as follows:
\begin{equation}\label{padding_eq_zzz}\begin{array}{c@{\:\:}l}
        & \left (\left (n^{1\slash h}\left (1+\frac{\log n}{n^{1\slash (d_Lh)}}\right )^{d_L}-1\right )^{h^2}-\frac{1}{2}\cdot\frac{\log\log (n)}{c_L}\right )^h \\
\leq    & \left (n^{1\slash h}\left (1+\frac{\log n}{n^{1\slash (d_Lh)}}\right )^{d_L}-1\right )^{h^2\cdot h}-\frac{1}{2}\cdot\frac{\log\log (n)}{c_L}\\
\leq    & (n^{1\slash h})^{h^3}\cdot \left (1+\frac{\log n}{n^{1\slash (d_Lh)}}\right )^{d_L\cdot h^3}\: -1\: -\frac{1}{2}\cdot\frac{\log\log (n)}{c_L}\\
\leq    & n^{h^2}\cdot \left (1+\frac{\log n}{n^{1\slash (d_Lh)}}\right )^{d_L\cdot h^3}\: -\frac{1}{2}\cdot\frac{\log\log (n)}{c_L}
\end{array}\end{equation}
Thus in order to satisfy the inequality (\ref{padding_eq_xyz}), it suffices to satisfy the following inequality:
\begin{equation}\label{padding_eq_zzzz}
n^{h^2}\cdot \left (1+\frac{\log n}{n^{1\slash (d_Lh)}}\right )^{d_L\cdot h^3}\: -\frac{1}{2}\cdot\frac{\log\log (n)}{c_L}
\:\:\leq\;\: (n^{h^2}-1)\cdot\left (1+\frac{\log N}{N^{1\slash (d'h)}}\right )^{h\cdot d'}
\end{equation}
%
%
We have $\left (1+\frac{\log (n)}{n^{1\slash (d_L\cdot h)}}\right )^{d_L\cdot h^3}=o\left (\frac{1}{2}\cdot\frac{\log\log (n)}{c_L}\right )$. 
Thus there exists some $n_0$ such that for all $n\geq n_0$, 
$\frac{1}{2}\cdot\frac{\log\log (n)}{c_L}>\left (1+\frac{\log (n)}{n^{1\slash (d_L\cdot h)}}\right )^{d_L\cdot h^3}$.
This $n_0$ only depends on $c_L,d_L$ and $h$. Recall that in Condition (i) we have the requirement that $N\geq 2^{c'^2}$ implies $n\geq 2^{c_L^2}$.
Now we add the requirement that this also implies $n\geq n_0$:
\begin{itemize}
\item[(v)] $N\geq 2^{c'^2}$ implies $n\geq n_0$.
\end{itemize}
Condition (v) can be satisfied by choosing the parameter $c'$ sufficiently large. Thus we may now assume that (v) holds, and therefore we can conclude that
$\frac{1}{2}\cdot\frac{\log\log (n)}{c_L}>\left (1+\frac{\log (n)}{n^{1\slash (d_L\cdot h)}}\right )^{d_L\cdot h^3}$.
%
Then the left hand side in (\ref{padding_eq_zzzz}) is upper bounded by 
$(n^{h^2}-1)\cdot \left (1+\frac{\log (n)}{n^{1\slash (d_L\cdot h)}}\right )^{d_L\cdot h^3}$.
Now it suffices to choose $d'$ sufficiently large, namely such that $d'\geq h^2\cdot d_L$.
We have $\log (n)\leq \log (n^{h^2}-1)$, and our choice of $d'$ implies that 
\[d_L\cdot h^3\leq h\cdot d'\:\:\mbox{and}\:\: n^{\frac{1}{d_L\cdot h}}\:\geq\: (n^{h^2})^{\frac{1}{d'\cdot h}}\:\geq\: (n^{h^2}-1)^{\frac{1}{d'\cdot h}},\]
which then yields 
\[\left (1+\frac{\log (n)}{n^{1\slash (d_Lh)}}\right )^{d_L\cdot h^3}\:\leq\:
  \left (1+\frac{\log (n^{h^2}-1)}{(n^{h^2}-1)^{\frac{1}{d'\cdot h}}}\right )^{h\cdot d'}.\]
Thus we have shown the following: If we choose the parameters $c',d',p'$ sufficiently large such as to satisfy the conditions (i)-(v), 
then machine $S'$ satisfies Property $[\star ]$ with parameters $c',d',p'$. Recall that in the original list of $\Sigma_2$-machines, each machine occurs
infinitely often. In the construction of our subfamily $(\tilde{S}_{i,d})$, the associated parameters $c_{i,d}$ and $p_i$ are
monotone increasing in $i$ and $d$. Thus we choose $d'$ and $i$ sufficiently large such that $L(\tilde{S}_{i,d'})=L'$ and such that (i)-(v) are satisfied
for the parameters $c'=c_{i,d'},p'=p_i$ and $d'$. 

Hence we have shown that $L'\in\tilde{\Sigma}_2(F(n+1))$, which means $L'\in P$. Since $L'$ is a polynomially padded version of $L$, this also gives 
$L\in P=\tilde{\Sigma}_2(F(n+1))$. This implies that 
$L\in P=\mbox{DTIM}\tilde{\mbox{E}}(F)\subseteq \mbox{DTIM}\tilde{\mbox{E}}(F^{C^2})$, which concludes the proof of the lemma. 
\end{proof}
Now we will show that powers of the union function $F$ satisfy Property $[\star ]$. Below in the next section we will then show that Gupta's result 
also holds for complexity classes $\mbox{DTIM}\tilde{\mbox{E}}(t)$ and $\tilde{\Sigma}_2(t)$ where $t(n)\geq n\log^*(n)$ is a function that can be computed in time $t(n)^{1-\epsilon}$
and satisfies Property $[\star ]$. In particular this will then hold for the function $F(n)^{C^2}$.
\begin{lemma}\label{F_propertystar_lemma}
For every integer $q\geq 1$, there exist $c_q,d_q,p_q\in {\mathbb N}$ such that the function $F^q(n)$ satisfies \emph{Property $[\star ]$} with parameters 
$c_q,d_q,p_q$. 
\end{lemma}
\begin{proof}
Suppose that we have already chosen the parameters $c_q,d_q$ and $p_q$, and suppose that $c_q$ is sufficiently large compared to $d_q$ that
$n\geq 2^{c_q^2}$ implies $\log n\leq n^{1\slash (d_q\cdot h)}$. Now suppose that
$n\geq 2^{c_q^2}$ and $c_q\leq a\leq b\leq \frac{\log (n)}{c_q}$ are such that $F^q(n)>an^b$. Intuitively, we have to show that there are sufficiently many 
integers $m$ in the interval $I_{n,d_q}$ such that $F(m)$ is sufficiently large. In the following we will show that this holds true. The reason is that
from the construction 
of the union function $F$ it follows directly that for most integers $m$, $F(m)$ is equal to $m^{\log^*(m)}$, and this will turn out to be sufficient in order to 
satisfy the implication in Property $[\star ]$. 

Let us now give the details. Since $F(n)\leq n^{\log^*(n)}$, this implies
$c_q\leq a\leq b\leq q\cdot\log^*(n)$. We consider the associated interval 
\[I_{n,d_q}\: =\: \left (n^{1\slash h},\: n^{1\slash h}\cdot \left (1+\frac{\log (n)}{n^{1\slash (d_q\cdot h)}}\right )^{d_q}\right ).\]
Let $t=\log^*(n)$. Recall that $I_t=[\delta_{t-1}+1,\delta_t]$ is the interval on which the $\log^*$ function is equal to $t$.
We show that $I_{n,d_q}\subseteq I_t\cup I_{t-1}$. Since $I_{t-1}=[\delta_{t-2}+1,\delta_{t-1}]$, it suffices to show that 
$n^{1\slash h}\geq\delta_{t-2}$. Since $t=\log^*n$, we have $n\geq\delta_{t-1}=2^{\delta_{t-2}}$. Now since $n\geq 2^{c_q^2}$, we have 
that $\log n\leq n^{1\slash (d_q\cdot h)}\leq n^{1\slash h}$. We conclude that 
\[\delta_{t-2}=\log (\delta_{t-1})\leq\log (n)\leq n^{1\slash h}.\]
Thus the inclusion  $I_{n,d_q}\subseteq I_t\cup I_{t-1}$ holds. 
Therefore, one of the two intervals $I_t$ and $I_{t-1}$ contains at least half of the 
elements from $I_{n,d_q}$:
\[\exists\tau\in\{t,t-1\}\:\: |I_{n,d_q}\cap I_{\tau}|\:\geq\:\frac{|I_{n,d_q}|}{2}\]
We consider this $\tau$ and the construction of the union function $F$ within the interval $I_{\tau}$. 
In the interval $I_{\tau}$, the list ${\mathcal L}={\mathcal L}_1\cup {\mathcal L}_2$ contains guesses for the first $\tau$ machines.
From the construction of $F$ it follows that for each of the two sublists ${\mathcal L}_1,\: {\mathcal L}_2$, each machine 
$\tilde{S}_{(i)}$ can be selected at most once within a guess $(\tilde{S}_{(i)},b_i-p_{(i)})$ from the extended sublist. Moreover, whenever 
$\tilde{S}_{(i)}$ is selected from a sublist ${\mathcal L}_j,j\in\{1,2\}$ within a guess $(\tilde{S}_{(i)},b_i)$, then afterwards the guess 
will be replaced by $(\tilde{S}_{(i)},\tau)$ in that sublist. 
Hence there are at most $4\cdot\tau$ stages $m$ within the interval $I_{\tau}$ such that $F(m)<m^{\tau}$.
Thus we obtain
\[|\{m\in I_{n,d_q}\cap I_{\tau }\mid F(m)=m^{\tau}\}|\:\geq\: |I_{n,d_q}\cap I_{\tau}|-4\cdot\tau.\] 
We want to show that Property $[\star ]$ holds for $F^q$, i.e. that there exist $m_1,\ldots ,m_{\left\lceil\frac{\log\log (n)}{c_q}\right\rceil}$ in $I_{n,d_q}$
not being $h$-powers such that $F(m_l)>(a-p_q)\cdot m_l^{b-p_q},l=1,\ldots  \left\lceil\frac{\log\log (n)}{c_q}\right\rceil$.
Thus it suffices to show that the following two properties hold:
\begin{itemize}
\item[(i)] $\frac{|I_{n,d_q}|}{2}-4\cdot\tau\geq\left\lceil\frac{\log\log (n)}{c_q}\right\rceil$, or equivalently 
          $\frac{|I_{n,d_q}|}{2}\geq\left\lceil\frac{\log\log (n)}{c_q}\right\rceil+4\cdot\tau$.
\item[(ii)] For all $m\in I_{n,d_q}\cap I_{\tau}$, $m^{q\cdot \tau}>(a-p_q)m^{b-p_q}$.
\end{itemize}
\emph{Concerning Property (i).} Directly from the definition of the interval $I_{n,d_q}$ we get that this inequality is equivalent to
\begin{equation}\label{eqn_xxx}
\frac{1}{2}\cdot n^{1\slash h}\cdot \left (\left (1+\frac{\log (n)}{n^{1\slash (d_q\cdot h)}}\right )^{d_q}-1 \right )
\geq\left\lceil\frac{\log\log (n)}{c_q}\right\rceil +4\cdot\tau
\end{equation}
Now we will make use of the following inequality: 
\[\left (1+\frac{\log (n)}{n^{1\slash (d_q\cdot h)}}\right )^{d_q}\:\geq\: 1+d_q\cdot\left (\frac{\log (n)}{n^{1\slash (d_q\cdot h)}}\right )^{d_q}.\]
This implies that the left hand side in (\ref{eqn_xxx}) is 
$\geq \frac{1}{2}\cdot n^{1\slash h}\cdot d_q\cdot \left (\frac{\log (n)}{n^{1\slash (d_q\cdot h)}}\right )^{d_q}$.
Hence the logarithm of the left hand side in (\ref{eqn_xxx}) is 
\begin{eqnarray*}
 & \geq & \frac{1}{h}\log (n)-\log (2)+\log (d_q)+d_q\cdot\left (\log\log n\: -\frac{1}{d_q\cdot h}\cdot\log (n)\right )\\
 & =    & \log d_q+d_q\log\log n -1.
\end{eqnarray*}
The logarithm of the right hand side in (\ref{eqn_xxx}) is 
\begin{eqnarray*}
  \log\left (\left\lceil\frac{\log\log (n)}{c_q}\right\rceil+4\tau\right ) & \leq &  \log\left (\frac{\log\log (n)}{c_q}+4\tau+1\right )\\
  & = & \log\left (\frac{\log\log (n)}{c_q}\cdot \left (1+\frac{(4\tau +1)c_q}{\log\log (n)}\right )\right )\\
  & = & \log\log\log n-\log c_q+\log \left (1+\frac{(4\tau +1) c_q}{\log\log n}\right )\\
  & \leq & \log\log\log (n)-\log (c_q) +\log (1+4c_q)\quad\mbox{(since $\tau\leq\log^*(n)$)}\\
  & \leq & \log\log\log (n)-\log (c_q) +\log (5c_q) = \log\log\log (n)+\log (5).
\end{eqnarray*}
Thus it suffices to choose $d_q$ sufficiently large such that $\log (d_q)-1\geq\log (5)$, and then Property (i) holds.\\[0.4ex]
\emph{Concerning Property (ii).} From the construction of the union function $F$ it follows that $F(n)\leq n^{\log^*(n)}$.
Since $n^{q\cdot\log^*(n)}\geq F(n)^q>an^b$ and $a\leq b$, we have $a\leq b\leq q\cdot\log^*n\leq q\cdot (\tau +1)$. Thus it suffices to choose 
the parameter $p_q$ sufficiently large such that 
$q\cdot\tau > q\cdot (\tau +1)-p_q$, i.e. $p_q>q$. It follows that Property (ii) holds as well. This concludes the proof of the Lemma.
\end{proof}
\section{A Separation Result for $\tilde{\Sigma}_2(t)$}\label{gupta_section}
In this section we show that the separation between deterministic and $\Sigma_2$ classes from \cite{G96} also holds for $\tilde{\Sigma}_2$-classes.
In the introduction we already formulated Gupta's separation result in Theorem \ref{gupta_thm}: For every time-constructible function $t(n)\geq n\log^*(n)$,
$\mbox{DTIME}(t)\subsetneq\Sigma_2(t)$. In the proof of this result in \cite{G96}, Gupta states that if $t(n)$ is a time-constructible function, then the function
$\max\{2n,t(n)\slash\sqrt{\log^*(t(n))}\}$ is also time-constructible. For a proof of this, he refers to a paper by Kobayashi \cite{K85}. Unfortunately, we have not been able to 
verify that it is shown in or follows easily from results in \cite{K85} that for each time-constructible function $t(n)$, the function 
$t(n)\slash\sqrt{\log^*(t(n))}$ is also time-constructible. However, this problem does not occur in case if $t(n)$ is constructible in time $t(n)^{1-\epsilon}$
for some $\epsilon >0$. In our case, we know that the union function $F$ is deterministically computable in time $F(n)^C$ for some constant $C$. Furthermore we know that
for every $q\geq C$, the function $F^q$ has Property $[\star ]$.   
\begin{lemma}\label{eqlemma}
For every function $t(n)\geq n\cdot\log^*(n)$ which is deterministically computable in time $t(n)^{1-\epsilon}$ for some constant $\epsilon >0$, if 
$t(n)$ has Property $[\star ]$, then $\Sigma_2(t) = \tilde{\Sigma}_2(t)$.
\end{lemma}
\begin{proof}
Obviously, $\tilde{\Sigma}_2(t)\subseteq\Sigma_2(t)$. Hence it is sufficient to show the inclusion $\Sigma_2(t)\subseteq \tilde{\Sigma}_2(t)$. 
Suppose that $L\in\Sigma_2(t)$, say via some
$\Sigma_2$-machine $S_i$. We have to show that $L\in \tilde{\Sigma}_2(t)$. 
According to Lemma \ref{progsys_lemma}(c), it suffices to construct some
$\Sigma_2$-machine $S_j$ and some $d\in {\mathbb N}$ such that $L(S_j)=L$, $\mbox{time}_{S_j}(n)=O(\mbox{time}_{S_i}(n))$ and such that the function 
$\mbox{time}_{S_j}(n)$ satisfies Property $[\star ]$ with parameters $c_{j,d},p_j,d$.
Suppose that $\mbox{time}_{S_i}(n)\leq\gamma\cdot t(n)$ for all $n$, for some constant $\gamma$.
Consider the $\Sigma_2$-machine $S_j$ which simulates on input $x$ the computation of $S_i(x)$ and in parallel computes the function value $t(|x|)$
and uses precisely $\gamma\cdot t(|x|)$ steps in total.
We claim that for $j$ being sufficiently large, there exists some $d\in {\mathbb N}$ such that $\mbox{time}_{S_j}(n)$ satisfies Property $[\star ]$ with parameters
$c_{j,d},p_j,d$. 
Suppose that $c_t,p_t,d_t\in {\mathbb N}$ are such that 
the function $t(n)$ satisfies Property $[\star ]$ with parameters $c_t,p_t,d_t$. We choose $d=d_t$.
Since each $\Sigma_2$-machine occurs infinitely often in the
family $(S_i)$, we may assume that $j$ is sufficiently large such that $c_t\cdot\gamma\leq c_{j,d}$ and $p_t\cdot\gamma\leq p_j=\lceil\frac{j}{2}\rceil$. 
Since $d=d_t$, we have $I_{n,d}=I_{n,d_t}$.
Now suppose that $n\geq 2^{c_{j,d}}$ and $c_{j,d}\leq a\leq b\leq\frac{\log (n)}{c_{j,d}}$ are such that $\mbox{time}_{S_j}(n)=\gamma\cdot t(n)>an^b$.
From $c_{j,d}\geq \gamma\cdot c_t$ we conclude that $c_t\leq a\slash\gamma\leq b\leq\frac{\log (n)}{c_t}$. Thus since $t$ satisfies Property $[\star ]$ with parameters
$c_t,p_t,d_t$, there exist non-$h$-powers 
\[m_1,\ldots , m_{\lceil\frac{\log\log n}{c_t}\rceil}\in I_{n,d_t}=I_{n,d}\]
with $t(m_l)>(a\slash\gamma -p_t)m_l^{b-p_t},l=1,\ldots \lceil\frac{\log\log n}{c_t}\rceil$. Since $\mbox{time}_{S_j}=\gamma\cdot t$, this yields
$\mbox{time}_{S_j}(m_l)>(a-\gamma\cdot p_t)m_l^{b-p_t}$ for those $l$. Since $p_j\geq \gamma\cdot p_t$, we have
$a -\gamma p_t\geq a-p_j$ and $b-p_t\geq b-p_j$. Finally, from $c_{j,d}\geq \gamma\cdot c_t\geq c_t$ it follows that
$\lceil\frac{\log\log n}{c_t}\rceil\geq \lceil\frac{\log\log n}{c_{j,d}}\rceil$. Thus we have shown that the function $\mbox{time}_{S_j}(n)$
satisfies Property $[\star ]$ with parameters $c_{j,d},p_j,d$. This concludes the proof of the lemma. 
\end{proof}
\begin{lemma}\emph{(Simulation)}\\
For every function $t(n)\geq n\cdot\log^*(n)$ which is deterministically computable in time $t(n)^{1-\epsilon}$ for some fixed $\epsilon >0$
and has \emph{Property $[\star ]$},
$\mbox{DTIM}\tilde{\mbox{E}}(t\log^*t)\subseteq\tilde{\Sigma}_2(t)$.
\end{lemma}
\begin{proof}
The proof consists of using the simulation result from \cite{G96} and the preceeding lemma. This gives
\[\mbox{DTIM}\tilde{\mbox{E}}(t\log^*(t))\subseteq \mbox{DTIME}(t\log^*(t))\subseteq \Sigma_2(t)=\tilde{\Sigma}_2(t).\]
\end{proof}
\begin{theorem}\label{separation_theorem}
For every function $t(n)\geq n\cdot\log^*(n)$ which is deterministically computable in time $t(n)^{1-\epsilon}$ and has \emph{Property $[\star ]$},
$\mbox{DTIM}\tilde{\mbox{E}}(t)\subsetneq\tilde{\Sigma}_2(t)$.
\end{theorem}
\begin{proof}
Suppose $\mbox{DTIM}\tilde{\mbox{E}}(t)=\tilde{\Sigma}_2(t)$. Using Lemma \ref{eqlemma}, we have $\Sigma_2(t)=\tilde{\Sigma}_2(t)$. On the other hand, 
we have $\mbox{DTIM}\tilde{\mbox{E}}(t)\subseteq \mbox{DTIME}(t)\subsetneq\Sigma_2(t)$, where the last strict inclusion in this chain holds due to 
Gupta's result \cite{G96}. 
\end{proof}
\begin{corollary}
$P\neq \Sigma_2^p$, and therefore also $P\neq \Sigma_1^p$.
\end{corollary}
\begin{proof}
Suppose $P=\Sigma_2^p$. Then we obtain $P=\tilde{P}=\mbox{DTIM}\tilde{\mbox{E}}(F)=\tilde{\Sigma}_2(F)=\tilde{\Sigma}_2^p=\Sigma_2^p$.
The function $F$ is deterministically computable in time $F(n)^C$ and satisfies $F(n)\geq n\cdot\log^*(n)$. 
Thus, the function $t(n):=F(n)^{C^2}$ is deterministically computable in time $t(n)^{1-\epsilon}$ for 
some $\epsilon >0$. Furthermore, according to Lemma \ref{F_propertystar_lemma}, 
$t(n)$ satisfies Property $[\star ]$. On the one hand, Lemma \ref{padding_lemma} yields that 
$\mbox{DTIM}\tilde{\mbox{E}}(t)=\tilde{\Sigma}_2(t)$. On the other hand, Theorem \ref{separation_theorem} yields 
$\mbox{DTIM}\tilde{\mbox{E}}(t)\subsetneq \tilde{\Sigma}_2(t)$, a contradiction. Therefore, $P=\Sigma_2^p$ does not hold. 
The second statement from the Corollary holds since the polynomial hierarchy is well known to provide downward separation.
\end{proof}
{\bf Discussion.} The method presented in this paper can also be used to consider the case of $\Sigma_1^p$ versus $\mbox{AP}$, where
$\mbox{AP}$ is the class of all problems solvable in alternating polynomial time, or equivalently in deterministic polynomial space.
The situation is now slightly different, and the approach needs to be adjusted accordingly. We assume $\Sigma_1^p=\mbox{AP}$
and let $(S_i)$ be a standard enumeration of alternating machines. Now the associated problem $\mbox{Check}$ - given a machine index $i$,
an input length $n$ and two integers $a,b$, is $\mbox{time}_{S_i}(n)>an^b$ ? - is only computable in nondeterministic intersect co-nondeterministic
time $c(ian^b)^c$ for some constant $c$. We proceed in the same way as before and first construct a subfamily of alternating machines
$\tilde{S}_{i,d}$  whose running time functions satisfy Property $[\star ]$. However, in order to obtain an analogue of Gupta's result \cite{G96},
we will now have to construct these machines in such a way that for every nondeterministic machine, the associated machines $\tilde{S}_{i,d}$ only make
a bounded number of nondeterministic computation steps. It turns out that the squareroot of the running time is an appropriate bound.
Then we construct the union function $F$ in the same way as before. This function is now computable in nondeterministic intersect co-nondeterministic time
$F^C$, for some constant $C$.
We obtain the following Separation Theorem: For every function $t\colon {\mathbb N}\to {\mathbb N}$
which is computable in $\mbox{NTIME}(t^{1-\epsilon})\cap\mbox{coNTIME}(t^{1-\epsilon})$ for some constant $\epsilon >0$ and satisfies Property $[\star ]$,
$\tilde{\Sigma}_1(t,\sqrt{t})\subsetneq\mbox{ATIM}\tilde{\mbox{E}}(t)$, where now $\tilde{\Sigma}_1(t,\sqrt{t})$ denotes the class of all decision problems
solvable nondeterminsitically by an $O(t)$ time bounded machine $\tilde{S}_{i,d}$ with the number of guesses being bounded by the squareroot of the running time.  
We obtain the desired contradiction, and thus the assumption $\Sigma_1^p=\mbox{AP}$ cannot 
hold. The details will be given in a subsequent paper.\\
$\mbox{ }$\\
$\mbox{ }$\\
{\bf Acknowledgement.} I would like to thank Norbert Blum for carefully reading preliminary versions of the paper, for helpful remarks and discussions,
for his guidance and patience and for being 
my mentor.


\end{document}